\documentclass[reqno,11pt]{amsart}
\usepackage[utf8]{inputenc}
\usepackage{graphicx}
\usepackage{slashed}
\usepackage{amscd}
\usepackage{amssymb}
\usepackage{esint}
\usepackage{enumitem}
\usepackage{subfig}
\usepackage{comment}
\usepackage{dsfont}
\usepackage[mathscr]{eucal}
\usepackage{mathtools}
\usepackage{pst-all}

\numberwithin{equation}{section}
\allowdisplaybreaks[1]
\definecolor{labelkey}{gray}{.65}

\title[The Fermionic Entanglement Entropy of a Schwarzschild Black Hole]{The Fermionic Entanglement Entropy of the Vacuum State of a Schwarzschild Black Hole Horizon}
\author[F.\ Finster, M.\ Lottner]{Felix Finster and Magdalena Lottner \\ \\ February 2023 / March 2024}
\address{Fakult\"at f\"ur Mathematik \\ Universit\"at Regensburg \\ D-93040 Regensburg \\ Germany}
\email{finster@ur.de, magdalena.lottner@ur.de}

\newtheorem{Def}{Definition}[section]
\newtheorem{Thm}[Def]{Theorem}
\newtheorem{Proposition}[Def]{Proposition}
\newtheorem{Lemma}[Def]{Lemma}
\newtheorem{Remark}[Def]{Remark}
\newtheorem{Corollary}[Def]{Corollary}
\newtheorem{Example}[Def]{Example}

\newtheorem{Notation}[Def]{Notation}
\newtheorem{cond}[Def]{Condition}
\newcommand{\Thanks}{\vspace*{.5em} \noindent \thanks}
\newcommand{\beq}{\begin{equation}}
	\newcommand{\eeq}{\end{equation}}
\newcommand{\Proof}{\begin{proof}}
	\newcommand{\QED}{\end{proof} \noindent}
\newcommand{\QEDrem}{\ \hfill $\Diamond$}

\newcommand{\la}{\langle}
\newcommand{\ra}{\rangle}

\newcommand{\Sl}{\mbox{$\prec \!\!$ \nolinebreak}}
\newcommand{\Sr}{\mbox{\nolinebreak $\succ$}}

\newcommand{\C}{\mathbb{C}}
\newcommand{\R}{\mathbb{R}}
\newcommand{\1}{\mbox{\rm 1 \hspace{-1.05 em} 1}}
\newcommand{\Z}{\mathbb{Z}}
\newcommand{\N}{\mathbb{N}}

\DeclareMathOperator{\tr}{tr}

\renewcommand{\L}{{\mathcal{L}}}

\newcommand{\Dir}{{\mathcal{D}}}
\renewcommand{\H}{\mathscr{H}}

\newcommand{\Lin}{\text{\rm{L}}}

\newcommand{\D}{\mathscr{D}}

\DeclareMathOperator{\supp}{supp}

\newcommand{\scrM}{\mycal M}
\newcommand{\scrN}{\mycal N}

\newcommand{\bitem}{\begin{itemize}[leftmargin=2.5em]}
\newcommand{\eitem}{\end{itemize}}

\newcommand{\loc}{\text{\rm{loc}}}

\newcommand{\Opa}{\mathrm{Op}_\alpha}

\newcommand{\fp}{f_{0,1}^+}
\newcommand{\fm}{f_{0,1}^-}

\newcommand{\Four}{\mathcal{F}}
\newcommand{\SN}{\mathfrak{S}}
\newcommand{\id}{\mathds{1}}
\newcommand{\bl}{{\,\vrule depth3pt height9pt}{\vrule depth3pt height9pt}
	{\vrule depth3pt height9pt}{\vrule depth3pt height9pt}\,}
\newcommand{\Afull}{\CA_{\mathrm{full}}}
\newcommand{\afull}{\mathfrak{a}_{\mathrm{full}}}
\newcommand{\Da}{ \Delta \mathfrak{a}}

\newcommand{\afrak}{\mathfrak{a}}
\newcommand{\opa}{\mathrm{op}_\alpha}
\newcommand{\Forany}{\qquad \text{for any }}
\newcommand{\CA}{\mathcal A}
\newcommand {\bxi}{\boldsymbol\xi}
\newcommand {\bmu}{\boldsymbol\mu}

\newcommand {\bx}{\boldsymbol x}
\newcommand {\by}{\boldsymbol y}
\newcommand {\bz}{\boldsymbol z}
\newcommand {\bc}{\boldsymbol c}
\newcommand{\CB}{\mathcal B}

\newcommand{\Op}{\mathrm{Op}}
\newcommand{\Fock}{{\mathscr{F}}}

\DeclareFontFamily{OT1}{rsfso}{}
\DeclareFontShape{OT1}{rsfso}{m}{n}{ <-7> rsfso5 <7-10> rsfso7 <10-> rsfso10}{}
\DeclareMathAlphabet{\mycal}{OT1}{rsfso}{m}{n}

\newcommand\mpar[1]{}
\newcommand\Magdalena[1]{}
\newcommand\Felix[1]{}

\begin{document}

\begin{abstract}
We define and analyze the fermionic entanglement
entropy of a Schwarz\-schild black hole horizon for the regularized vacuum state of an observer at infinity.
Using separation of variables and an integral representation of the Dirac propagator,
the entanglement entropy is computed to be a prefactor times the number of occupied angular momentum modes
on the event horizon.
\end{abstract}

\maketitle

\vspace*{-0.6cm}

\tableofcontents

\section{Introduction}
Black hole thermodynamics is an exciting topic of current research in both physics and mathematics.
It was initiated by the discovery of Bekenstein and Hawking that black holes behave thermally if one interprets
surface gravity as temperature and the area of the event horizon as entropy~\cite{Bekenstein:1975tw, Hawking:1976de}.
The analogy to the second law of thermodynamics suggests that the area of the black hole horizon
should only increase in time. However, this is in contradiction with the discovery of
Hawking radiation and the resulting ``evaporation'' of a black
hole~\cite{Hawking:1974rv, hawking}. This so-called {\em{information paradox}}~\cite{Hawking:1976ra}
inspired the holographic principle \cite{Susskind:1993ki, Hooft:1999bw} and the current program of attempting to understand the structure of spacetime via information theory, entropies and gauge/gravity dualities.

The present work is concerned with the entropy of a black hole.
Generally speaking, {\em{entropy}} is a measure for the disorder of a physical system.
There are various notions of entropy, like the entropy in classical statistical mechanics
as introduced by Boltzmann and Gibbs, the Shannon and R{\'e}nyi entropies in information theory
or the von Neumann entropy for quantum systems.
Here we focus on the {\em{entanglement entropy}}, 
which quantifies the quantum entanglement of a spatial region with its surrounding
(for the general physical and mathematical context see for example~\cite{Rangamani:2016dms},
\cite{hollands-sanders}). The entanglement entropy of the {\em{event horizon}}
tells us about the quantum entanglement between the interior and exterior regions of the black hole.
For technical simplicity, we here restrict attention to the simplest mathematical model of a black hole:
a {\em{Schwarzschild black hole}} of mass~$M$ (more general black holes will be discussed
in the outlook section after~\eqref{endsummary}).
We consider the {\em{R{\'e}nyi entropy}} functional for a general R{\'e}nyi parameter~$\kappa>0$.
The case~$\kappa=1$ gives the {\em{von Neumann entropy}} functional.
We compute the corresponding entanglement entropies for the {\em{quasi-free fermionic state}}
describing the {\em{vacuum}} of an observer at infinity.
More precisely, we consider the quasi-free fermionic Hadamard state which is obtained by frequency splitting for the observer in a rest frame in Schwarzschild coordinates, with an {\em{ultraviolet regularization}} on a length
scale~$\varepsilon$.
In a more physical language, we consider a free Fermi gas formed of non-interacting one-particle Dirac states.
Based on formulas derived in~\cite{helling-leschke-spitzer, leschke-sobolev-spitzer}
(for more details see the preliminaries in Section~\ref{secentquasi}), the entanglement entropy
can be expressed in terms of the reduced one-particle density operator.
We choose this one-particle density operator as the regularized projection operator to
all negative-frequency solutions of the Dirac equation in the exterior Schwarzschild geometry
(where ``frequency'' refers to the Schwarzschild time of an observer at rest).
Making use of the integral representation of the Dirac propagator in~\cite{tkerr}
and employing techniques developed in~\cite{leschke-sobolev-spitzer, widom1,
sobolev-schatten, sobolev-functions, sobolev-pseudo},
it becomes possible to compute the entanglement entropy on the black hole horizon explicitly.
We find that, up to a prefactor which depends on~$\varepsilon M$,
this entanglement entropy is given by the number of occupied angular momentum modes,
making it possible to reduce the computation of the entanglement entropy to counting the number of occupied
one-particle states.
A similar result is obtained for the {\em{R{\'e}nyi entropies}} with R{\'e}nyi index~$\kappa > \frac{2}{3}$.

We now outline our setting and the main result. The quasi-free regularized Dirac vacuum state
can be described completely by the corresponding reduced one-particle density operator
(for details see Section~\ref{secentquasi}). We choose this operator 
as the regularized projection operator to the negative frequency solutions of the Dirac equation by~$\Pi_-^\varepsilon$ (for details see Sections~\ref{secregproj} and~\ref{secregvac}).
Given a parameter~$\kappa>0$ (the {\em{R{\'e}nyi index}}), we introduce the {\em{R{\'e}nyi entropy
function}}~$\eta_\kappa$ as follows. If~$t\notin (0,1)$ then we set~$\eta_\kappa(t) = 0$.
For~$t\in (0, 1)$ we define
\begin{align}
	\eta_\kappa(t)= & \ 
			\displaystyle \frac{1}{1-\kappa}\log \big( t^\kappa + (1-t)^\kappa \big)
			&& \text{for } \kappa\neq 1 \label{Def eta_gamma} \\[0.2cm]
\eta_1(t):=\lim\limits_{\kappa' \rightarrow 1} \eta_{\kappa'}(t)
	= &\ -t\log t - (1-t)\log (1-t) && \text{for }  \kappa = 1 \label{Def eta}
\end{align}
(the last limit can be computed directly with l'Hospital's rule).
Note that the function~$\eta_\kappa$ is continuous and smooth except at~$t=0$ and~$t=1$,
as shown in Figure~\ref{fig: Plot eta} for various values of~$\kappa$.
\begin{figure}[t]
\centering
\includegraphics[width=0.55\textwidth]{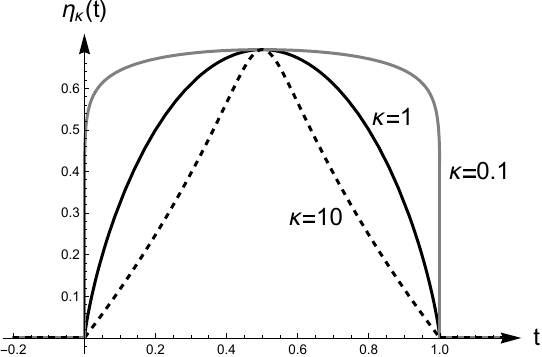}
\caption{Plot of the function~$\eta_\kappa$ for~$\kappa=0.1$, $1$ and~$10$.}
\label{fig: Plot eta}
\end{figure}%
Note that~$\eta_1$ is the familiar {\em{von Neumann entropy function}}. 
Next, we consider the {\em{entropic difference operator}} corresponding to the subset~$\Lambda$
as introduced~\cite[Section~3]{leschke-sobolev-spitzer2}
(for more details and references see the preliminaries in Section~\ref{secentquasi})
\beq
\label{EntropicDiff}
\Delta_\kappa(\Pi_-^\varepsilon, \Lambda):=\eta_\kappa \big( \chi_{\Lambda} \:\Pi_-^\varepsilon\: \chi_{\Lambda} \big) -\chi_{\Lambda} \,\eta_\kappa(\Pi_-^\varepsilon)\,\chi_{\Lambda} \:,
\eeq
In order to obtain the entropy of the event horizon, we choose~$\Lambda$ as an annular region around the
event horizon. As the radial coordinate we choose Regge-Wheeler coordinate~$u \in \R$, in which
the event horizon is located at~$u\rightarrow -\infty$ (for details see~\eqref{RW} in the preliminaries).
We then parametrize~$\Lambda$ by
\begin{align}
\label{def:TildeK}
\Lambda &:= {\mathcal{K}} \times S^2 \qquad \text{with} \qquad
\mathcal{K}:=(u_0-\rho, u_0)\\
\notag
&\hspace{.1cm}\equiv  \left\{ \left.
\begin{pmatrix}
	u \sin \vartheta \cos \varphi\\
	u \sin \vartheta \sin \phi\\
	u \cos \vartheta
\end{pmatrix} \:\right| \: u_0-\rho<u< u_0,\:\: 0<\vartheta<\pi, \: \: 0< \varphi< 2\pi \right\}
\end{align}
(see also Figure~\ref{Fig:Lambda2} on page~\pageref{Fig:Lambda2}).
The {\em{fermionic entanglement entropy}} is obtained as the trace of the entropic difference operator~\eqref{EntropicDiff} in the limit when~$\Lambda$ moves toward the event horizon, i.e.\
\beq \label{RenyEnt}
S_\kappa(\Pi_-^\varepsilon, \Lambda):= \lim_{u_0 \rightarrow -\infty}
\tr \big( \Delta_\kappa(\Pi_-^\varepsilon, \Lambda) \big) \:.
\eeq
We shall prove that, to leading order in the regularization length~$\varepsilon$,
this trace is independent of~$\rho$.
It turns out that we get equal contributions from the two boundaries at~$u_0-\rho$ and~$u_0$
as~$u_0 \rightarrow -\infty$. Therefore, the fermionic entanglement entropy is given by one half this trace. 

Before stating our main result, we note that
the trace of the entropic difference operator can be decomposed into a sum over all occupied angular
momentum modes. For ease in presentation, we begin with one angular momentum mode
parametrized by~$(k,n)$ with~$k \in \Z+1/2$ and~$n \in \N$ (for details on the separation of variables
of the Dirac equation see the preliminaries in Section~\ref{The Schwarzschild Propagator}).
Then the trace in~\eqref{RenyEnt} becomes
\[ \tr \Delta_\kappa\big((\Pi_-^\varepsilon)_{kn}, {\mathcal{K}}\big) \:, \]
where~$(\Pi_-^\varepsilon)_{kn}$ is the operator~$\Pi^i_\varepsilon$ restricted to
the angular mode~$(k,n)$. This operator depends only on the radial variable; this is why
the characteristic function~$\chi_{\Lambda}$ has been replaced by~$\chi_{\mathcal{K}}$.
We define the {\em{mode-wise R{\'e}nyi entropy}} of the black hole as
\begin{align} \label{def modewise entropy}
S_{\kappa, kn}^{\mathrm{BH}} := \frac{1}{2}\:\lim\limits_{\rho \rightarrow \infty}\lim\limits_{\varepsilon \searrow 0} \frac{1}{f(\varepsilon)}\lim\limits_{u_0\rightarrow -\infty}
\tr \Delta_\kappa\big((\Pi_-^\varepsilon)_{kn}, \mathcal{K}\big)\:,
\end{align}
where~$f(\varepsilon)$ is a function describing the highest order of divergence in~$\varepsilon$ (we will later see that here~$f(\varepsilon)=\log (M / \varepsilon)$ with~$M$ the black hole mass).
Our main result shows that~$S_{\kappa, kn}^{\mathrm{BH}}$ has the same numerical value for each angular mode:
\begin{Thm}\label{Main Res.}
Let~$\kappa > \frac{2}{3}$ and let~$n \in \Z$ and~$k \in \Z +1/2$ arbitrary then
\begin{flalign} \label{limrenyi}
	\lim\limits_{\varepsilon \searrow 0} \;\lim\limits_{u_0 \rightarrow -\infty}  \frac{1}{\log (\ell / \varepsilon)} \tr \Delta_\kappa\big((\Pi_-^\varepsilon)_{kn}, \mathcal{K}\big) = \frac{1}{12} \frac{\kappa +1}{\kappa} \:,
\end{flalign}
where~$\ell$ is a reference length. Due to the form of the R{\'e}nyi entropy functions, the right hand side is always positive.
For the entanglement entropy, i.e.\ $\kappa =1$, we obtain in particular
\beq \label{limcount}
\lim\limits_{\varepsilon \searrow 0} \;\lim\limits_{u_0 \rightarrow -\infty}  \frac{1}{\log (\ell / \varepsilon)} \tr  \Delta_1 \big((\Pi_-^\varepsilon)_{kn}, \mathcal{K}\big)	= \frac{1}{6}\:.
\eeq
\end{Thm}	 \noindent
We note for clarity that the only purpose of the reference length~$\ell$ is to make the argument of the
logarithm dimensionless. The choice of~$\ell$ is irrelevant, because writing~$\log(\ell/\varepsilon) = \log \ell - \varepsilon$, the term~$\log \ell$ is sub-leading. Thus the main statement of~\eqref{limrenyi} and~\eqref{limcount}
is that, in the limit~$\varepsilon \searrow 0$, the traces are logarithmically divergent, and we determine
the corresponding proportionality factor. Since the mass of the Dirac particles will be irrelevant, it is
natural to associate~$\ell$ with the only other parameter with dimension of length: the mass~$M$ of the
black hole. Therefore, in what follows we will always replace the logarithm in the above formulas
by~$\log(M / \varepsilon)$ (for more on units see the last paragraph of the introduction).

In simple terms, the above result shows that each occupied angular momentum mode gives the
same contribution to the (R{\'e}nyi) entanglement entropy.
This result can be understood immediately from the infinite red shift effect at the event horizon.
Indeed, asymptotically near the event horizon, a Dirac wave behaves
like a massless particle without angular momentum (as will be made precise
in Lemma~\ref{X_at_hor} below), suggesting that also the entanglement entropy should
be the same for each angular mode. Proving this result, however, makes it necessary to
estimate different error terms, which constitutes the technical core of this work.

The (R{\'e}nyi) entanglement entropy of the black hole can be written formally
as the sum of all angular momentum modes,
\begin{align} \label{entropy BH}
S_{\kappa}^\mathrm{BH} = \sum_{k,n} S_{\kappa, kn}^\mathrm{BH} \:.
\end{align}
Since each angular mode gives the same non-zero contribution~\eqref{limcount},
the sum in~\eqref{entropy BH} clearly diverges if an infinite number of angular modes
are occupied. This leads us to regularize the vacuum state by also occupying only a
finite number of angular momentum modes (for details see~\eqref{occupy} in Section~\ref{secregvac}).
After this has been done, the sum in~\eqref{entropy BH} becomes finite.
Then, applying Theorem~\ref{Main Res.} to each angular mode,
it becomes possible to compute the entanglement entropy of the horizon
simply by counting the number of occupied angular momentum modes.
This is reminiscent of the counting of states in string theory~\cite{strominger-vafa}
and loop quantum gravity~\cite{ashtekar-baez}. In order to push the analogy further,
assuming a minimal area~$\varepsilon^2$ on the horizon and keeping in mind that the area of the
event horizon scales like~$M^2$ (as is obvious already from dimensional considerations),
the number of occupied angular modes should scale like~$M^2/\varepsilon^2$. In this way, we find that the entanglement entropy is indeed proportional to the area of the black hole.
The factor~$\log(M / \varepsilon)$
in the above theorem is usually referred to as an {\em{enhanced}} area law.
Such an enhanced area law typically occurs if the considered fields are massless, in which case
long-range effects give rise to an additional logarithmic divergence.
This fits to our physical situation because, as mentioned above,
the Dirac wave behaves near the event horizon like as massless field
due to the red shift effect at the event horizon. 
We also note for clarity that we make essential use of the fact that our vacuum state
involves a hard cutoff between the occupied one-particle states of negative frequency
and the non-occupied states of positive frequency. It is not clear to us if considering instead a smooth
cutoff function would still lead to an enhanced area law.

The article is structured as follows. 
Section~\ref{SecPreliminiaries} provides the necessary preliminaries on
entanglement entropy, the Dirac equation, the Dirac propagator in the Schwarzschild geometry
and some technical tools involving Schatten classes and pseudo-differential operators.
In Section~\ref{secregproj} the regularized projection operator on the negative-frequency solutions
of the Dirac equation is defined and decomposed into angular momentum modes.
For each angular momentum mode, the resulting functional calculus is formulated and the corresponding operator is rewritten in the
language of pseudo-differential operators. Moreover, the symbol will be further simplified at the horizon.
After these preparations, the core of this works begins in Section~\ref{secregvac}, where the
regularized fermionic vacuum and the corresponding
(R{\'e}nyi) entanglement entropy of the event horizon is defined. In Section~\ref{Sec trace of the limiting operator}
the entropy of a simplified limiting operator (in the sense that the regularization goes to zero) at the horizon. 
Afterward, we estimate the error caused by using the limiting operator instead of the regularized one
(Section~\ref{Sec Estimating  the Error Terms}). It turns out that this error drops out in the limiting process.
Subsequently, we complete the proof of the main result (Theorem~\ref{Main Res.}) by combining the results from the previous sections (Section~\ref{SecProofofMain}). We conclude with a brief summary and a discussion of open problems (Section~\ref{Sec Conclusion}).
The appendices contain additional material and give some background information. \\
%

\noindent
{\bf{Units and notational conventions.}} 
We work throughout in natural units~$\hbar = c = 1$. Then the only remaining unit is that of
a length (measured for examples in meters). It is most convenient to work with dimensionless
quantities. This can be achieved by choosing an arbitrary reference length~$\ell$ and multiplying
all dimensional quantities by suitable powers of~$\ell$. For example, we work with the
\begin{equation} \label{dimensionless}
	\text{dimensionless quantities} \qquad
	m \ell\:,\quad \frac{M}{\ell} \quad \text{and} \quad \frac{\varepsilon}{\ell}
\end{equation}
(where~$m$ is the mass of the Dirac particles, $M$ is the mass of the black hole
times the gravitational constant, and~$\varepsilon$ is the regularization length).
For ease in notation, in what follows we set~$\ell=1$, making it possible to
leave out all powers of~$\ell$. The dimensionality can be recovered by rewriting
all formulas using the dimensionless quantities in~\eqref{dimensionless}. In the Schwarzschild geometry,
it is natural to choose~$\ell$ as the black hole mass~$M$.

We conclude the introduction with some general notational conventions.
For two non-negative numbers (or functions) 
$X$ and~$Y$ depending on some parameters, 
we write~$X\lesssim Y$ (or~$Y\gtrsim X$) if~$X\le C Y$ for
some positive constant~$C$ independent of those parameters.
To avoid confusion we may comment on the nature of 
(implicit) constants in the bounds. 

For any vector space~$V$ we denote
\[
\Lin(V) := \big\{ f: V \rightarrow V \: \big| \: f \text{ bounded and linear} \big\} \:.
\]

Finally, for ease of notation the operator of multiplication by~$f$ is denoted with the same letter, i.e.\ $(f \psi)(x)
:= f(x)\, \psi(x)$.

\section{Preliminaries}
\label{SecPreliminiaries}
\subsection{The Entanglement Entropy of a Quasi-Free Fermionic State} \label{secentquasi}
Given a Hilbert space~$(\H_m, \langle .|. \rangle_m$) (the ``one-particle Hilbert space''),
we let~$(\mathscr{F}, \la .|. \ra_\mathscr{F})$ be the corresponding fermionic Fock space, i.e.\
\[ \mathscr{F} = \bigoplus_{k=0}^\infty \;\underbrace{\H_m \wedge \cdots \wedge \H_m}_{\text{$k$ factors}} \]
(where~$\wedge$ denotes the totally anti-symmetrized tensor product).
We define the {\em{creation operator}}~$\Psi^\dagger$ by
\[ \Psi^\dagger \::\: \H_m \rightarrow \text{\rm{L}}(\Fock) \:,\qquad
\Psi^\dagger(\psi) \big( \psi_1 \wedge \cdots \wedge \psi_p \big) := \psi \wedge \psi_1 \wedge \cdots \wedge \psi_p \:. \]
Its adjoint is the annihilation operator denoted by~$\Psi(\overline{\psi}) := (\Psi^\dagger(\psi))^*$.
These operators satisfy the canonical anti-commutation relations
\[ \label{CAR}
	\big\{ \Psi(\overline{\psi}), \Psi^\dagger(\phi) \big\} = (  \psi | \phi ) \qquad\text{and}\qquad
	\big\{ \Psi(\overline{\psi}), \Psi(\overline{\phi}) \big\} = 0 = \big\{ \Psi^\dagger(\psi), \Psi^\dagger(\phi) \big\} \:. \]
Next, we let~$W$ be a {\em{statistical operator}} on~$\Fock$, i.e.\ a positive semi-definite linear operator of trace one,
\[ W \::\: \Fock \rightarrow \Fock\:,\qquad W \geq 0 \quad \text{and} \quad \tr_\Fock(W)=1 \:. \]
Given an observable~$A$ (i.e.\ a symmetric operator on~$\Fock$), the expectation value of the measurement
is given by
\[ \langle A \rangle := \tr_\Fock\big( A W) \:. \]
The corresponding {\em{quantum state}}~$\Omega$ is the linear functional which to every observable
associates the expectation value, i.e.\
\[ \Omega \::\: A \mapsto \tr_\Fock\big( A W) \:. \]

In this work, we restrict our attention to the subclass of so-called \emph{quasi-free} quantum states,  fully determined by their two-point functions
\[ \Omega_2(\overline{\psi}, \phi) := \Omega\big( \Psi^\dagger(\phi)\,\Psi(\overline{\psi}) \big) \:, \qquad \text{for any } \psi, \phi \in \mathscr{H}_m \:. \]
\begin{Def} \label{defreduced}
The {\bf{reduced one-particle density operator}}~$D$ is the positive linear operator on the Hilbert space~$(\H_m, (.|.)_\scrM)$ defined by
\[ \Omega_2(\overline{\psi}, \phi) = \langle \psi \,|\, D \phi\rangle_m \:, \qquad \text{for any } \psi, \phi \in \mathscr{H}_m \:. \]
\end{Def}

The von Neumann entropy~$S(\Omega)$ of the quasi-free fermionic
state~$\Omega$ can be expressed in terms of the reduced one-particle density operator by
\beq \label{Sred}
S(\Omega) = \tr \eta_1(D) \:,
\eeq
where~$\eta_\kappa$ is the function from~\eqref{Def eta} (for a plot see Figure~\ref{fig: Plot eta}).
This formula appears commonly in the literature
(see for example~\cite[Equation 6.3]{ohya-petz}, \cite{klich, casini-huerta, longo-xu}
and~\cite[eq.~(34)]{helling-leschke-spitzer}).
A detailed derivation is found in~\cite[Appendix~A]{fermientropy}.
Similar to~\eqref{Sred} also other entropies can be expressed in terms of the reduced one-particle density
operator. In particular, the R{\'e}nyi entropy can be written
as~$S_\kappa(\Omega) = \tr \eta_\kappa(D)$ 
This formula is also derived in~\cite[Appendix~A]{fermientropy}.

For the {\em{entanglement entropy}} we need to assume that the Hilbert space~$\H_m$ is formed
of wave functions in spacetime. Restricting them to a Cauchy surface, we obtain functions defined on
three-dimensional space~$\scrN$ (which could be~$\R^3$ or, more generally, a three-dimensional
manifold). Given a spatial subregion~$\Lambda \subset \scrN$, we define the (R{\'e}nyi) entanglement entropy by
\beq \label{entropygen}
S_\kappa(D, \Lambda) := \tr \big( 
\eta_\kappa \big( \chi_{\Lambda} \:D\: \chi_{\Lambda} \big) -\chi_{\Lambda} \,\eta_\kappa(D)\,\chi_{\Lambda}
\big) \:.
\eeq
More details in the case~$\kappa=1$ is found in~\cite[Section~3]{leschke-sobolev-spitzer2}.

\subsection{The Dirac Equation in Globally Hyperbolic Spacetimes} \label{secdirglobhyp}
Since we are ultimately interested in Schwarzschild space time, the abstract setting for the Dirac equation is given as follows (for more details see for example~\cite{finite}).
Our starting point is a four dimensional, smooth, globally hyperbolic Lorentzian spin manifold~$(\scrM, g)$, with metric~$g$ of signature~$(+ ,-, -, -)$. 
We denote the corresponding spinor bundle by~$S\scrM$. Its fibres~$S_x\scrM$ are endowed
with an inner product~$\Sl .|. \Sr_x$ of signature~$(2,2)$, referred to as the spin inner product.
Moreover, the mapping
\[ \gamma \::\: T_x\scrM \rightarrow \Lin(S_x\scrM) \:, \quad u \mapsto \sum\nolimits_{j=0}^3 \gamma^j u_j\:,  \]
where the~$\gamma^j$ are the Dirac matrices defined via the anti-commutation relations
\[ \gamma(u) \,\gamma(v) + \gamma(v) \,\gamma(u) = 2 \, g(u,v)\,\1_{S_x(\scrM)} \:, \]
provides the structure of a Clifford multiplication.

Smooth sections in the spinor bundle are denoted by~$C^\infty(\scrM, S\scrM)$.
Likewise, $C^\infty_0(\scrM, S\scrM)$ are the smooth sections with compact support.
We also refer to sections in the spinor bundle as {\em{wave functions}}.
The Dirac operator~$\Dir$ takes the form
\[ \Dir := i \gamma^j \nabla_j \::\: C^\infty(\scrM, S\scrM) \rightarrow C^\infty(\scrM, S\scrM)\:,\]
where~$\nabla$ denotes the connections on the tangent bundle and the spinor bundle.
Then the Dirac equation with parameter~$m$ (in the physical context corresponding to the particle mass) reads
\beq \label{Dirac}
(\Dir - m) \,\psi = 0 \:.
\eeq

Due to global hyperbolicity, our spacetime admits a foliation by Cauchy surfaces~$\scrM = (\scrN_t)_{t \in \R}$. Smooth initial data on any such Cauchy surface yield a unique global solution of the Dirac equation.
Our main focus lies on smooth solutions with spatially compact support, denoted by~$C^\infty_{\rm{sc}}(\scrM, S\scrM)$. The solutions in this class are endowed with the scalar product
\beq \label{print}
(\psi | \phi)_m = \int_\scrN \Sl \psi \,|\, \nu^j \gamma_j\, \phi \Sr_x\: d\mu_\scrN(x) \:,
\eeq
where~$\scrN$ is a Cauchy surface~$\scrN$ with future-directed normal~$\nu$ and~$d\mu_\scrN$ denotes the measure on~$\scrN$ induced by the metric~$g$ 
(compared to the conventions in~\cite{finite}, we here preferred to leave out a factor of~$2 \pi$).
This scalar product is independent of the choice of~$\scrN$ (for details see~\cite[Section~2]{finite}).
Finally we define the Hilbert space~$(\H_m, (.|.)_m)$ by completion,
\[ \H_m:= \overline{ C^\infty_{\rm{sc}}(\scrM, S\scrM)}^{  (.|.)_m} \:.\]

\subsection{The Dirac Propagator in the Schwarzschild Geometry}
	\label{The Schwarzschild Propagator}
\subsubsection{The Integral Representation of the Propagator}
We recall the form of the Dirac equation in the Schwarzschild geometry and its separation,
closely following the presentation in~\cite{tkerr} and \cite{sigbh}.
Given a parameter~$M>0$ (the black hole mass), the exterior Schwarzschild metric reads
\[
ds^2 = g_{jk}\:dx^j \,dx^k = \frac{\Delta(r)}{r^2} \: dt^2
- \frac{r^2}{\Delta(r)}\:dr^2 - r^2\: d \vartheta^2 - r^2\: \sin^2 \vartheta\: d\varphi^2\:, \]
where
\[ \Delta(r) := r^2 - 2M r \:. \]
Here the coordinates~$(t, r, \vartheta, \varphi)$ takes values in the intervals
\[ -\infty<t<\infty,\qquad r_1<r<\infty,\qquad 0<\vartheta<\pi,\qquad 0<\varphi<2\pi \:, \]
where~$r_1:=2M$ is the event horizon.
It is most convenient to transform the radial coordinate to the 
so called \textit{Regge-Wheeler-coordinate}~$u \in \R$ defined by
\beq \label{RW}
u(r) = r + 2M \log(r-2M)\:, \qquad \text{so that} \qquad \frac{du}{dr} = \frac{r^2}{\Delta(r)}\:.
\eeq
In this coordinate, the event horizon is located at~$u \rightarrow - \infty$, whereas~$u\rightarrow \infty$ corresponds to spatial infinity, i.e.\ $r \rightarrow \infty$.

In this geometry, the Dirac operator takes the form (see also~\cite[Section~2.2]{sigbh}):
\begin{flalign*}
	\Dir &= \begin{pmatrix} 0 & 0 & \alpha_+ & \beta_+ \\
		0 & 0 & \beta_- & \alpha_- \\
		\alpha_- & -\beta_+ & 0 & 0 \\
		-\beta_- & \alpha_+ & 0 & 0 \end{pmatrix} \qquad \text{with} \\
	\beta_\pm &= \frac{i}{r} \left( \frac{\partial}{\partial
		\vartheta} + \frac{\cot \vartheta}{2} \right) \pm
	\frac{1}{r \sin \vartheta} \:\frac{\partial}{\partial \varphi} \qquad \text{and} \\
	\alpha_\pm &= -\frac{ir} {\sqrt{\Delta(r)}} \:\frac{\partial}{\partial t}
	\pm \frac{\sqrt{\Delta(r)}}{r} \left( i \frac{\partial}{\partial r}
	\:+\: i \:\frac{r-M}{2 \Delta(r)} \:+\: \frac{i}{2r} \right) .
\end{flalign*}
Then the Dirac equation can be separated with the ansatz
\begin{flalign*}
	\psi^{kn}(t,u,\varphi,\vartheta) = e^{-ik\varphi}\:\frac{1}{\Delta(r)^{1/4} \sqrt{r}}\:
	\begin{pmatrix}
		X_-^{kn}(t,u)Y^{kn}_-(\vartheta)\\
		X_+^{kn}(t,u)Y^{kn}_+(\vartheta)\\
		X_+^{kn}(t,u)Y^{kn}_-(\vartheta)\\
		X_-^{kn}(t,u)Y^{kn}_+(\vartheta)
	\end{pmatrix}
\end{flalign*}
with~$k \in \Z+1/2$, $n \in \N$ and~$\omega \in \R$. The angular functions~$Y^{kn}_\pm$
can be expressed in terms of spin-weighted spherical harmonics and form an orthonormal basis of~$L^2\big(\big((-1,1),\:d\vartheta \cos \vartheta\big), \C^2\big)$ (see~\cite[Section 2.4]{sigbh} with additional reference to~\cite{goldberg}). The radial functions~$X^{k n}_\pm$
satisfy a system of partial differential equations
\begin{flalign}
	\label{Radial Dirac Eq}
	\begin{pmatrix}
		\sqrt{\Delta(r)}\: \mathcal{D}_+ & imr-\lambda\\
		-imr-\lambda & \sqrt{\Delta(r)}\: \mathcal{D}_-
	\end{pmatrix}
	\begin{pmatrix}
		X_+^{kn} \\
		X_-^{kn}
	\end{pmatrix}
	=
	0\:,
\end{flalign}
where~$m$ denotes the particle mass and
\begin{flalign*}
	\mathcal{D}_\pm = \frac{\partial}{\partial r} \mp \frac{r^2}{\Delta(r)}\:\frac{\partial }{\partial t} \:,
\end{flalign*}
for details see \cite[Section~2]{tkerr}. 
Moreover, employing the ansatz
\[ X^{kn}_\pm(t,u)= e^{-i\omega t} \:X^{kn\omega}_\pm(u) \:, \]
equation~\eqref{Radial Dirac Eq} goes over to a system of ordinary differential equations, which admits two two-component fundamental solutions labeled by~$a=1,2$. 
We denote the resulting Dirac solution by~$X^{kn\omega}_a=(X^{kn\omega}_{a,+}, X^{kn\omega}_{a,- })$
In the case~$|\omega|<m$, these solutions behave exponentially near infinity.
We always choose the fundamental solution for
\beq \label{aeone}
\text{$a=1$ as the fundamental solution which
{\em{decays}} at infinity.}
\eeq
For more details on the choice of the fundamental solutions see Section~\ref{rnProp at the Horizon} below.

In what follows we will often use the following notation for two-component functions 
\[  A:=\begin{pmatrix}
	A_+ \\ A_-
\end{pmatrix}\:. \]
The norm in~$\C^2$ will be denoted by~$| \,.\, |$, the canonical inner product on~$L^2(\R,\C^2)$ by~$\la.|.\ra$ and the corresponding norm by~$\|.\|$.

As implied by \cite[Theorem~3.6]{tkerr}, one can then find the following formula for the mode-wise propagator:
\begin{Thm} \label{thmintrep}
Given initial radial data~$X_0 \in C^\infty_0(\R, \C^2)$ at time~$t=0$, the corresponding solution~$X \in C^\infty(\R^2,\C^2)$ of the radial Dirac equation~\eqref{Radial Dirac Eq} can be written as
\beq \label{SchwarzschildProp}
X(t, u) = 
\frac{1}{\pi} \int_{-\infty}^\infty d\omega \:e^{-i\omega t} \sum_{a,b =1}^{2} t_{ab}^{kn\omega} \,X_a^{kn\omega} (u) \:\la X_b^{kn\omega} |  X_0 \ra \:,
\eeq
for any~$t,u \in \R$. The~$X_a^{kn\omega} (x)$ are the fundamental solutions mentioned before.
Here the coefficients~$t_{ab}^{kn\omega}$ satisfy the relations
\[ \overline{t^{kn\omega}_{ab}} = t^{kn\omega}_{ba} \]
and
\beq
\label{tab}
\left\{ \begin{array}{cl}
t^{kn\omega}_{ab} = \delta_{a,1} \:\delta_{b,1} & \qquad \text{if~$|\omega| \leq m$} \\[0.3em]
\displaystyle t^{kn\omega}_{11} = t_{22}=\frac{1}{2}\:,\quad \big|t^{kn\omega}_{12} \big| \leq \frac{1}{2} & \qquad \text{if~$|\omega| > m$}\:.
\end{array} \right. \eeq
\end{Thm} \noindent
We note for clarity that, in view of~\eqref{tab} and~\eqref{aeone}, in the case~$|\omega|<m$ only
the exponentially decaying wave function enters the integral representation.
This has the effect that, asymptotically near infinity, only the spectrum for~$|\omega| \geq m$
is visible, in agreement with the mass gap in Minkowski space.

\subsubsection{Hamiltonian Formulation} \label{Sec Hamiltonian}
The Dirac equation~\eqref{Dirac} can be written in the Hamiltonian form
\beq \label{dirHamilton}
i \partial_t \psi = H \psi \:,
\eeq
where the Hamiltonian~$H$ is a spatial operator acting on the spinors.
Choosing the Cauchy surface~$\scrN$ as the surface of constant Schwarzschild time
and the domain~$\D(H)$
as the smooth and compactly supported spinorial wave functions on~$\scrN$,
the Hamiltonian is symmetric with respect to the scalar product~\eqref{print}, i.e.\
\[ (H \psi \,|\, \phi)_m = (\psi \,|\, H \phi)_m \qquad \text{for all~$\psi, \phi \in \D(H)$} \]
(for more details on this point in general stationary spacetimes see~\cite[Section~4.6]{intro}).
The Hamiltonian is essentially selfadjoint (see~\cite{chernoff} for details in a more general context).
Denoting the unique selfadjoint extension again by~$H$,
the Cauchy problem can be solved with the spectral calculus by
\beq \label{fullprop}
\psi(t) = e^{-i t H}\: \psi_0 \:.
\eeq
This is the abstract counterpart of the integral representation of Theorem~\ref{thmintrep}.
In simple terms, the solution~\eqref{SchwarzschildProp} can be understood as giving an integral
representation for the operator~$e^{-itH}$ restricted to an angular mode.
Noting that~$\omega$ is the spectral parameter, the integral in~\eqref{SchwarzschildProp}
can be understood as a spectral decomposition in terms of the spectral measure
(in particular, the spectrum of the Hamiltonian is the whole real axis).
In order to make these connections more precise, we first note that also the radial
Dirac equation after separation of variables~\eqref{Radial Dirac Eq} can be written in the Hamiltonian form,
\begin{flalign*}
	&i \frac{\partial}{\partial t} X^{kn} (t,u)= \big(H_{kn} X^{kn}|_t\big)(u)  \\
	\Longleftrightarrow \qquad &(\Dir - m)\,e^{-ik\varphi}\:\frac{1}{\Delta(r)^{1/4} \sqrt{r}}\:
	\begin{pmatrix}
		X^{kn}_-(t,u)Y^{kn}_-(\vartheta)\\
		X^{kn}_+(t,u)Y^{kn}_+(\vartheta)\\
		X^{kn}_+(t,u)Y^{kn}_-(\vartheta)\\
		X^{kn}_-(t,u)Y^{kn}_+(\vartheta)
	\end{pmatrix} = 0 \;,
\end{flalign*}
where the Hamiltonian~$H_{kn}$ now is an essentially self-adjoint operator on~$L^2(\R,\C^2)$
with dense domain~$\D(H_{kn}) = C^\infty_0(\R ,\C^2 )$.
This makes it possible to write the solution of the Cauchy problem as
\[ X (t,u) = \big(e^{-itH_{kn}} \,X_0 \big)(u) \qquad \text{with~$u \in \R$}\:. \]
Here, the initial data can be an arbitrary vector-valued function in the Hilbert space, i.e.\ $X_0 \in L^2(\R,\C^2)$.
If we specialize to smooth initial data with compact support, i.e.\ $X_0 \in C^\infty_0(\R ,\C^2)$,
then the time evolution operator can be written with the help of Theorem~\ref{thmintrep} as
\begin{align*}
\big(e^{-i tH_{kn}} X_0 \big)(u) = 
\frac{1}{\pi} \int_{-\infty}^\infty d\omega \:e^{-i\omega t} \sum_{a,b =1}^{2} t_{ab}^{kn\omega}& \,X_a^{kn\omega} (u) \:\la X_b^{kn\omega} | X_0 \ra \\[-0.7em] 
&\text{for~$X_0 \in C^\infty_0(\R ,\C^2)$} \notag
\end{align*}
We point out that this formula does not immediately extend to general~$X_0 \in L^2(\R,\C^2)$;
we will come back to this technical issue a few times in this work.

\subsubsection{Connection to the Full Propagator}
\label{SecCon}
We now explain how the solution of the Cauchy problem as given abstractly in~\eqref{fullprop}
can be decomposed into angular modes. Our considerations explain why we may
restrict attention to one angular mode instead of the full propagator and why we can use the ordinary $L^2$-scalar product instead of~$(.|.)_m$. We introduce the function
\[ S:=\Delta(r)^{1/4} \sqrt{r} \:. \] 
Moreover, for each fixed~$k \in \Z+1/2$ and~$n\in \Z$ we denote by~$(\H_m^0)_{kn}$ the completion of
the vector space
\[ V_{kn} := \bigg\{  S^{-1}\:e^{-ik\varphi} \:\begin{pmatrix}
	X^{kn\omega}_-(u)Y^{kn}_-(\vartheta)\\
	X^{kn\omega}_+(u)Y^{kn}_+(\vartheta)\\
	X^{kn\omega}_+(u)Y^{kn}_-(\vartheta)\\
	X^{kn\omega}_-(u)Y^{kn}_+(\vartheta)
\end{pmatrix} \: \bigg|\: X=(X_+,X_-) \in L^2(\R,\C^2)\bigg\} \:, \]
with respect to the scalar product~$(.|.)_m$ introduced in~\eqref{print}, i.e.\
\[ (\H_m^0)_{kn} := \overline{V_{kn}}^{(.|.)_m} \:. \]
 This space can be thought of as the mode-wise solution space of the Dirac-equation at time~$t=0$. Note that the entire Hilbert space of solutions at time~$t=0$, namely
\[ \H_m|_{t=0}  =: \H_m^0 \]
has the orthogonal decomposition
\beq
\label{HmDec}
	 \H_m^0 = \bigoplus_{i \in \N} (\H_m^0)_{k_i n_i} \:.
\eeq
(again with respect to~$(.|.)_m$), where~$((k_i,n_i))_{i\in \N}$ is an enumeration of~$(\Z+1/2)\times \Z$.
Furthermore, each space~$(\H_m^0)_{kn}$ can be connected with~$L^2(\R,\C^2)$ using the mapping
\begin{flalign*}
	\tilde{S}: \big((\H^0_m)_{kn} \:,\: (.|.) \big) &\rightarrow L^2(\R , \C^2)\:,
\end{flalign*} 
which for any~$(\psi_1, \cdots, \psi_4) \in (\H^0_m)_{kn}$ is given by
\begin{align*}
	&\big( \tilde{S}(\psi_1, \cdots, \psi_4)\big)_1 \\
	&= \int_{-1}^1 d \vartheta \:\cos \vartheta \int_0^{2\pi} d\varphi \:\Big\la \big(
	\psi_2(u, \vartheta, \varphi) \,, \,
	\psi_3(u, \vartheta, \varphi)
	\big) \:\Big|\: e^{-ik\varphi} \:\big(
	Y^{kn}_+(\vartheta)\:, \:
	Y^{kn}_-(\vartheta)
	\big) \Big\ra_{\C^2} \:,\\
	&\big( \tilde{S}(\psi_1, \cdots, \psi_4)\big)_2 \\
	&= \int_{-1}^1 d \vartheta \:\cos \vartheta \int_0^{2\pi} d\varphi \: \Big\la \big(
	\psi_4(u, \vartheta, \varphi) \,, \,
	\psi_1(u, \vartheta, \varphi)
	\big) \:\Big|\: e^{-ik\varphi} \:\big(
	Y^{kn}_+(\vartheta)\:, \:
	Y^{kn}_-(\vartheta)
	\big) \Big\ra_{\C^2} \:.
\end{align*}
It has the inverse
\begin{flalign*}
	\tilde{S}^{-1}:  L^2(\R , \C^2) &\rightarrow \big((\H_m^0)_{kn}\:,\: (.|.)_m\big)\:,  \\
	(X_+,X_-) &\mapsto S^{-1}e^{-ik\varphi}
	\begin{pmatrix}
		X_-(u)Y^{kn}_-(\vartheta)\\
		X_+(u)Y^{kn}_+(\vartheta)\\
		X_+(u)Y^{kn}_-(\vartheta)\\
		X_-(u)Y^{kn}_+(\vartheta)
	\end{pmatrix}\:.
\end{flalign*}
Then a direct computation shows the scalar products transform as
\[ \la \tilde{S} \psi\:  |\: \tilde{S} \phi \ra_{L^2} =(\psi \,| \, \phi)_m \qquad \text{for any }\phi,\psi \in (\H_m^0)_{kn}\:. \]
This implies that~$\tilde{S}$ is unitary and we can identify the two spaces. 

Now recall that the Dirac-equation can be separated by solutions of the form
\[ \hat{\psi}= S^{-1} e^{-ik\varphi}	\begin{pmatrix}
	X_-(t,u)Y^{kn}_-(\vartheta)\\
	X_+(t,u)Y^{kn}_+(\vartheta)\\
	X_+(t,u)Y^{kn}_-(\vartheta)\\
	X_-(t,u)Y^{kn}_+(\vartheta)
\end{pmatrix} \:, \]
and can then be described mode-wise by the Hamiltonian~$H_{kn}$ on the space~$L^2(\R, \C^2)$.
Therefore denoting
\[ \tilde{H}_{kn}:=\tilde{S}^{-1}H_{kn} \tilde{S} \:,\]
the diagonal block operator (with respect to the decomposition~\eqref{HmDec})
\[ \tilde{H}:=\text{diag} \big(\tilde{H}_{(k_1,n_1)}\:,\: \tilde{H}_{(k_2,n_2)}\:, \: \dots \:\big)\:,  \]
defines an essentially self-adjoint Hamiltonian for the original Dirac equation on the space~$\H_m^0$.

Moreover, any function of~$\tilde{H}$ is of the same diagonal block operator form.  The same holds for any multiplication operator~$\mathcal{M}_{\chi_{\tilde{U}}}$, where~$\tilde{U}$ is a spherically symmetric set
\[ \tilde{U}:= U \times S^2 \subseteq \R \times S^2 \:.\]
In particular, such an operator has the block operator representation
\[ \mathcal{M}_{\chi_{\tilde{U}}} = \text{diag} \big( \mathcal{M}_{\chi_{\tilde{U}}} \:,\: \mathcal{M}_{\chi_{\tilde{U}}} \:,\: \dots \big) \:. \]
We therefore conclude that when computing traces of operators of the form
\[ \chi_{\tilde{U}}f(\tilde{H}) \chi_{\tilde{U}}\qquad \text{or} \qquad f(\chi_{\tilde{U}} \tilde{H}\chi_{\tilde{U}}) \:, \]
(for some suitable function~$f$), we may consider each angular mode separately and then sum over the occupied states (and similarly for Schatten norms of such operators).

Moreover we point out that instead of~$(\H_m^0)_{kn}$ we can work with the corresponding objects in~$L^2(\R,\C^2)$, as the spaces are unitarily equivalent.
Note, that then the multiplication operator~$\mathcal{M}_{\chi_{\tilde{U}}}$ goes over to~$\mathcal{M}_{\chi_{U}}$, i.e.\
\[ \tilde{S}^{-1} \mathcal{M}_{\chi_{\tilde{U}}}\tilde{S} = \mathcal{M}_{\chi_{U}} \:.\]
In particular this leads to
\[ \tr \Big( \chi_{\tilde{U}}f(\tilde{H}) \chi_{\tilde{U}} \Big)  = \: \sum_{k,n} \: \tr \Big(  \chi_{\tilde{U}}f(\tilde{H}_{kn}) \chi_{\tilde{U}} \Big) = \: \sum_{k,n} \:\tr \Big(  \chi_{U}f(H_{kn}) \chi_{U} \Big)  \]
and
\[ \tr f\big( \chi_{\tilde{U}}\tilde{H} \chi_{\tilde{U}}\big) = \: \sum_{k,n} \: \tr f\big( \chi_{\tilde{U}}\tilde{H}_{kn} \chi_{\tilde{U}}\big) = \: \sum_{k,n} \:\tr f\big(  \chi_{U}H_{kn} \chi_{U} \big) \:.\]
	
\subsubsection{Asymptotics of the Radial Solutions} \label{rnProp at the Horizon}
We now recall the asymptotics of the solutions of the radial ODEs and specify our choice of
fundamental solutions.
Since we want to consider the propagator at the horizon, we will need near-horizon approximations of the solutions~$X^{kn\omega}$. In order to control the resulting error terms, we now state a slightly stronger version of \cite[Lemma 3.1]{tkerr}, specialized to the Schwarzschild case.

\begin{Lemma}
	\label{X_at_hor}
	For any~$u_2 \in \R$ fixed, in Schwarzschild space every solution~$X\equiv X^{kn\omega}$ for~$ u \in (-\infty, u_2)$ is of the form
	\begin{flalign*}
		X(u) = \begin{pmatrix}
			f_0^+ e^{-i\omega u} \\
			f_0^- e^{i\omega u}
		\end{pmatrix}
		+ R_0(u)
	\end{flalign*}
	where the error term~$R_0$ decays exponentially in~$u$, uniformly in~$\omega$. More precisely, writing
	\begin{flalign*}
		R_0(u)= 
		\begin{pmatrix}
			e^{-i\omega u} g^+(u) \\
			e^{i\omega u} g^- (u)
		\end{pmatrix} \:,
	\end{flalign*}
	the vector-valued function~$g=(g^+, g^-)$ satisfies the bounds
	\[ |g(u)| < c e^{du}\:, \quad  \left| \frac{d}{du}g(u)\right| \leq d ce^{du} \qquad \mathrm{for\;all\;} u<u_2\:, \]
	with coefficients~$c,d>0$ that can be chosen independently of~$\omega$ and~$u<u_2$.
\end{Lemma} \noindent
The proof, which follows the method in~\cite{tkerr}, is given in detail in Appendix~\ref{appA}.

We can now explain how to construct the fundamental solutions~$X_a=(X_a^+, X_a^-)$ 
for~$a=1$ and~$2$ (for this see also \cite[p. 41]{tkerr} and \cite[p. 9-10]{sigbh}).	
In the case~$|\omega| > m$ we choose~$X_1$ and~$X_2$ such that the corresponding functions~$f_0$ from the previous lemma are of the form
\begin{flalign*}
	f_0=\begin{pmatrix}
		1\\
		0
	\end{pmatrix}
	\;\mathrm{for \;} X_1 \qquad \text{and} \qquad
	f_0=\begin{pmatrix}
		0\\
		1
	\end{pmatrix}
	\;\mathrm{for \;} X_2 \:.
\end{flalign*}
In the case~$|\omega| \leq m$ we consider the behavior of solutions at infinity
(i.e.\ asymptotically as~$u \rightarrow \infty$). It turns out that there is
(up to a prefactor) a unique fundamental solution which decays exponentially.
We denote it by~$X_1$. Moreover, we choose~$X_2$ as an exponentially increasing fundamental solution.
We normalize the resulting fundamental system at the horizon by
\[ \lim\limits_{u \rightarrow - \infty}|X_{1\!/\!2}| =1 \:. \]
Representing these solutions in the form of the previous lemma we obtain
\begin{flalign*}
	X_{1\!/\!2}(u)=\begin{pmatrix}
		e^{-i\omega u} f^+_{0,1\!/\!2} \\
		e^{i\omega u} f^-_{0,1\!/\!2}
	\end{pmatrix}
	+ R_{0,1\!/\!2}(u)
\end{flalign*}
with coefficients~$f_{0,1\!/\!2}^\pm \in \C$. Due to the normalization, we know that
\[ |f_{0,1\!/\!2}|=1 \qquad \text{and in particular} \qquad |f_{0,1\!/\!2}^\pm| \leq 1\:. \]
Note however, that~$f_0$ and~$R_0$ from the previous Lemma may in general also depend on~$k$ and~$n$, but 
for ease in notation this dependence will be suppressed.

\subsection{A Few Functional Analytic Tools}
	\label{Some Elementary Properties of Symbols}
	\subsubsection{Basic Definitions}
	\label{Sec:TechToolsDefs}
	Later will often rewrite operators on~$L^2(\R^d,\C^n)$ as pseudo-differential operators of the form
	\begin{flalign}
		\label{Def Opa}
		\begin{split}
			\Big(\Opa (\CA)\psi \Big)(\bx) := \Big(\frac{\alpha}{2\pi}\Big)^d \int_{\R^d} \int_{\mathcal{U}} e^{-i\alpha \bxi \cdot (\bx-\by)}
			\: \CA(\bx, \by, \bxi)\: \psi(\by)\, d\by\, d\bxi \\ \text{for any } \psi \in C^\infty_0(\mathcal{U},\C^n)\:.
		\end{split}
	\end{flalign}
	where~$\mathcal{U}\subseteq \R^d$ is some open set. The so called {\em{symbol}}~$\CA$ is a suitable measurable matrix-valued map~$\CA: (\R^d)^3 \times (0,\infty) \rightarrow \mathrm{M}(n,n)$ such that the operator on~$C^\infty_0(\mathcal{U},\C^n)$ defined by~\eqref{Def Opa} can be extended continuously to all of~$L^2(\mathcal{U},\C^n)$. The parameter~$d\in \N$ can be thought of as the spatial dimension and the parameter~$n\in \N$ as the number of components of the wave function~$\psi$.
	Note that if~$\mathcal{U} \subsetneq \R^d$, then the operator~$\Opa(\CA)$ may still be considered an operator on~$L^2(\R^d,\C^n)$ if one replaces~$\CA(\bx, \by, \bxi)$ by~$\chi_\mathcal{U}(\bx) \CA(\bx, \by, \bxi) \chi_\mathcal{U}(\by)$. In fact in what follows we often identify these operators.
	Symbols denoted by lowercase letters usually indicate that the symbol is scalar-valued.
	Moreover, the symbols sometimes additionally depend on~$\alpha$ or other parameters. We usually denote this by corresponding super- or subscripts.
	For some symbols~$\CA$ the integral representation~\eqref{Def Opa} extends to all Schwartz- or even all $L^2$-functions. If this condition is needed for specific results, we will mention it explicitly. We will also establish some conditions on~$\CA$ that guarantee such extensions.
	The names of the arguments of the symbol~$\CA$ are adapted to the application in mind. In particular, if symbols that are {\em{not}} boldface this usually implies that they are scalar-valued, i.e.\ $d=1$.

	
	In order to conveniently compute the entanglement entropy we will often be interested in the trace of the following operator:
	\[ 
	D_\alpha(\eta_\kappa,\Lambda,\CA) :=  \eta_\kappa\big( \chi_\Lambda \Opa(\CA) \chi_\Lambda \big) - \chi_\Lambda \eta_\kappa\big(\Opa(\CA)\big) \chi_\Lambda\:, \]
	where~$\Lambda \subseteq \R^d$ is some measurable set which will be specified later.
	
	Moreover, in what follows we will often use the notation
	\[P_{\Omega,\alpha}:=\Opa(\chi_\Omega)\]
	for some measurable set~$\Omega\subseteq \R^d$, which emphasizes that this is a projection operator (that it is well-defined follows from Lemma~\ref{fOpa} in Appendix~\ref{SecPropertiesOpa}). In Remark~\ref{ExtScwartzFcts} we will see that the integral representation of such operators always extends to all Schwartz functions.
	
Furthermore as we will later see, we can often estimate traces of functions of operators 
in terms of \textit{Schatten norms}, which we now introduce (we refer to \cite[Ch.~11]{birman-solomjak} for more details on this topic). For a compact operator~$A$ in a separable Hilbert space~$\mathfrak{H}$ we denote 
by~$s_k(A), k = 1, 2, \dots,$ its singular values i.e. eigenvalues 
of the self-adjoint compact operator~$ \sqrt{A^*A}$ labeled in non-increasing order counting 
multiplicities.
For the sum~$A+B$ the following inequality holds: 
\begin{align}\label{eq:kyfan}
s_{2k}(A+B)\le s_{2k-1}(A+B)\le s_k(A) + s_k(B) \:. 
\end{align}
We say that~$A$ belongs to the Schatten-von Neumann class~$\SN_p$, $p>0$, if 
\[ \|A\|_p \coloneqq \big(\tr (A^*A)^{\frac{p}{2}}\big)^{\frac{1}{p}} \]
is finite. The functional 
$\|A\|_p$ defines a norm if~$p\ge 1$ 
and a quasi-norm if~$0<p<1$. With this (quasi-)norm the class~$\SN_p$ is a complete space. 
For~$0<p \leq 1$ the quasi-norm is actually a \emph{$p$-norm}, that is,  
it satisfies the following triangle inequality for all~$A, B \in \SN_p$,
\begin{align} \label{est:triangle}
\|A+B\|_p^p\le \|A\|_p^p + \|B\|_p^p\qquad (0<p \leq 1)\:.
\end{align} 
Moreover, for all~$A \in \SN_{q_1}$ and~$B \in \SN_{q_2}$ the following H\"older-type inequality holds 
(see \cite[Section~2.1]{sobolev-schatten} with reference to \cite[p.~262]{birman-solomjak}),
\beq \label{est:HoelderLike}
\| A B\|_q \leq  \| A \|_{q_1} \| B\|_{q_2}\:, \qquad \text{with } q^{-1} =q_1^{-1}+q_2^{-1}, \quad 0<q_1,q_2 \leq \infty \:.
\eeq

	\begin{Remark}\label{Rem:Schatten} $ $
		{\em{
%
We note that the $q$-th Schatten-norm is invariant under unitary transformations: Let~$\mathfrak{H}$ and~$\mathfrak{G}$ be Hilbert spaces, $U \in \Lin(\mathfrak{G},\mathfrak{H})$ unitary and~$A\in\SN^q\subseteq \Lin(\mathfrak{H})$ then
		\[ (U^{-1}AU)^* U^{-1}AU = (U^*AU)^* U^{-1}AU = U^* A^* U U^{-1} A U = U^{-1}A^*A U\;, \]
		which is unitarily equivalent to~$A^*A$ and thus has the same eigenvalues showing that
		\[ \|A\|_q = \|U^{-1}AU\|_q \:. \]
		In particular, in the case~$q=1$ this shows that the trace norm of~$A$ is conserved under unitary transformation. 
		\QEDrem}}
	\end{Remark}

Moreover, we will frequently use the following function norms (see for example \cite[p. 5-6]{sobolev-pseudo} with slight modifications)

\begin{Def}
	Let~$S^{(n,m,k)}(\R^d)$ with~$m,n,k \in \N_0$ be the space of all complex-valued functions on~$(\R^d)^3$, which are continuous, bounded and continuously partially differentiable in the first variable up to order~$n$, in the second to~$m$ and in the third to~$k$ and whose partial derivatives up to these orders are bounded as well. For~$a \in S^{(n,m,k)}(\R^d)$ and~$l,r>0$ we introduce the norm
	\[ N^{(n,m,k)}(a;l, r ) := \max_{\substack{0\leq \tilde{n} \leq n\\ 0\leq \tilde{m} \leq m \\ 0\leq \tilde{k} \leq k }} \sup_{\bx,\by,\bxi} l^{\tilde{n}+\tilde{m}}r^{\tilde{k}} \big| \nabla_{\bx}^{\tilde{n}}\nabla_{\by}^{\tilde{m}} \nabla_{\bxi}^{\tilde{k}}a(\bx,\by,\bxi) \big|\:. \]
	Similarly, $S^{(n,k)}(\R^d)$ with~$n,k \in \N_0$ denotes the space of all complex-valued functions on~$(\R^d)^2$, which are continuous and bounded and continuously partially differentiable in the first variable up to order~$n$ and in the second to~$k$ and whose partial derivatives up to these orders are bounded. For~$a \in S^{(n,m)}(\R^d)$
	and~$l,r>0$ we introduce the norm
	\[ N^{(n,k)}(a;l, r ) := \max_{\substack{0\leq \tilde{n} \leq n\\  0\leq \tilde{k} \leq k }} \sup_{\bx,\bxi} l^{\tilde{n}}r^{\tilde{k}} \big| \nabla_{\bx}^{\tilde{n}} \nabla_{\bxi}^{\tilde{k}}a(\bx,\bxi) \big|\:. \]
	Finally, by~$S^{(k)}(\R^d)$ with~$k\in \N_0$ we denote the space of all complex-valued functions on~$\R^d$, which are continuous and bounded and continuously partially differentiable up to order~$k$ and whose partial derivatives up to these orders are bounded. For~$a \in S^{(k)}(\R^d)$
	and~$r>0$ we introduce the norm
	\[ N^{(k)}(a; r ) := \max_{ 0\leq \tilde{k} \leq k } \sup_{\bxi} r^{\tilde{k}} \big| \nabla_{\bxi}^{\tilde{k}}a(\bxi) \big|\:. \]
\end{Def}
Note that any function~$a\in S^{(n,k)}$ may be interpreted as element of in~$S^{(n,m,k)}(\R^d)$ for any~$m\in \N_0$ by the identification
\[ a(\bx,\by,\bxi)\equiv a(\bx,\bxi) \qquad \text{for all } \by\in \R^d\:. \]
Then, for any~$l,r>0$ one has
\[ N^{(n,m,k)}(a;l,r) = N^{(n,k)}(a;l,r)\:. \]

\subsubsection{Estimates on $q$-normed ideals}
In order to generalize a theorem by Widom (Theorem~\ref{ThmWidom}) later on, we need a few estimates, which we state here.
Since Theorem~\ref{ThmWidom} admits only smooth functions, the next two lemmata will allow us to extend it to certain functions which do not need to be differentiable everywhere.

\begin{Lemma}\cite[Cor.~2.11, Cond.~2.9, Thm.~2.10]{sobolev-functions} \\
	\label{TechLemmaNonSmt}%
	Let~$q,r>0$ parameters, $n\geq 2$ a natural number and~$f \in C^n_0(-r,r)$. Let~$\SN \subset \Lin( \mathfrak{H})$
	(with a Hilbert space~$\mathfrak{H}$) be a~$q$-normed ideal such that there is~$\sigma \in (0,1]$ with
	\[ (n-\sigma)^{-1} <q\leq1\:.\]
	Moreover, consider a self-adjoint operator~$A$ on~$D(A)\subseteq \mathfrak{H}$ and a projection
	operator~$P$ such that~$PD(A) \subseteq D(A)$ and~$|PA(\id-P)|^\sigma \in \SN$ and~$PA(\id-P)$ extends to a bounded operator. Then
	\[
	\| f(PAP)P-Pf(A) \|_\SN \lesssim \max_{0 \leq k \leq n} \Big( r^k\|f^{(k)}\|_{L^\infty} \Big)\:
	r^{-\sigma} \:\big\| |PA(\id-P)|^\sigma \big\|_\SN
	\]
	with an implicit constant independent of~$A$, $P$ and~$f$.
\end{Lemma}

We will usually apply this and the following lemma to operators of the form 
\[A=\Opa(\CA),\qquad P = \chi_\Lambda\:, \]
for some subset~$\Lambda\subseteq \R^d$ and with~$\SN$ as the $q^\text{th}$ Schatten class.

In preparation for the next lemma we introduce the following condition.

\begin{cond}\cite[Condition~2.3]{sobolev-functions} 
	\label{cond:f4}
	Let~$n\in \N$ and~$f \in C^n(\R\setminus \{t_0\})\cap C(\R)$ a function with~$t_0\in \R$, such that there exist~$\gamma, R>0$ with
	\begin{align}
		\label{def:bl}
		\bl f \bl_n := \max_{0 \leq k \leq n}\sup_{t \neq t_0}|f^{(k)}(t)|\cdot |t-t_0|^{-\gamma+k}< \infty \:,
	\end{align}
	and~$\supp f \subseteq [t_0-R,t_0+R]$.
\end{cond}


\begin{Lemma}\cite[Theorem~2.10]{sobolev-functions}\\
	\label{TechLemmaNonDiff}
	Let~$f$ satisfy Condition~\ref{cond:f4} for some~$n\geq 2$ and~$\gamma, R >0$.
	Let~$\SN$, $A$ and~$P$ as in Lemma~\ref{TechLemmaNonSmt} with~$\sigma < \gamma$, then
	\[ \big\| f(PAP)P-Pf(A) \big\|_\SN \lesssim \bl f\bl_n \:R^{\gamma-\sigma} \:\big\| |PA(\id- P)|^\sigma \big\|_\SN\:, \]
	with an implicit constant independent of~$A$, $P$, $f$ and~$R$.
\end{Lemma}

\begin{Example}
	{\em{
			Ultimately we want to apply Lemma~\ref{TechLemmaNonDiff} to~$\eta_\kappa$ (times a cutoff function). Therefore we introduce the function
			\[ f(x): =\eta_\kappa(x)\,\Phi_1(x) \:, \Forany x \in \R \:,\]
			with a smooth non-negative function~$\Phi_1$ such that 	
			\[  
			\supp \Phi_1 =\Big[-\frac{3}{4},\:\frac{3}{4}\Big] \qquad \text{and} \qquad \Phi_1\big|_{\left[-\frac{1}{2},\:\frac{1}{2}\right]} \equiv 1  \:.
			\]
			Note that~$f$ then satisfies the conditions of Lemma~\ref{TechLemmaNonDiff} for any~$n\in \N$,  $\gamma<1$ arbitrary, $R=3/4$ and~$x_0=0$ . 
			
			This gives an idea how Lemma~\ref{TechLemmaNonDiff} can be used to estimate the error in Theorem~\ref{ThmWidom} caused by functions like~$\eta_\kappa$, which are not differentiable everywhere.
			\QEDrem}}
\end{Example}      

\subsubsection{Estimates of Pseudo-Differential operators}
Here we list a few previously establishes estimates on pseudo-differential operators, which we will use later on.
The first lemma shows that~$\Opa(a)$ is bounded with respect to the operator norm uniformly in~$\alpha$ as long as~$a\in S^{(n,m,k)}(\R^d)$.

\begin{Lemma}\cite[Lemma 3.9]{sobolev-pseudo}\label{OpaBounded} (adapted to our notation)
	Let
	$a\in S^{(k,k,d+1)}(\R^d)$ be a symbol such that~$\Opa(a)$ is well defined and its integral representation extends to all Schwartz functions. We choose~$k := \lfloor d/2 \rfloor +1$, $l_0>0$ and~$l, r >0$ such that~$\alpha l r \geq l_0$. Then
	\[ \| \Opa(a) \|_{\infty} \lesssim N^{(k,k,d+1)}(a;r,l)\:, \]
	with an implicit constant only depending on~$d, r$ and~$l_0$.
\end{Lemma}

The next corollary helps us to estimate the error caused by interchanging characteristic functions in position and momentum space.
\begin{Corollary}\cite[Corollary 4.7]{sobolev-schatten}(case~$d=1$)
	\label{Cor4.7}
	For any two open bounded intervals~$K,J$ as well as numbers~$q \in (0,1]$ and~$\alpha \geq 2$,
	the following estimate holds,
	\[ \| \chi_K P_{J,\alpha} (1-\chi_K)\|_q \lesssim (\log\alpha)^{1/q} \:, \]
	with a constant independent of~$\alpha$.
\end{Corollary}

The next proposition gives an estimate for terms of the form~$\chi_{\Lambda} \Opa(a)(1-\chi_\Lambda)$, which we will come up when applying Lemmata~\ref{TechLemmaNonSmt} or~\ref{TechLemmaNonDiff}. 
In order to state it we first need to introduce the following condition.
\begin{cond}\cite[Condition~3.1]{leschke-sobolev-spitzer}
	\label{cond:domain2}
	For~$d\geq 1$ the set~$\Lambda \subset \R^d$, satisfies one of the following requirements:
	\begin{itemize}
		\item[(1)] If~$d=1$, then~$\Lambda$ is a finite union of open intervals (bounded or unbounded) such that their closures are pair-wise disjoint.
		\item[(2)] If~$d \geq 2$, then~$\Lambda$ is  a Lipschitz region
		(i.e.\ an open set whose boundary is locally Lipschitz), and either~$\Lambda$ or~$\R^d \setminus \Lambda$ is bounded.
	\end{itemize}
\end{cond}

\begin{Proposition}\label{prop:cross_smooth} \cite[Proposition~3.2]{leschke-sobolev-spitzer}(Adapted to the cases needed and our notation)
	Let the region~$\Lambda\subset \R^d$ satisfy Condition~\ref{cond:domain2} and let~$\alpha_0>0$ be a constant. Let~$q \in (0,1]$ and
	\[
	\tilde{m} := \lceil (d+1)q^{-1} \rceil +1 \:.
	\]
	Let~$a$ be a scalar-valued symbol only depending on~$\bxi$, i.e.\ $a(\bx,\by,\bxi)\equiv a(\bxi)$ with support contained in~$B_\tau(\bmu)$ for some~$\bmu \in \R^d$ and~$\tau >0$. Assume that~$a \in S^{(\tilde{m})}(\R^d)$ and~$\Opa(a)$ is well defined with integral representation extending to all Schwartz functions. Then for any~$\alpha \tau \geq \alpha_0$,
	\[
	\| \chi_{\Lambda} \Opa(a)(1-\chi_\Lambda) \|_q \lesssim (\alpha \tau)^{\frac{d-1}{q}}N^{(\tilde{m})}(a;\tau)\:, 
	\]
	with implicit constants independent of~$a,\alpha, \tau$ and~$\bmu$.
\end{Proposition}

\begin{Proposition}\cite[Proposition 3.8 and p.~17]{sobolev-pseudo} with reference to \cite[Theorem 11.1]{birman-solomjak2}, \cite[Section 5.8]{birman-karadzov} and \cite[Theorem 4.5]{simon2005}
	\label{prp72} \\
	For~$\bz \in \Z$ set~$\mathscr{C}_{\bz} := \bz+(0,1]^d$ and for~$\sigma \in (0,\infty)$ and~$g \in L^2_{\loc}(\R^d)$ 
	\[ |g|_\sigma := \bigg[ \sum_{\bz \in \Z} \bigg( \int_{\mathscr{C}_{\bz}} |g(x)|^2 dx \bigg)^{\sigma/2} \:\bigg]^{1/\sigma}\:. \]
	Then, given functions~$a, h\in L^2_{\loc}(\R)$ with~$|a|_\sigma, |h|_\sigma<\infty$ for some~$\sigma \in (0,2)$,
	it follows that~$h \,\mathrm{Op}_1(a) \in \SN_\sigma$ (with integral representation extending to all Schwartz functions) and
	\[ \big\| h\: \mathrm{Op}_1 (a) \big\|_\sigma \leq C \:|h|_\sigma \:|a|_\sigma \:.  \]
\end{Proposition}

\section{The Regularized Projection Operator} \label{secregproj}
\subsection{Definition and  Basic Properties} \label{SecBlockForm}
As previously mentioned, the entropy is computed using the mode-wise regularized projection operator
to the negative frequency space~$(\Pi_-^\varepsilon)_{kn}$.
This operator emerges from~$e^{-itH_{kn}}$ from Section~\ref{Sec Hamiltonian} by setting~$t=i\varepsilon$ (the ``$i\varepsilon$"-regularization) and restricting to the negative frequencies.  Similar as explained in Section~\ref{SecCon}, for operators of this form it suffices to consider the corresponding operator for one angular mode~$(\Pi_-^\varepsilon)_{kn}$.
So more precisely, for any~$X \in C^\infty_0(\R, \C^2)$ the operator~$(\Pi_-^\varepsilon)_{kn}$ is defined by
\beq
	\label{Pi reg Def}
	\big( (\Pi_-^\varepsilon)_{kn}\: X \big) (x) := \frac{1}{\pi} \int_{-\infty}^0 d\omega \:e^{\varepsilon \omega} \sum_{a,b =1}^{2} t_{ab}^{kn\omega} X_a^{kn\omega} (x) \la X_b^{kn\omega} | X \ra \:,
\eeq
for any~$x \in \R$.

Since in this section we focus on one angular mode, we will drop the superscripts~$kn$ on the functions~$X^{kn\omega}_a$ and~$t^{kn\omega}_{ab}$. Moreover, we will sometimes write the~$\omega$-dependence of~$X^{kn\omega}_{a}$ or~$t^{kn\omega}_{ab}$ as an argument, i.e.\
\[ X^{kn\omega}_a(u) \equiv X^{\omega}_a(u)\equiv  X_a(u,\omega) \qquad \textit{for any }u \in \R \:.\]

The asymptotics of the radial solutions at the horizon (Lemma~\ref{X_at_hor}) yield the following boundedness properties for the functions~$X^{\omega}_a$:
\begin{Remark}
	\label{X Bounded}
{\em{Given~$u_2 \in \R$ and a constant~$C>0$, we consider measurable  functions~$X, Z : \R \rightarrow \C^2$
with the properties
\[ \supp X,\: \supp Z \subset (-\infty, u_2] \qquad \text{and} \qquad \|X\|_\infty, \|Z\|_\infty < C \:. \]
Then the estimate in Lemma~\ref{X_at_hor}  yields
	\[ \sum_{a,b}|t_{ab}^{\omega}|  \:|X(u)|\: |X_a^{\omega}(u,\omega)|\: |X_b^{\omega}(u',\omega)|\: |Z(u')| \leq 2 C^2 \,(1+ ce^{du})
	\,(1+ ce^{du'})  \:,  \]
	for almost all~$u,u',\omega \in \R$ and with constants~$c,d$ only depending on~$k,n$~and~$u_2$.
	
	If we assume in addition that~$X$ and~$Z$ are compactly supported, then
	for any~$g\in L^1(\R)$ the Lebesgue integral
	\[ \int_{-\infty}^\infty du \int_{-\infty}^\infty \frac{d\omega}{\pi}\:g(\omega)\int_{-\infty}^\infty du' \:t_{ab}^{\omega}\:X^\dagger(u)\: X_a^{\omega}(u,\omega)\: X_b^{\omega}(u',\omega)^\dagger\: Z(u')\:,  \]
	is well-defined. Moreover, applying Fubini we may interchange the order of integration arbitrarily. \QEDrem}}
\end{Remark}

Furthermore, we will need the following technical Lemma, which tells us that testing with smooth and compactly supported functions suffices to determine if a function is in~$L^2$ and to estimate its $L^2$-norm:
\begin{Lemma} 	\label{Lemma Testing} Let~$N$ be a manifold with integration measure~$\mu$.
Given a function~$f\in L^1_{\text{\rm{loc}}}(N, \C^n)$ (with~$n \in \N$), 
we assume that the corresponding functional on the test functions
\[ \Phi \::\: C^\infty_0(N, \C^n) \rightarrow \C \:,\quad v \mapsto \int_N \la v(x) \: | \: f(x) \ra_{\C^n}\; d\mu(x) \]
is bounded with respect to the $L^2$-norm, i.e.\
\[ \big| \Phi(v) \big| \leq C\: \|v\|_{L^2(N, \C^n)} \qquad \text{for all~$v \in C^\infty_0(N, \C^n)$}\:. \]
Then~$f \in L^2(N, \C^n)$ and~$\|f\|_{L^2(N, \C^n)} \leq C$.
\end{Lemma}
\Proof Being bounded, the functional~$\Phi$ can be extended continuously to~$L^2(N, \C^n)$.
The Fr{\'e}chet-Riesz theorem makes it possible to represent this functional by an $L^2$-function~$\hat{f}$
i.e.\ $\|\hat{f}\| _{L^2(N, \C^n)} \leq C$ and
\[ \int_N \la v(x), \big( f(x) - \hat{f}(x) \big) \ra_{\C^n}\; d\mu(x) = 0 
\qquad \text{for all~$v \in C^\infty_0(N, \C^n)$}\:. \]
The fundamental lemma of the calculus of variations (for vector-valued functions on a manifold)
yields that~$f=\hat{f}$ almost everywhere.
\QED
Now we have all the tools to prove the boundedness of the operator~$(\Pi_-^\varepsilon)_{kn}$.
\begin{Lemma}\label{lem:Pi-epsBounded}
	Equation~\eqref{Pi reg Def} defines a continuous endomorphism~$(\Pi_-^\varepsilon)_{kn}$ on~$L^2(\R, \C^2)$ with operator norm
	\[ \| (\Pi_-^\varepsilon)_{kn} \|_\infty \leq 1\:. \]
\end{Lemma}
\begin{proof}
	Let~$X,Z\in C^\infty_0(\R,\C^2)$ be arbitrary. We apply~$(\Pi_-^\varepsilon)_{kn}$ to~$X$ and test with~$Z$, i.e.\ consider
	\[ \Big\la Z \: \Big| \: \frac{1}{\pi} \int_{-\infty}^0 d\omega \: e^{\varepsilon \omega} \: \sum_{a,b =1}^2 t_{ab}^{\omega} \: X_a(u,\omega)  \: \big\la X_b^\omega\: \big|\: X \big\ra \Big \ra=:(*)\:. \]
	Applying Remark~\ref{X Bounded}, we may interchange integrations such that
	\[ (*) = \frac{1}{\pi} \int_{-\infty}^0 d\omega \: e^{\varepsilon \omega} \: \sum_{a,b =1}^2 t_{ab}^{\omega} \: \big\la Z \: \big| \: X_a^\omega \big\ra  \: \big\la X_b^\omega \: \big|\: X \big\ra \:.  \]
	Moreover, from \cite[proof of Theorem~3.6]{tkerr} we obtain the estimate
	\begin{flalign}
		\label{Prop Norm Est}
		\int_{-\infty}^\infty \frac{d\omega}{\pi} \:\bigg| \sum_{a,b =1}^2 t_{ab}^{\omega}\:
		\big\la X\:\big| \:X^{ \omega}_a\big\ra \big\la X^{\omega}_b\:\big|\:Z \big\ra  \bigg| \leq \| X \| \|Z \|\:,
	\end{flalign}
	which yields
	\[ |(*) | \leq \| X \| \|Z\| \:. \]
	Now by Lemma~\ref{Lemma Testing} we conclude that
	\[ (\Pi_-^\varepsilon)_{kn}X \in L^2(\R,\C^2)\qquad \text{and} \qquad \| (\Pi_-^\varepsilon)_{kn}X \| \leq \|X\| \:. \]
	This estimate shows that~$(\Pi_-^\varepsilon)_{kn}$ extends to a continuous endomorphism on~$L^2(\R,\C^2)$ with operator norm~$\| (\Pi_-^\varepsilon)_{kn} \|_\infty \leq 1$.
\end{proof}

\subsection{Functional Calculus for~$H_{kn}$}
In order to derive some more properties of~$(\Pi_-^\varepsilon)_{kn}$ we need to employ the functional calculus of~$H_{kn}$, as we want to rewrite
\[ (\Pi_-^\varepsilon)_{kn} = g(H_{kn}) \:,\]
for some suitable function~$g$.


The following two propositions constitute the main result of this section.
\begin{Proposition}
	\label{Prop g(H)1}
	Let~$g \in L^1(\R) \cap L^\infty(\R)$ be a real valued function.
	Then for any~$X \in C^\infty_0(\R,\C^2)$, the operator~$g(H_{kn})$ has the integral representation
	\begin{flalign}
		\label{g(H)-X}
		\big( g(H_{kn}) X\big)(u) 
		= \int_{-\infty}^{\infty} \frac{d\omega}{\pi} \: g(\omega) \int_{-\infty}^{\infty} du' \sum_{a,b=1}^{2} t_{ab}^{\omega}\: X_a(u,\omega) \:\big\la X_b(u',\omega) \:\big| \: X(u')\big\ra_{\C^2}\:,
	\end{flalign}  
	valid for almost any~$u\in \R$. Moreover, for any~$Z\in C^\infty_0(\R,\C^2)$,
	\begin{flalign}
		\label{phi g(H) psi eq}
		\la Z \,|\, g(H)X \ra  =\int_{-\infty}^{\infty} \frac{d\omega}{\pi} \:g(\omega) \sum_{a,b =1}^2 t_{ab}^{\omega}\:
		\la Z | X_a^{\omega} \ra  \la X_b^{\omega}| X \ra\:,
	\end{flalign}
\end{Proposition}

\begin{Proposition}
	\label{Prop g(H)2}
	Let~$g \in L^1(\R) \cap L^\infty(\R)$ be a real valued function.
	Then the operator~$g(H_{kn})$ has the following properties:
	\bitem
	\item[{\rm{(i)}}] The operator norm of~$g(H_{kn})$ is bounded, namely
	\[  \| g(H_{kn})\|_{\infty} \leq \|g\|_{L^\infty}\:. \]
	\item[{\rm{(ii)}}] The operator~$g(H_{kn})$ is self-adjoint.
	\eitem
\end{Proposition}

\begin{proof}[Proof of Proposition~\ref{Prop g(H)1}] We proceed in two steps. \\[0.5em]
	\underline{First step: Proof for~$g \in C^\infty_0(\R)$}:
	Since the Fourier transform is a bijection on the Schwartz space, for any~$g\in C^\infty_0(\R)$ there is a function~$\hat{g} \in S(\R)$ such that
	\[ g(\omega) = \int \hat{g}(t) \:e^{-i\omega t} \:dt \qquad  \mathrm{for\;any\;}\omega \in \R\:. \]
	We evaluate the right hand side of~\eqref{g(H)-X} for~$X\in C^\infty_0(\R,\C^2)$ arbitrary.
	Note that, when testing this with some~$Z \in C^\infty_0(\R,\C^2)$, we may interchange the~$u$-
	and~$\omega$-integrations due to an argument similar as in Remark~\ref{X Bounded}
	We thus obtain
	\begin{flalign*}
		\Big\la & Z \:\Big|\: \frac{1}{\pi} \int g(\omega) \sum_{a,b =1}^2 t_{ab}(\omega)\:X_a^{\omega}\: \la X_b^{\omega} | X \ra \: d\omega \Big\ra \\
		&=  \frac{1}{ \pi}\int g(\omega) \sum_{a,b =1}^2 t_{ab} (\omega) \:\big\la Z \:\big|\:X_{a}^\omega \big\ra\: \big\la X _{b}^\omega\:\big|\: X \big\ra \:d\omega  \\
		&= \frac{1}{\pi} \int d\omega \:\bigg( \int dt\: \hat{g}(t) \:e^{-it\omega} \bigg) \sum_{a,b =1}^2 t_{ab} (\omega)\:\big\la Z \:\big|\: X_{a}^\omega \big\ra\: \big\la X_{b}^\omega\:\big|\: X \big\ra =: (*) \:.
	\end{flalign*}
	Using the rapid decay of~$\hat{g}$ together with~\eqref{Prop Norm Est}, we can make use of the Fubini-Tonelli theorem which leads to
	\begin{flalign*}
		(*)=& \frac{1}{ \pi} \int dt\: \hat{g}(t) \int d\omega\: e^{-it\omega} \sum_{a,b =1}^2 t_{ab}(\omega)\:
		\big\la Z\: \big|\: X_{a}^\omega \big\ra \:\big\la X_{b}^\omega\:\big|\: X \big\ra \:.
	\end{flalign*}
	It is shown in~\cite{tkerr} that
	\[ (*)= \int\hat{g}(t) \:\big\la Z \:\big|\: e^{-itH_{kn}}  X \big\ra \:dt \:. \]
	Now we can again apply Fubini's theorem due to the rapid decay of~$\hat{g}$ and the boundedness of the operator~$e^{-itH_{kn}}$ (which follows from~\eqref{Prop Norm Est}), leading to
	\[ (*) = \Big\la  Z \:\Big|\: \bigg(\int \hat{g}(t) \:e^{-itH_{kn}} dt \bigg) \: X \Big\ra \:. \]
	Next we use the multiplication operator version of the spectral theorem to rewrite~$H_{kn}$ as
	\[ H_{kn} = U \: f\: U^{-1} \:,\]
	with a suitable unitary operator~$U$ and a Borel function~$f$ on the corresponding measure space~$\big(\sigma(H_{kn}) ,\Sigma, \mu\big)$. Then
	\[ e^{-itH_{kn}} = U \:e^{-itf}\: U^{-1}\:, \]
	and thus for any~$\tilde{X} \in L^2(\R, \C^2)$ and almost any~$x\in A$ it holds that
	\[\Big(  \Big(\int \hat{g}(t) e^{-itf}dt\Big) U^{-1} \tilde{X}\Big) (x)  = \Big(\int \hat{g}(t) e^{-itf(x)} dt\Big)\: (U^{-1} \tilde{X})(x) = g\big(f(x)\big)\: (U^{-1} \tilde{X})(x) \:, \]
	which leads to
	\[ (*) = \big\la  Z\:\big|\:U (g \circ f)\: U^{-1}  X\big\ra = \big\la  Z \: \big|\: g(H_{kn})  X \big\ra\:. \]
	Thus we conclude that for any~$X,Z \in C^\infty_0(\R\times S^2, \C^4)$,
	\[ \Big\la Z \:\Big|\: \frac{1}{\pi}\int_{-\infty}^\infty g(\omega) \sum_{a,b =1}^2 t_{ab}^{\omega} \:X_a^{\omega} \:\la X_b^{\omega} | X \ra \Big\ra = \big\la Z\:\big|\: g(H_{kn})X \big\ra \:.  \]
	Then Lemma~\ref{Lemma Testing} (together with similar estimates as before) yields that
	\[\int_{-\infty}^\infty g(\omega) \sum_{a,b =1}^2 t_{ab}^{\omega}\: X_a^{\omega}(.) \:\la X_b^{\omega} | X \ra\in L^2(\R \times S^2,\C^4)\:, \]
	and therefore
	\[ g(H)X =\sum_{kn} \int_{-\infty}^\infty g(\omega) \sum_{a,b =1}^2 t_{ab}^{\omega}\: X_a^{\omega} \:\la X_b^{\omega} | X \ra \:, \]
	almost everywhere. \\[0.5em]
	\underline{Second step: Proof for~$g \in L^1(\R) \cap L^\infty(\R)$:} $\quad$
	We choose a sequence of test func\-tions~$(g_n)_{n\in \N}$ in~$C^\infty_0(\R)$
	which is uniformly bounded by a constant~$C>0$ such that\footnote{Note that such a sequence can always be constructed from an arbitrary sequence~$(\tilde{g}_n)_{n\in \N}\subseteq C^\infty_0(\R)$ converging to~$g$ in~$L^1(\R)$ by smoothly cutting off the values of the function whenever its absolute value is larger than~$\|g\|_\infty+1$ (to ensure uniform boundedness) and then going over to a subsequence (to get pointwise convergence a.e., see for example \cite[Theorem 3.12]{rudin}).}
	\[ g_n \rightarrow g \qquad \text{in~$L^1(\R)$ and pointwise almost everywhere } .\]
	Then with~$f$ and~$U$ as before (where we applied the spectral theorem to~$H_{kn}$) we obtain for any~$X \in L^2(\R,\C^2)$
	\[  g_n\big(f(x)\big) (U^{-1} X)(x) \;\rightarrow\;  g\big(f(x)\big) (U^{-1} X)(x) \qquad
	\text{for almost all } x \in \sigma(H_{kn}) \:. \]
	Moreover, with the notation~$\Delta g_n := g_n-g$ we can estimate
	\[ \big| \:\Delta g_n\big(f(x)\big)\:(U^{-1}X)(x) \big| \leq \big(C+\|g\|_{\infty}\big)\: \big| (U^{-1}X)(x) \big| \qquad \text{for almost all } x \in \sigma(H_{kn}) \:.\] 
	So the function~$\big(C+\|g\|_{\infty}\big) |U^{-1}X| \in L^2(A,\mu)$ dominates the sequence of measurable functions~$\big(\:\mathcal{M}_{\Delta g_n \circ f} (U^{-1}X)\:\big)_{n\in \N}$ which additionally tends to zero pointwise almost everywhere. Therefore, using Lebesgue's dominated convergence theorem, we conclude that
	\[ \mathcal{M}_{g_n\circ f} \:U^{-1}\:X  \rightarrow \mathcal{M}_{g \circ f} \:U^{-1}\:X
	\qquad \text{in~$L^2(\sigma(H_{kn},\mu)$} \]
	and thus
	\[ g_n(H_{kn}) X  \rightarrow g(H_{kn}) X \qquad \text{in~$L^2(\R,\C^2)$} \:. \]
	In particular, we conclude that for any~$X,Z \in C^\infty_0(\R,\C^4)$,
	\begin{flalign}
		\label{g(H) L1 Cont.}
		\la  Z\:|\:  g_n(H_{kn}) X\ra \rightarrow \la  Z\:| \: g(H_{kn}) X\ra  \:.
	\end{flalign}

	Next we need  to show that the corresponding integral representations converge. To this end, we note that, just as in the first case, we may interchange integrations in the way
	\begin{flalign*}
		\Big \la & Z \:\Big|\: \int_{-\infty}^\infty d\omega \:\Delta g_n(\omega) \sum_{a,b =1}^2 t_{ab}^{\omega}\: X_a^{\omega} \:\la X_b^{\omega}| X \ra \Big\ra \\ 
		&=  \int_{-\infty}^\infty d\omega  \:\Delta g_n(\omega) \sum_{a,b =1}^2 t_{ab}^{\omega}\: \la Z | X_a^{\omega} \ra \la X_b^{\omega}| X \ra =:(**)\:,
	\end{flalign*}
	Now keep in mind that Remark~\ref{X Bounded} also yields the bound
	\[ \bigg|\sum_{a,b =1}^2 t_{ab}^{\omega} \:\la Z | X^{\omega}_a \ra \la X^{\omega}_b| X \ra\bigg| \leq  C_{Z,X} \:, \]
	which holds uniformly in~$\omega$. Using this inequality, we obtain the estimate
	\begin{flalign*}
		|(**)| \leq& \: \int_{-\infty}^\infty d\omega \big|\Delta g_n(\omega)\big| \:\Big|\sum_{a,b =1}^2 t_{ab}^{\omega} \la Z | X_a^{kn\omega} \ra \la X_b^{kn\omega} | X \ra\Big|\\
		 \leq&\: C_{Z,X}\int_{-\infty}^\infty d\omega \big|\Delta g_n(\omega)\big| \overset{n \rightarrow \infty}{\longrightarrow} 0  \:.
	\end{flalign*}	Combined with~\eqref{g(H) L1 Cont.}, this finally yields for any~$X,Z \in C^\infty_0(\R\times S^2, \C^4)$
	\[ \la Z \:|\: g(H) X \ra =  \Big \la Z \:\Big|\: \sum_{kn} \int_{-\infty}^\infty d\omega \:g(\omega) \sum_{a,b =1}^2 t_{ab}^{\omega} \:X_a^{\omega}\: \la X_b^{\omega}| X \ra \Big\ra \:. \]
We obtain~\eqref{g(H)-X} just as in the first case using Lemma~\ref{Lemma Testing}.
Finally, \eqref{phi g(H) psi eq} follows by testing with~$Z$ and again interchanging the integrals as explained before.
\end{proof}	

\begin{proof}[Proof of Proposition~\ref{Prop g(H)2}]
\bitem
\item[{\rm{(i)}}] This follows directly from~\eqref{phi g(H) psi eq}  together with~\eqref{Prop Norm Est}, since
	\begin{flalign}
		|\la Z\:|\: g(H_{kn}) X \ra_{L^2}| &\overset{\eqref{phi g(H) psi eq}}{=} \Big| \int \frac{d\omega}{\pi}\: g(\omega) \sum_{a,b=1}^{2} t_{ab}^{\omega}  \big\la Z\:|\: X_a^{\omega} \big\ra_{L^2}\big\la  X_b^{\omega} \:|\: X \big\ra_{L^2} \Big| \label{Int rep est2} \\
		&\overset{\eqref{Prop Norm Est}}{\leq} \|g\|_\infty \|Z\|_{L^2} \|X\|_{L^2} \:. \label{Int rep est3 }
	\end{flalign}
\item[{\rm{(ii)}}] Using~\eqref{g(H)-X}, the following computation shows that the operator~$g(H_{kn})$ is also self-adjoint because for any~$X,Z\in \Xi$ we have
	\begin{flalign*}
		&\; \;\,\la Z\:|\: g(H_{kn}) X \ra_{L^2}\\
		 &\; \;\,= \int du \int \frac{d\omega}{\pi}\: g(\omega) \int du' \sum_{a,b=1}^{2} \overline{t_{ba}^{\omega} \: \big\la X(u') \:|\: X_b(u',\omega) \big\ra_{\C^2}\big\la X_a(u,\omega)\:|\: Z(u) \big\ra_{\C^2}} \\
		&\overset{\mathrm{Fubini}}{=} \overline{\int du' \int  \frac{d\omega}{\pi} \: g(\omega) \int du \sum_{a,b=1}^{2} t_{ba}^{\omega} \: \big\la Z (u) \:|\: X_b(u,\omega) \big\ra_{\C^2}\big\la X_a(u',\omega)\:|\: X(u') \big\ra_{\C^2}}\\
		&\; \;\, = \overline{ \big\la X \: |\:g(H_{kn}) Z\big\ra_{L^2} } 
		= \big\la  g(H_{kn}) Z \:|\:X \big\ra_{L^2} \:.
	\end{flalign*}
	Note that applying Fubini is justified in view of Remark~\ref{X Bounded}. 
	From this equation the self-adjointness follows by continuous extension.
\eitem
\end{proof}

Now we apply these results to the operator~$\Pi_-^\varepsilon$:

\begin{Corollary}	
	\label{Cor g(H)}
	Consider the function
	\[ 
		g: \R \rightarrow \R\:, \qquad \omega \mapsto\chi_{(-\infty,0)}(\omega) \:e^{\varepsilon \omega}\:, \]
	then we have
\[ 
(\Pi_-^\varepsilon)_{kn}=g(H_{kn}) \]
	Moreover, for~$\eta_\kappa$ as before we have:
	\begin{flalign}
		\label{Cor g(H) eq 2}
		\eta_\kappa( (\Pi_-^\varepsilon)_{kn})=(\eta_\kappa \circ g) (H_{kn})
	\end{flalign}
\end{Corollary}
\begin{proof}
	First of all note that 
	\[ (\Pi_-^\varepsilon)_{kn} = g(H_{kn})\:, \]
	as both operators clearly agree on the dense subset~$C^\infty_0(\R,\C^2) \subseteq L^2(\R,\C^2)$ (see Proposition~\ref{Prop g(H)1}) and are bounded (see Lemma~\ref{lem:Pi-epsBounded} and Proposition~\ref{Prop g(H)2}). 
	Equation
	\eqref{Cor g(H) eq 2} then follows by applying the functional calculus of~$H_{kn}$ (which is applicable due to Proposition~\ref{Prop g(H)2}). 
\end{proof}

\subsection{Representation as a Pseudo-Differential Operator}
\label{sec:RepPseudo}
The general idea is to rewrite~$\Pi_-^\varepsilon$ in the form of~$\Opa(\CA)$ and identify~$\alpha$ with the inverse regularization constant:
\[ \alpha = \frac{l_0}{\varepsilon} \:, \] 
with a suitable reference length~$l_0$.

With the help of~\eqref{g(H)-X}, we obtain for any~$\psi \in C^\infty_0(\R, \C^2)$
\[ \big( (\Pi_-^\varepsilon)_{kn} \psi \big)(u) = \int_{-\infty}^\infty (\Pi_-^\varepsilon)_{kn}(u,u')\: \psi(u')\: du' \]
with the kernel
\begin{flalign}
\label{Pi kernel}
	&(\Pi_-^\varepsilon)_{kn}(u,u')=\frac{1}{\pi} \int_{-\infty}^\infty  e^{-i\omega(u-u')} \:\bigg[
	\begin{pmatrix}
		\mathfrak{a}_{\varepsilon,11}(\omega) & \mathfrak{a}_{\varepsilon,12}(u,\omega) \\
	\mathfrak{a}_{\varepsilon,21}(u,\omega) & \mathfrak{a}_{\varepsilon,22}(\omega)
	\end{pmatrix} + \mathcal{R}_0(u,u',\omega) \bigg] \: d\omega
\end{flalign}
and
\begin{flalign*}
	\mathfrak{a}_{\varepsilon,11}(\omega) &= e^{\varepsilon \omega}\: \Big(\big|f_{0,1}^+(\omega) \:\big|^2\:\chi_{(-m,0)}(\omega)+\frac{1}{2}\:\chi_{(-\infty, -m)}(\omega)\Big)\\ 
	\mathfrak{a}_{\varepsilon,12}(u,\omega) &= e^{-\varepsilon\omega} \:e^{2i\omega u}\:\Big( \overline{\fm(-\omega)}\: \fp (-\omega) \:\chi_{(0,m)}(\omega) + t_{12}(-\omega)\:\chi_{(m, \infty)}(\omega) \Big)\\ 
	\mathfrak{a}_{\varepsilon,21}(u,\omega) &= e^{\varepsilon\omega}\: e^{2i\omega u}\:\Big( \fm(\omega)\:\overline{\fp (\omega)}\:\chi_{(-m,0)}(\omega) + t_{21}(\omega)\:\chi_{(-\infty, -m)}(\omega) \Big)\\
	\mathfrak{a}_{\varepsilon,22}(\omega) &= e^{-\varepsilon \omega}\: \Big(\big|f_{0,1}^-(\omega) \:\big|^2\:\chi_{(0,m)}(\omega)+\frac{1}{2}\:\chi_{(m,\infty)}(\omega)\Big)\:, 
\end{flalign*}
and some error term~$\mathcal{R}_0(u,u',\omega)$ related to the error term~$R_0(u)$ in Lemma~\ref{X_at_hor}. A more detailed computation is given in Appendix~\ref{appB}. Moreover, the more precise form of~$\mathcal{R}_0(u,u',\omega)$ is found in Section~\ref{Fast decaying terms}. Note that for the Schwarzschild case we always replace~$\bx$ by~$u$ and~$\by$ by~$u'$ to emphasize that we are working with Regge-Wheeler coordinates. 

In order to bring~$(\Pi_-^\varepsilon)_{kn}$ in the form of~$\Opa(\CA)$, we need to rescale the~$\omega$-integral by a dimensionless parameter~$\alpha$. As previously mentioned the idea is to set~$\alpha = l_0/\varepsilon$ with some reference length~$l_0$. In Schwarzschild-space the only scaling parameter of the geometry is the mass of the black hole~$M$. Thus we choose as reference length~$l_0=M$ and rescale the $\omega$-integral by
\[ \alpha := \frac{M}{\varepsilon} \:. \]
Introducing the notation
\[ \xi := \frac{\omega}{\alpha} = \frac{\varepsilon \omega}{M} \:, \]
we thereby obtain
\begin{flalign}
	\label{FullOpKernel}
	\begin{split}
		(\Pi_-^\varepsilon)_{kn}(u,u')=&\frac{\alpha}{\pi} \int  e^{-i\alpha \xi(u-u')}\Big[
		\begin{pmatrix}
			\mathfrak{a}_{\varepsilon,11}(M\xi/\varepsilon) & \mathfrak{a}_{\varepsilon,12}(u,M\xi/\varepsilon) \\
			\mathfrak{a}_{\varepsilon,21}(u,M\xi/\varepsilon) & \mathfrak{a}_{\varepsilon,22}(M\xi/\varepsilon)
		\end{pmatrix} \\
		&\phantom{\frac{\alpha}{\pi} \int e^{-i\alpha \xi(u-u')}\Big[}+ \mathcal{R}_0(u,u',(M\xi/\varepsilon)) \Big] \: d\xi\:,
	\end{split}
\end{flalign}
and set
\begin{flalign}
	\mathcal{R}_0^{(\varepsilon)}(u,u',\xi)&:=\mathcal{R}_0(u,u',(M\xi/\varepsilon)) \\
	\label{Def Aeps}
	\CA^{(\varepsilon)}(u,\xi)&:= \begin{pmatrix}
		\mathfrak{a}_{\varepsilon,11}(M\xi/\varepsilon) & \mathfrak{a}_{\varepsilon,12}(u,M\xi/\varepsilon)\\
		\mathfrak{a}_{\varepsilon,21}(u,M\xi/\varepsilon) & \mathfrak{a}_{\varepsilon,22}(M\xi/\varepsilon)
	\end{pmatrix}\:.
\end{flalign}
Note that in the matrix-valued functions~$a_{\varepsilon}$ and~$\mathcal{R}_0$ we use the scaling parameter~$\varepsilon$ and otherwise~$\alpha$. This is convenient because we will first consider the~$\alpha \rightarrow \infty$ limit of an operator related to the~$\varepsilon \rightarrow 0$ limit of~$\CA^{(\varepsilon)}$ and then estimate the errors caused by this procedure. In this sense~$\alpha$ and~$\varepsilon$ can at first be considered independent scaling parameters. When considering the limiting case~$\varepsilon \rightarrow 0$ however, one has to keep their relation in mind.

\section{The Regularized Fermionic Vacuum State and its Entanglement Entropy} \label{secregvac}
After the above preparations, we can now define the regularized fermionic vacuum state
in the Schwarzschild geometry as well as the corresponding
(R{\'e}nyi) entanglement entropy of the event horizon.
Our starting point is the observation that a quasi-free fermionic state is uniquely described
by its reduced one-particle density operator~$D$ (see Definition~\ref{defreduced}).
We want to choose~$D$ as the regularized projection operator onto all the negative-frequency
solutions of the Dirac equation. Using the spectral calculus for the Hamiltonian~$H$ in the
Dirac equation in the Hamiltonian form~\eqref{dirHamilton}, our first ansatz is
\beq \label{Dfirst}
D = g(H) \qquad \text{with} \qquad g(\omega) := \chi_{(-\infty,0)}(\omega) \:e^{\varepsilon \omega} \:.
\eeq
Note that, in the limiting case~$\varepsilon \searrow 0$, the operator~$D$ goes over to the
projection operator to all negative-frequency solutions. The corresponding quasi-free state~$\Omega$
is pure. In the case~$\varepsilon>0$, the function~$e^{\varepsilon \omega}$ gives a smooth cutoff
for large frequencies on the energy scale~$1/\varepsilon$. Even in this regularized situation,
the spectral function~$g$ is discontinuous at~$\omega=0$. This implements the physical picture that,
in the vacuum, all negative-frequency one-particle states are occupied, whereas all positive-frequency
states are not.

Our fist ansatz~\eqref{Dfirst} has the shortcoming that it involves an infinite number of angular momentum modes,
giving rise to a divergence in~\eqref{entropy BH}. In order to remedy the situation, we
let~${\mathscr{O}}$ be a finite subset of~$(\Z+1/2) \times \N$ (referred to as
the {\em{occupied angular momentum modes}} or {\em{occupied one-particle states}})
and choose
\beq \label{occupy}
D = \Pi_-^\varepsilon := \sum_{\text{$(k,n)$ occupied}} (\Pi_-^\varepsilon)_{kn} \:,
\eeq
where we sum over all~$(k,n) \in {\mathscr{O}}$. Here
the mode-wise regularized projection operator~$(\Pi_-^\varepsilon)_{kn}$
is defined using the integral representation by~\eqref{Pi reg Def}. Alternatively and equivalently,
it can be characterized using the spectral calculus as described in Corollary~\ref{Cor g(H)}.

We note that, for simplicity, for each angular momentum mode we choose the same
regularization length~$\varepsilon$. More generally, one could consider a regularized vacuum
state where~$\varepsilon=\varepsilon(k,n)$ depends on the angular mode. Since the entropy
can be decomposed into the sum of the entropies of all angular momentum modes, all our
results generalize immediately to this more general state.

Choosing~$D$ according to~\eqref{occupy}, we consider the (R{\'e}nyi) entanglement entropy as defined by~\eqref{entropygen},
where~$\Lambda$ is chosen as
an annular region in~\eqref{def:TildeK}; see also Figure~\ref{Fig:Lambda2}.
Note that in the Regge-Wheeler coordinates the horizon is located at~$-\infty$, so ultimately we want to consider the limit~$u_0 \rightarrow - \infty$ and~$\rho \rightarrow \infty$.

	As explained in Section~\ref{SecCon} we can compute the trace mode-wise by going over to the subregions~$\mathcal{K}$:
\[ 
		\tr \Big( \eta_\kappa\big( \chi_{\Lambda} \Pi_-^\varepsilon\big) \chi_{\Lambda} )  -\chi_{\Lambda} \eta_\kappa\Big(\Pi_-^\varepsilon\big)\chi_{\Lambda} \Big)		
		=\: \sum_{ \mathclap{\substack{ (k,n) \\ \text{occupied}}}} \: \tr\Big( \eta_\kappa \big( \chi_{\mathcal{K}} (\Pi_-^\varepsilon)_{kn} \chi_{\mathcal{K}} \big) - \chi_{\mathcal{K}} \eta_\kappa\big((\Pi_-^\varepsilon)_{kn}\big) \chi_{\mathcal{K}} \big)\Big) \:.  \]
Thus we define the mode-wise R{\'e}nyi entropy of the black hole as in~\eqref{def modewise entropy} by
	\[S_{\kappa, kn}^{\mathrm{BH}} = \frac{1}{2}\:\lim\limits_{\rho \rightarrow \infty}\lim\limits_{\alpha \rightarrow \infty} \frac{1}{\tilde{f}(\alpha)}\lim\limits_{u_0\rightarrow -\infty}
	\tr \Delta_\kappa\big((\Pi_-^\varepsilon)_{kn}, \mathcal{K}\big) \:, \]
	where~$\tilde{f}(\alpha)$ is a function describing the highest order of divergence in~$\alpha$ (we will later see that here~$\tilde{f}(\alpha)=\log \alpha$).
The complete entanglement entropy of the black hole is then the sum over all occupied modes (see~\eqref{entropy BH}).
\begin{figure}
	\centering
	\includegraphics[width=.8\textwidth]{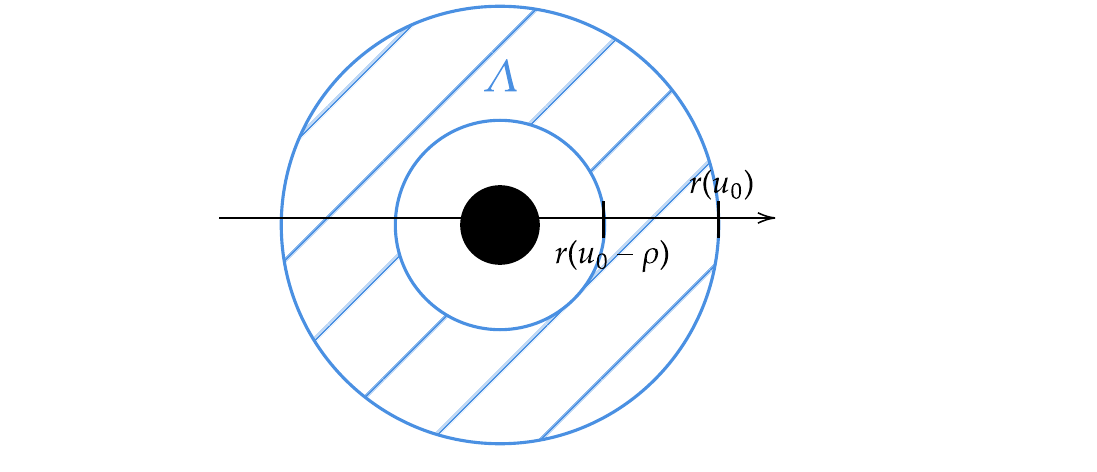}
	\caption{Cross section visualizing the set~$\Lambda=\mathcal{K} \times S^2$.}
	\label{Fig:Lambda2}
\end{figure}

In order to compute this in more detail, we will prove that
\begin{flalign}
	&\label{Limit rnProp_kn}
	\lim\limits_{\rho \rightarrow \infty}\lim\limits_{\alpha \rightarrow \infty} \frac{1}{f(\alpha)} \lim\limits_{u_0\rightarrow -\infty}
	\tr\Big( \eta_\kappa  \big( \chi_{\mathcal{K}} (\Pi_-^\varepsilon)_{kn} \chi_{\mathcal{K}} \big) - \chi_{\mathcal{K}} \eta_\kappa\big((\Pi_-^\varepsilon)_{kn}\big) \chi_{\mathcal{K}} \Big)\\
	&= \label{Limit Opa(b)}
	\lim\limits_{\rho \rightarrow \infty}\lim\limits_{\alpha \rightarrow \infty} \frac{1}{f(\alpha)} \lim\limits_{u_0\rightarrow -\infty} \tr\Big( \eta_\kappa \big( \chi_{\mathcal{K}} \Opa(\CA_0) \chi_{\mathcal{K}} \big) - \chi_{\mathcal{K}} \eta_\kappa\big(\Opa(\CA_0)\big) \chi_{\mathcal{K}} \Big)\:,
\end{flalign}
where 
\begin{flalign}
	\label{LimOpKernel}
	\CA_0(\xi) :=  \begin{pmatrix}
		e^{M\xi} \chi_{(-\infty, 0)}(\xi) & 0 \\
		0 & e^{-M\xi} \chi_{(0, \infty)}(\xi)
	\end{pmatrix}\:.
\end{flalign}
(We will later see that the operators in~\eqref{Limit rnProp_kn} and~\eqref{Limit Opa(b)} are well-defined and trace class). The notation~$\CA_0$ is supposed to emphasize the connection to the~$\varepsilon \rightarrow 0$ limit of~$A_\varepsilon$.
Since~$\CA_0$ is diagonal the computation of~\eqref{Limit Opa(b)} is much easier than the one for~\eqref{Limit rnProp_kn}. In fact we have
\[ \eqref{Limit Opa(b)} = \sum_{j=1}^2 \lim\limits_{\rho \rightarrow \infty}\lim\limits_{\alpha \rightarrow \infty} \frac{1}{f(\alpha)} \lim\limits_{u_0\rightarrow -\infty} \tr\Big( \eta_\kappa \big( \chi_{\mathcal{K}} \Opa(\mathfrak{a}_{0,j}) \chi_{\mathcal{K}} \big) - \chi_{\mathcal{K}} \eta_\kappa\big(\Opa(\mathfrak{a}_{0,j})\big) \chi_{\mathcal{K}} \Big) \]
with the scalar functions
\beq \label{a01def}
\mathfrak{a}_{0,1}(\xi) := e^{M\xi} \chi_{(-\infty, 0)}(\xi) \qquad \mathrm{and } \quad \mathfrak{a}_{0,2}(\xi) := e^{-M\xi} \chi_{(0, \infty)}(\xi)\:.
\eeq
This reduces the computation of~\eqref{Limit Opa(b)} to a problem for real-valued symbols for which many results are already established.

\begin{Remark} {\em{
We point out that our definition of the entanglement entropy differs from
the conventions in~\cite{helling-leschke-spitzer, leschke-sobolev-spitzer2}
in that we do not add the entropic difference operator of the complement of~$\Lambda$.
This is justified as follows. On the technical level, our procedure is easier, because
it suffices to consider compact spatial regions (indeed, we expect that the
entropic difference operator on the complement of~$\Lambda$ is not trace class).
Conceptually, restricting attention to the entropic difference operator of~$\Lambda$
can be understood from the fact that occupied states which are supported either inside or outside~$\Lambda$
do not contribute to the entanglement entropy.
Thus it suffices to consider the states which are non-zero both inside and outside.
These ``boundary states'' are taken into account already in the entropic difference operator~\eqref{EntropicDiff}.

This qualitative argument can be made more precise with the following formal computation,
which shows that at least the unregularized entropic difference is the same for the inner and the outer parts:
 First of all note that~$\eta_\kappa(x)$ vanishes at~$x=0$ and~$x=1$. Since~$\Pi_-$ is a projection this means that
\[
	\eta_\kappa(\Pi_-)=0\qquad 
	\text{and therefore} \qquad \tr\big( \chi_{\Lambda}\:\eta_\kappa(\Pi_-) \:\chi_{\Lambda}\big) = 0 = \tr\big( \chi_{\Lambda^c}\:\eta_\kappa(\Pi_-)\: \chi_{\Lambda^c}\big) \:. 
\] 
Moreover, if we assume that both
$ \chi_{\Lambda} \: \Pi_- \: \chi_{\Lambda}$ and~$\Pi_-  \: \chi_{\Lambda} \: \Pi_-$ are compact operators, we can find a one-to-one correspondence between their non-zero eigenvalues: Take any eigenvector~$\psi$ of~$\chi_{\Lambda} \: \Pi_- \: \chi_{\Lambda}$ with eigenvalue~$\lambda \neq 0$, then we must have
\[ \chi_{\Lambda} \psi = \frac{1}{\lambda} \:\chi_{\Lambda}^2 \: \Pi_- \: \chi_{\Lambda}\psi =\psi \qquad \text{and} \qquad \Pi_- \psi \neq 0 \:,\]
which yields
\[ \lambda \: \psi = \big( \chi_{\Lambda} \: \Pi_- \: \chi_{\Lambda} \big) \psi = \big( \chi_{\Lambda} \: \Pi_- \big) \psi  \:.\]
Then~$\Pi_- \psi $ is an eigenvector of~$\Pi_- \: \chi_{\Lambda} \: \Pi_-$ with eigenvalue~$\lambda$ because
\[ \big( \Pi_- \: \chi_{\Lambda} \: \Pi_- \big) (\Pi_- \psi) = \Pi_- \big( \chi_{\Lambda} \: \Pi_- \big) \psi  = \lambda\: \Pi_- \psi \:. \]
 Since the same argument also works with the roles of~$\Pi_- \: \chi_{\Lambda} \: \Pi_-$ and~$\chi_{\Lambda} \: \Pi_- \: \chi_{\Lambda}$ interchanged, this shows, that the nonzero eigenvalues of both operators (counted with multiplicities) coincide. Then the same holds true for~$\eta_\kappa(\Pi_- \: \chi_{\Lambda} \: \Pi_-)$ and~$\eta_\kappa(\chi_{\Lambda} \: \Pi_- \: \chi_{\Lambda} )$, proving that
\[ \tr\eta_\kappa(\chi_{\Lambda} \: \Pi_- \: \chi_{\Lambda} ) = \tr\eta_\kappa(\Pi_- \: \chi_{\Lambda} \: \Pi_-) \:. \]
Due to the symmetry of~$\eta_\kappa$, namely
\[\eta_\kappa(x)=\eta_\kappa(1-x) \Forany x \in \R \:,\]
this then leads to
\[ \tr\eta_\kappa(\chi_{\Lambda} \: \Pi_- \: \chi_{\Lambda} ) = \tr\eta_\kappa(\Pi_- \: \chi_{\Lambda} \: \Pi_-) =  \tr\eta_\kappa\big(\Pi_- - \Pi_- \: \chi_{\Lambda} \: \Pi_-\big) = \tr\eta_\kappa(\Pi_- \: \chi_{\Lambda^c} \: \Pi_-)   \:. \]
Repeating the same argument as before with~$\chi_{\Lambda^c} \: \Pi_- \: \chi_{\Lambda^c}$ finally gives
\begin{flalign*}
	&\tr\eta_\kappa(\chi_{\Lambda} \: \Pi_- \: \chi_{\Lambda} )  = \tr\eta_\kappa(\Pi_- \: \chi_{\Lambda^c} \: \Pi_-) =\tr \eta_\kappa(\chi_{\Lambda^c} \: \Pi_- \: \chi_{\Lambda^c})\:.
\end{flalign*}
Regularizing this expression suggests that the entanglement entropies of the inside and outside as defined in~\eqref{EntropicDiff} coincide. If this is the case, 
our definition of entanglement entropy agrees (up to a numerical factor) with that in~\cite{helling-leschke-spitzer, leschke-sobolev-spitzer2}.
\QEDrem
}}  \end{Remark} 

\section{Trace of the Limiting Operator} \label{Sec trace of the limiting operator}
In this section, we shall analyze
the operator~$\Opa(\mathfrak{a}_{0,1})$ in~\eqref{a01def}. Of course, the same methods apply to~$\Opa(\mathfrak{a}_{0,2})$.
	\begin{Notation} {\em{
		In the following it might happen that in the symbol~$\CA$ we can factor out a characteristic function in~$\xi$, i.e.\
		\[ \CA(u,u',\xi) = \chi_\Omega(\xi) \:\tilde{\CA}(u,u',\xi) \:. \]
		In this case, we will sometimes denote the characteristic function in~$\xi$ corresponding to the set~$\Omega$ by~$I_\Omega$ (this is to avoid confusion with the characteristic function~$\chi_{\mathcal{K}}$ in the variables~$u$ or~$u'$). }}
	\end{Notation}
	
	\begin{Remark} {\em{
		Note that the operator~$\Opa(\mathfrak{a}_{0,1})$ corresponds to
		\[
			\Opa(\mathfrak{a}_{0,1}) = \Four\: \mathfrak{a}_{0,1} \: \Four^{-1}\:,
		\]
		and is therefore well-defined on all of~$L^2(\R)$ and by Remark~\ref{ExtScwartzFcts} its integral representation extends to all Schwartz functions. Moreover, for any bounded subset~$U \subseteq \R$ the integral representation of the operator~$\chi_U \Opa(\mathfrak{a}_{0,1})\chi_U$ holds on all of~$L^2(\R)$ due to Lemma~\ref{Simple Integral Rep.}. \QEDrem }}
	\end{Remark}
	
\subsection{Idea for Smooth Functions}
\label{subsec:IdeaSmooth}
	The general idea is to make use of the following  one-dimensional result by Widom~\cite{widom1}.
	
	\begin{Thm}
		\label{ThmWidom}
		Let~$K, J\subseteq \R$ intervals, $f\in C^{\infty}(\R)$ be a smooth function with~$f(0)=0$ and~$a \in C^\infty (\R^2)$ a complex-valued
		Schwartz function which we identify with the symbol~$a(x,y,\xi)\equiv a(x,\xi)$ for any~$x,y,\xi \in \R$. Moreover, for any symbol~$b$ we denote its symmetric localization by
\beq \label{Abdef}
A(b) := \frac{1}{2}\Big(\chi_K  \,\Opa \big( I_J \,b \big)\, \chi_K+ \big(\chi_K  \Opa \big(I_J \,b \big) \chi_K \big)^*\Big) \:,
\eeq
		(recall that~$I_J$ is the characteristic function corresponding to~$J\subseteq \R$ with respect to the variable~$\xi$). Then
		\begin{flalign*}
			\tr \Big( f(A(a)) - \chi_K \:\Opa \big( I_J \,f(a) \big)\: \chi_K \Big) = \frac{1}{4\pi^2}\log(\alpha) \sum_i U \big( a(v_i); f \big) + \mathcal{O}(1)\:,
		\end{flalign*}
		where~$v_i$ are the vertices of~$K\times J$ (see Figure~\ref{fig:vertices})
		\begin{figure}
			\subfloat[Vertices if~$K$ and~$J$ are both finite: There are in total four vertices.]{\includegraphics[width=0.47\textwidth, trim=80pt 10pt 150pt 5pt,clip]{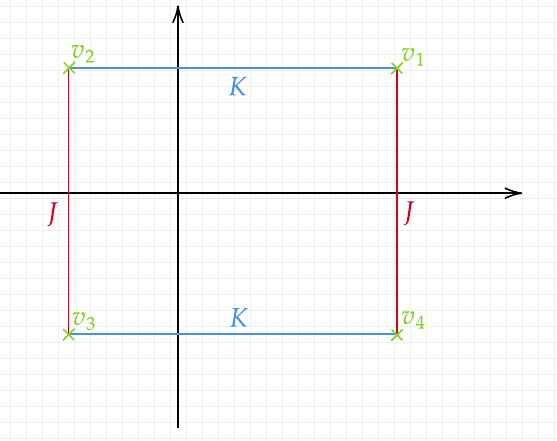}}\qquad
			\subfloat[Vertices if~$K$ is finite and~$J$ is infinite: There are only two vertices.]{\includegraphics[width=0.43\textwidth, trim=140pt 10pt 115pt 10pt,clip]{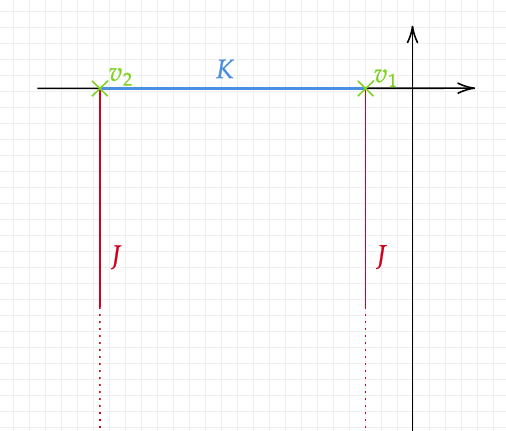}}
			\caption{Illustration with examples of the ``vertices" in Theorem~\ref{ThmWidom}.}
			\label{fig:vertices}
		\end{figure}
		and 
		\[ U(c;f):= \int_0^1\frac{1}{t(1-t)} \:\Big( f(tc) - tf(c) \Big)\: dt \Forany c \in \R \:. \]
	\end{Thm}
	\begin{Remark} {\em{
		\bitem
			\item[{\rm{(i)}}] To be precise, Widom considered operators with kernels 
			\[ \frac{\alpha}{2\pi}\int dy \int d\xi \:e^{\mathbf{+}i\alpha \xi (x-y)} \:a(x,\xi) \]
			but the results can clearly be transferred using the transformation~$\xi \rightarrow -\xi$.
			\item [{\rm{(ii)}}] Moreover, Widom considers operators of the form~$\Opa(a)$ whose integral representation extends to all of~$L^2(K)$. We note that, in view of Lemma~\ref{Simple Integral Rep.}, this assumption holds for any operator~$\Opa(a)$ with Schwartz symbol~$a=a(x,\xi)$, even if, a-priori, the integral representation holds only when inserting smooth compactly supported functions. \QEDrem
		\eitem }}
	\end{Remark}

	We want to apply the above theorem with~$J=(-\infty,0)=:\mathcal{J}$ and~$K=\mathcal{K}=(u_0-\rho,u_0)$, 
where we choose~$f$ as a suitable approximation of the function~$\eta_\kappa$ (again with~$f(0)=0$)
and~$a$ as
an approximation of the diagonal matrix entries~$\mathfrak{a}_{0,1\!/\!2}$ in~\eqref{LimOpKernel}
and~\eqref{a01def}. For ease of notation, we only consider~$a\approx \mathfrak{a}_{0,1}$,
noting that our methods apply similarly to~$\mathfrak{a}_{0,2}$.
To be more precise, we first introduce the smooth non-negative cutoff functions~$\Psi, \Phi \in C^{\infty}(\R)$ with
	\begin{flalign}
	\label{Psi Phi cutoff Def}
	\Psi(\xi) =\begin{cases}
		1\:,\quad \xi \leq 0\\
		0\:,\quad \xi > 1
	\end{cases} \qquad \text{and} \qquad
	\Phi(u)=\begin{cases}
		1\:, \quad u\in [-\rho,0]\\
		0\:, \quad u\notin (-\rho-1,1)\:,
	\end{cases}
\end{flalign}
and set~$\Phi_{u_0}(x):= \Phi(x-u_0)$. Then we may introduce~$a$ as the function
	\begin{flalign}
	\label{aDef}
	\mathfrak{a}(u,\xi):=\Psi(\xi)\: \Phi_{u_0}(u)\: e^{M\xi}\:.
\end{flalign}
For a plot of~$a$ see Figure~\ref{fig:aplot}.
	\begin{figure}
		\centering
		\includegraphics[width=0.7\linewidth]{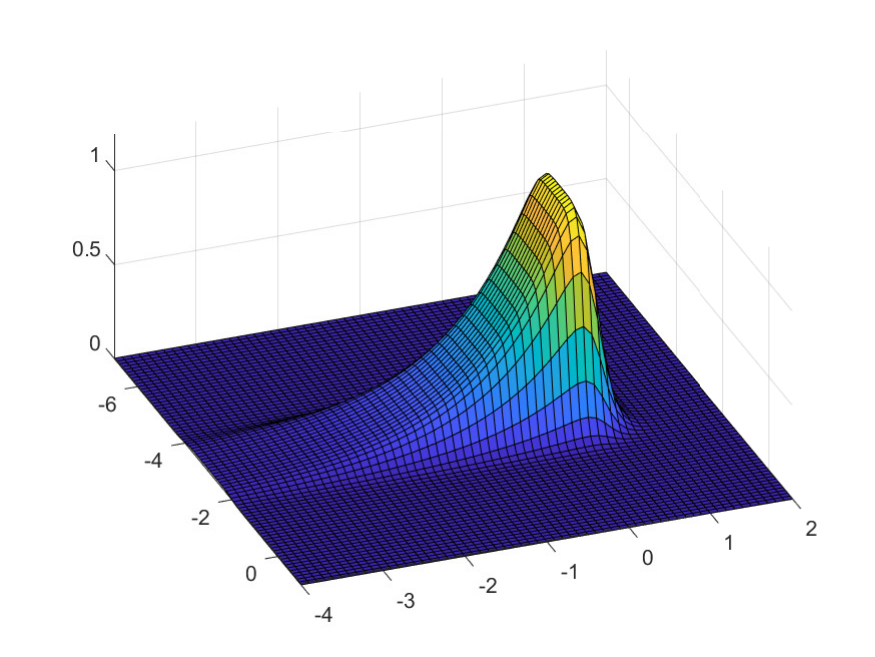}
		\caption{Plot of the function~$a$ in~\eqref{aDef}.}
		\label{fig:aplot}
	\end{figure}%
	Note that then~$\mathfrak{a}$ is a Schwartz function and
	\begin{align}
		\label{eq:conn.a/fraka}
		 \chi_K\:\Opa(I_J \,\mathfrak{a})\:\chi_K = \chi_K\:\mathrm{Op}_{\alpha}(\mathfrak{a}_{0,1})\:\chi_K \:.
	\end{align}
Moreover, the resulting symbol clearly fulfills the condition of Lemma~\ref{Simple Integral Rep.}, so we can extend the corresponding integral representation to all $L^2(\R,\C)$-functions. In addition, the operator is self-adjoint,
because of Lemma~\ref{fOpa}. This implies that we can leave out the symmetrization in~\eqref{Abdef}, i.e.\
\begin{align}
	\label{eq:SymmFrakA}
	A(\mathfrak{a}) = \chi_K \:\mathrm{Op}_{\alpha}(I_J\:\mathfrak{a})\:\chi_K = \chi_K \:\mathrm{Op}_{\alpha}(\afrak_{0,1})\:\chi_K \:.
\end{align}
	Furthermore, due to Lemma~\ref{fOpa}, we may pull out any function~$f$ as in the above Theorem~\ref{ThmWidom} in the sense that
	\[\chi_K\: \Opa \big(I_J \:f(\mathfrak{a}) \big) \:\chi_K = \chi_K \:\Opa \big( f(\afrak_{0,1}) \big)\: \chi_K = \chi_K \:f(\Opa(\afrak_{0,1})) \:\chi_K \:,\]
	where we used that~$\afrak_{0,1}$ vanishes
	outside~$J$ and that~$f(0)=0$.
	
	In our application the vertices of~$K \times J$ are (similar as in Figure~\ref{fig:vertices} (B)) given by 
	\[  v_1 = (u_0, 0) \qquad \text{and} \qquad v_2 = (u_0-\rho, 0) \:,\]
	and thus
	\[ \mathfrak{a}(v_i) = 1\:, \qquad \text{for any } i=1,2\:, \]
	leading to
	\begin{flalign}
	\label{D(afrak,f) smooth}
		\tr\Big( f \big( \chi_K\Opa(\afrak_{0,1})\chi_K \big) - \chi_K f \big(\Opa(\afrak_{0,1})\big) \chi_K \Big) = \frac{1}{2\pi^2}\:\log(\alpha)\:U(1;f) + \mathcal{O}(1)\:,
	\end{flalign}
	valid for any~$f\in C^{\infty}(\R)$ with~$f(0)=0$.
	
	Note that, using Lemma~\ref{Translation Lemma} and the fact that~$\afrak_{0,1}$ does not depend on~$u$ or~$u'$, the $\mathcal{O}(1)$-term does not change when varying~$u_0$, and therefore the result stays same when we take the limit~$u_0 \rightarrow -\infty$.
We need to keep this in mind because we shall take the limit~$u_0 \rightarrow -\infty$ before the limit~$\alpha \rightarrow \infty$ (cf.~\eqref{def modewise entropy}).

\subsection{Proof for Non-Differentiable Functions}
	In order to state the main result of this section we first need to introduce the following condition.
	\begin{cond}
		\label{cond:f3}
		Let~$T:=\{t_0, \dots, t_l \}$ be a finite set and~$g\in C^2(\R \setminus T)\cap C^0(\R)$ be a function such that there exists a constant~$\gamma>0$ and in the neighborhood of every~$t_i$ there are constants~$c_k>0$, $k=0,1,2$ satisfying the conditions
		\begin{align}
			|g^{(t)}(x)| \leq c_k \:|t-t_i|^{\gamma -k}\:.
		\end{align}
	\end{cond}
	\begin{Example}
	\label{Ex:CondFRenyi}
	{\em{As shown in detail in Lemma~\ref{Properties eta}, the functions~$\eta_{\kappa}$ satisfy Condition~\ref{cond:f3} with~$T=\{0,1\}$ and if~$\kappa\neq1$ for any~$\gamma \leq \min\{1, \kappa\}$ if~$\kappa=1$ we may take any~$\gamma <1$.
	
\QEDrem}}
\end{Example}
	
	The following theorem constitutes the main result of this section.
	\begin{Thm}
	\label{MainThmSec2}
		Let~$\mathcal{K} =(u_0-\rho, u_0)$, $\mathcal{J}=(-\infty,0 )$ (as in Section~\ref{subsec:IdeaSmooth}) and~$\afrak_{0,1}(\xi)=e^{M\xi}\: \chi_{(-\infty, 0)}(\xi)$ as in~\eqref{a01def}.
Moreover, let~$g\in C^2(\R \setminus \{t_0, \dots, t_l \})\cap C^0(\R)$ satisfy Condition~\ref{cond:f3} with~$g(0)=0$.
		Then
\beq \label{thmeq}
\lim\limits_{\alpha \rightarrow \infty} \lim\limits_{u_0 \rightarrow -\infty} \frac{1}{\log \alpha} \tr D_\alpha(g,\mathcal{K},\afrak_{0,1}) = \frac{1}{2 \pi^2}\: U(1;g) \:.
\eeq
	\end{Thm}

In the proof of Theorem~\ref{MainThmSec2}, we will apply Lemmata~\ref{TechLemmaNonSmt} and~\ref{TechLemmaNonDiff}. In order to complete the error estimates, one finally needs to control the term~$\| |PA(\id-P)|^\sigma\|_\SN$. This can be done with the following lemma.
	
	\begin{Lemma}
		\label{FinalEstq-Norm}
		Let~$u_0 \in \R$ arbitrary and~${\mathcal{K}}=(u_0-\rho,u_0)$.
		Choose numbers~$q\in (0,1]$, $\alpha \geq 3$
		and~$\rho \geq 2$. Then the symbol~$\mathfrak{a}_{0,1}$ from~\eqref{a01def} satisfies
		\[ \big\| \chi_{\mathcal{K}} \: \Opa(\mathfrak{a}_{0,1}) \: (1-\chi_{\mathcal{K}}) \big\|^q_q
		\lesssim \log \alpha \:.\]
		with implicit constants independent of~$\alpha$ and~$u_0$.
	\end{Lemma}
	\begin{proof}
		First of all make use Lemma~\ref{Translation Lemma} in order to replace the region~${\mathcal{K}}$ by~${\mathcal{K}}_0:=(-\rho,0)$:
		\[
		\| \chi_{\mathcal{K}} \:\Opa(\mathfrak{a}_{0,1})\: (1-\chi_{\mathcal{K}}) \|_q^q = \| \chi_{{\mathcal{K}}_0} \:\Opa(\mathfrak{a}_{0,1})\: (1-\chi_{{\mathcal{K}}_0}) \|_q^q \:.
		\]
		Next, let~$(\Psi_j)_{j\in \Z}$ be a partition of unity with
		~$\Psi_{j}(x)=\Psi_{0}(x-j)$ for all~$j\in \Z$ and~$\supp \Psi_0 \subseteq (-\frac{1}{2},\frac{3}{2})$.
		For any~$j \in \Z$ we consider the symbols
		\begin{align*}
			\mathfrak{a}_j(\xi) &:= \Psi_{j}(\xi) \:e^{M\xi}\:,
		\end{align*}
		Using the notation~$\mathcal{J}_j:=(j-1,j)$ for any~$j \in \Z_{\leq 0}$ we obtain with the help of Lemma~\ref{Opa Mult Rule} together with Remark~\ref{remark56},
		\[
		\chi_{{\mathcal{K}}_0} \:\Opa(I_{\mathcal{J}_j}\mathfrak{a}_{0,1})\: (1-\chi_{{\mathcal{K}}_0}) = \chi_{{\mathcal{K}}_0} \:\Opa(I_{\mathcal{J}_j}\mathfrak{a}_j)\: (1-\chi_{{\mathcal{K}}_0}) = \chi_{{\mathcal{K}}_0} \:P_{\mathcal{J}_j,\alpha} \:\Opa(\mathfrak{a}_{j})\:  (1-\chi_{{\mathcal{K}}_0})\:,
		\]
		so with the triangle inequality~\eqref{est:triangle} we conclude that
		\begin{align}
			\label{Proof610Step1}
			\| \chi_{{\mathcal{K}}_0} \:\Opa(\mathfrak{a}_{0,1})\: (1-\chi_{{\mathcal{K}}_0}) \|_q^q \leq \sum_{j\in \Z_{\leq 0}} \| \chi_{{\mathcal{K}}_0} \:P_{\mathcal{J}_j,\alpha}\:\Opa(\mathfrak{a}_{j})\:  (1-\chi_{{\mathcal{K}}_0}) \|_q^q\:.
		\end{align}
		In the next step we want to interchange~$P_{\mathcal{J}_j,\alpha}$ and~$\chi_{{\mathcal{K}}_0}$. To this end note that
		\begin{align*}
			&\| [P_{\mathcal{J}_j,\alpha} ,\chi_{{\mathcal{K}}_0} ]\|_q^q = 
			\| 	P_{\mathcal{J}_j,\alpha} \chi_{{\mathcal{K}}_0} -\chi_{{\mathcal{K}}_0}P_{\mathcal{J}_j,\alpha} \chi_{{\mathcal{K}}_0} + \chi_{{\mathcal{K}}_0}P_{\mathcal{J}_j,\alpha} \chi_{{\mathcal{K}}_0} - \chi_{{\mathcal{K}}_0}P_{\mathcal{J}_j,\alpha}\|_q^q\\
			&=
			\|(1-\chi_{{\mathcal{K}}_0}) P_{\mathcal{J}_j,\alpha} \chi_{{\mathcal{K}}_0} -  \chi_{{\mathcal{K}}_0}P_{\mathcal{J}_j,\alpha} (1-\chi_{{\mathcal{K}}_0})  \|_q^q \leq 2 \| (1-\chi_{{\mathcal{K}}_0}) P_{\mathcal{J}_j,\alpha}\chi_{{\mathcal{K}}_0} \|_q^q
		\end{align*}
		where we also used that
		\[ \big(\chi_{{\mathcal{K}}_0}\:P_{\mathcal{J}_j,\alpha}\:(1-\chi_{{\mathcal{K}}_0})\big)^*= (1-\chi_{{\mathcal{K}}_0}) \:P_{\mathcal{J}_j,\alpha}\:\chi_{{\mathcal{K}}_0}\:, \]
		together with the fact that singular values (and therefore the $q$-norm) are invariant under Hermitian 
		conjugation.\footnote{Note that for any~$A \in \SN_q$ we can write~$\|A\|_q=(\sum_i \sqrt{\lambda_i}^q)^{1/q}$ where the~$\lambda_i$ are the eigenvalues of~$A^*A$ (note that since~$A$ is compact, there are countably many). Moreover, for any eigenvector~$\psi$ of~$A^*A$ corresponding to a non-zero eigenvalue, the vector~$A\psi$ is an eigenvector of~$AA^*$ with the same eigenvalue, so we see that~$\|A^*\|_q\leq\|A\|_q$ (as~$\psi$ might lie in the kernel of~$A$). Then, symmetry yields the equality.}
		Moreover, using Remark~\ref{Rem:XiTrans} together with Corollary~\ref{Cor4.7} we conclude that for any~$j\in \Z_{\leq 0}$:
		\[
		\| (1-\chi_{{\mathcal{K}}_0}) P_{\mathcal{J}_j,\alpha}\chi_{{\mathcal{K}}_0} \|_q^q = \| (1-\chi_{{\mathcal{K}}_0}) \Opa(I_{\mathcal{J}_j}) \chi_{{\mathcal{K}}_0} \|_q^q = \| (1-\chi_{{\mathcal{K}}_0}) \Opa(I_{\mathcal{J}_0}) \chi_{{\mathcal{K}}_0} \|_q^q \lesssim \log\alpha\:,
		\]
		with implicit constant independent of~$j\in \Z_{\leq 0}$, $\alpha\geq 2$ and~$u_0$.
		Moreover making use of Lemma~\ref{OpaBounded} together with Remark~\ref{ExtScwartzFcts} and the fact that
		\[
		N^{(1,1,2)}(\mathfrak{a}_j;1,1) \lesssim e^{Mj}\:,
		\]
		with an implicit constant independent of~$j$ we obtain for any~$\alpha \geq 1$,
		\[
		\| \Opa(\mathfrak{a}_{j})\|^q_{\infty} \lesssim e^{qMj}\:,
		\]
		again with an implicit constant independent of~$j$ and~$\alpha$.
		Using the H\"older-type inequality~\eqref{est:HoelderLike}, this allows us to estimate
		\begin{align*}
			&\| \chi_{{\mathcal{K}}_0} \:P_{\mathcal{J}_j,\alpha}\:\Opa(\mathfrak{a}_{j})\:  (1-\chi_{{\mathcal{K}}_0}) \|_q^q \leq \| [\chi_{{\mathcal{K}}_0}, \:P_{\mathcal{J}_j,\alpha}]\:\Opa(\mathfrak{a}_{j})\:  (1-\chi_{{\mathcal{K}}_0}) \|_q^q \\
			&+ \| P_{\mathcal{J}_j,\alpha} \:\chi_{{\mathcal{K}}_0} \:\Opa(\mathfrak{a}_{j})\:  (1-\chi_{{\mathcal{K}}_0}) \|_q^q \leq \| [\chi_{{\mathcal{K}}_0}, \:P_{\mathcal{J}_j,\alpha}] \|_q^q \| \Opa(\mathfrak{a}_{j})\:  (1-\chi_{{\mathcal{K}}_0})\|^q_{\infty} \\
			&+ \| P_{\mathcal{J}_j,\alpha}\|^q_{\infty}\| \chi_{{\mathcal{K}}_0} \:\Opa(\mathfrak{a}_{j})\:  (1-\chi_{{\mathcal{K}}_0}) \|_q^q\lesssim e^{qMj}\log\alpha + \| \chi_{{\mathcal{K}}_0} \:\Opa(\mathfrak{a}_{j})\:  (1-\chi_{{\mathcal{K}}_0}) \|_q^q\:.
		\end{align*}
		with an implicit constant independent of~$j\in \N_0$ and~$\alpha \geq 2$. 
		Thus it remains to estimate the term 
		$\| \chi_{{\mathcal{K}}_0} \:\Opa(\mathfrak{a}_{j})\:  (1-\chi_{{\mathcal{K}}_0}) \|_q^q$. To this end we want to apply Proposition~\ref{prop:cross_smooth} to~$\mathfrak{a}_j$. So choose~$\tau = 2$ and~$\bmu =j-1/2$, then
		\[
		N^{(\tilde{m})}(\mathfrak{a}_j;\tau) \lesssim e^{Mj}\:, 
		\]
		with an implicit constant independent of~$j$. This yields
		\begin{align}
			\label{Proof610Step2}
			\begin{split}
				\| \chi_{{\mathcal{K}}_0} \:P_{\mathcal{J}_j,\alpha}\:\Opa(\mathfrak{a}_{j})\:  (1-\chi_{{\mathcal{K}}_0}) \|_q^q &\leq \| [\chi_{{\mathcal{K}}_0}, \:P_{\mathcal{J}_j,\alpha}]\:\Opa(\mathfrak{a}_{j})\:  (1-\chi_{{\mathcal{K}}_0}) \|_q^q \\
				&\lesssim  e^{qMj} \log \alpha \:,
			\end{split}
		\end{align}
		with an implicit constant independent of~$j$ and~$\alpha\geq 3$.
		Then, summarizing~\eqref{Proof610Step1} and~\eqref{Proof610Step2} yields
		\begin{align*}
			\| \chi_{\mathcal{K}} \:\Opa(\mathfrak{a}_{0,1})\: (1-\chi_{\mathcal{K}}) \|_q^q &\leq \sum_{j\in \Z_{\leq 0}} \| \chi_{{\mathcal{K}}_0} \:P_{\mathcal{J}_j,\alpha}\:\Opa(\mathfrak{a}_{j})\:  (1-\chi_{{\mathcal{K}}_0}) \|_q^q\\
			& \lesssim \sum_{j=0}^{\infty} e^{-qMj}\:\log \alpha \lesssim \log\alpha\:,
		\end{align*}
		with an implicit constant independent of~$\alpha\geq 3$ and~$u_0$.
	\end{proof}

In the proof of Theorem~\ref{MainThmSec2} we will also make use of the following continuity result for~$U(1;f)$.
\begin{Lemma}
	\label{lem:contU(1;g)}
	Let~$f$ be  a function on~$[0,1]$ with~$f(0)=0$.
	\bitem
	\item[{\rm{(i)}}] If~$f\in C^2([0,1])$ denote
	\[
	\|f\|_{C^2} := \max_{0 \leq k \leq 2}\max_{t \in [0,1]} \big| f^{(k)}(t) \big| \:.
	\]
	Then,
	\[
	|	U(1;f) | \leq \frac{9}{2}  \|f\|_{C^2}\:.
	\]
	\item[{\rm{(ii)}}] If~$f$ satisfies Condition~\ref{cond:f3} with~$X=\{z\}$ where~$z=0$ or~$z=1$ and is supported in~$[z-R,z+R]$ for some~$R<\frac{1}{2}$, then
	\[
	| U(1;f) | \leq\bl f \bl_2 \frac{R^\gamma}{\gamma(1-R)} \:.
	\]
	\eitem
\end{Lemma}
\begin{proof}
	First split the integral in the definition of~$U(1;f)$ as follows:
	\begin{flalign*}
		&\big| U(1;f) \big| \leq \underbrace{\left| \int_0^{1/2} \frac{1}{t(1-t)}\big(f(t)-tf(1)\big)dt \right|}_{=:\text{(I)}} +
		\underbrace{ \left| \int_{1/2}^1 \frac{1}{t(1-t)}\big(f(t)-tf(1)\big)dt \right| }_{=:\text{(II)}}\:.
	\end{flalign*}
	\bitem
	\item[(i)]
	For the estimate of~$\text{(I)}$ consider the Taylor expansion for~$f$ around~$t=0$ keeping in mind that~$f(0)=0$:
	\[ f(t) = tf'(0)+ \frac{t^2}{2} f''\big(\tilde{t} \big) \qquad \text{for suitable~$\tilde{t}\in[0,t]$} \:,\]
	and therefore
	\begin{flalign*}
		\text{(I)} &\leq \int_{0}^{1/2} \underbrace{\Big| \frac{1}{(1-t)}\Big|}_{\leq 2}\Big( \underbrace{|f'(0)|}_{\leq \|f\|_{C^2}}+\underbrace{|t/2|}_{\leq 1/4} \underbrace{|f''(\tilde{t})|}_{\leq \|f\|_{C^2}}\Big) dt + \int_{0}^{1/2}\underbrace{\Big| \frac{1}{(1-t)}\Big|}_{\leq 2} \underbrace{|f(1)|}_{\leq \|f\|_{C^2}} dt\\
		&\leq \frac{9}{4}\: \|f\|_{C^2}
	\end{flalign*}
	(note that~$\tilde{t}$ is actually a function of~$t$, but this is unproblematic because~$f''$ is uniformly bounded).
	
	Similarly, for the estimate of~(II) we use the Taylor expansion of~$f$, but now around~$t=1$,
	\[ f(t) = f(1) + (t-1)f'(1)+\frac{(t-1)^2}{2} f''(\tilde{t}) \qquad \text{for suitable~$\tilde{t}\in[0,t]$} \:. \]
	We thus obtain
	\begin{flalign*}
		\text{(II)} &\leq \int_{1/2}^{1} \Big| \frac{1}{t(1-t)}\Big|\Big( |1-t|\:|f(1)|+ |1-t|\:|f'(1)|+\frac{|1-t|^2}{2} |f''(\tilde{t})|\Big) \:dt\\
		&\leq \int_{1/2}^{1} \underbrace{|1/t|}_{\leq 2}\Big( \underbrace{|f(1)|}_{\leq \|f\|_{C^2}}+ \underbrace{|f'(1)|}_{\leq \|f\|_{C^2}}+\underbrace{\frac{|1-t|}{2}}_{\leq 1/4} \underbrace{|f''(\tilde{t})|}_{\leq \|f\|_{C^2}}\Big)\: dt \leq \frac{9}{4} \|f\|_{C^2} \:.
	\end{flalign*}
	\item[(ii)] \begin{itemize}
		\item[(a)] Case~$z=0$:
		First note that
		\[ |f(t)| \leq \bl f \bl_2 |t|^\gamma  \qquad \text{for any } t \in (0,1/2) \:. \]
		This yields for~$R< 1/2$,
		\begin{flalign*}
			|\text{(I)}| &= \Big|\int_{0}^{1/2}\frac{1}{t(1-t)} f(t)dt\Big| \leq \int_{0}^{R}\underbrace{\frac{1}{|1-t|}}_{\leq 1/(1-R)}\bl f\bl_2 |t|^{\gamma -1} dt \\
			&\leq \frac{1}{1-R} \bl f\bl_2 \underbrace{\int_{0}^{R} |t|^{\gamma -1}}_{=R^\gamma/\gamma} dt
			\leq \bl f\bl_2 \frac{R^\gamma}{\gamma(1-R)}\:,
		\end{flalign*}
		Moreover, the integral~(II) vanishes for~$R<1/2$.
		\item[(b)] Case~$z=1$:
		Similarly as in the previous case, we now have
		\[ \text{(I)} = 0 \qquad \text{for~$R<1/2$}\:.\]
		Moreover, just as in the previous case, one can estimate
		\[ |f(t)|\leq \bl f \bl_2 |1-t|^\gamma \qquad \text{for any } t \in (1/2,1) \:.  \]
		This yields for~$R< 1/2$,
		\begin{flalign*}
			|\text{(ii)}| \leq  \bl f \bl_2 \int_{1-R}^{1}\frac{1}{|t|}\: |1-t|^{1-\gamma}\: dt \leq \bl f \bl_2 \:\frac{R^\gamma}{\gamma(1-R)} \:.
		\end{flalign*}
	\end{itemize}
	\eitem
\end{proof}

Now we have all the tools to prove Theorem~\ref{MainThmSec2}.
	\begin{proof}[Proof of Theorem~\ref{MainThmSec2}]
	\label{Proof of Main Res}
Before beginning, we note that the $u_0$-limit in~\eqref{thmeq} may be disregarded,
because the symbol is translation invariant in position space
(see Lemma~\ref{Translation Lemma}, noting that~$\afrak_{0,1}$ does not depend on~$u\equiv \bx$ or~$u'\equiv \by$).

		The remainder of the proof is based on the idea of the proof of \cite[Theorem 4.4]{sobolev-functions}
		Let~$\mathfrak{a}$ be the symbol in~\eqref{aDef}. By Lemma~\ref{OpaBounded}, we can assume that the operator norm of~$\Opa (\mathfrak{a})$ is uniformly bounded in~$\alpha$.
We want to apply Lemma~\ref{Opa Mult Rule} with~$\CA=\mathfrak{a}$ and~$\CB=I_\mathcal{J}$ (recall that~$\mathcal{J}=(-\infty,0)$). In order to verify the conditions of this lemma, we first note that,
Remark~\ref{remark56}~(i) yields condition~(ii), whereas condition~(i) follows
from the estimate
		\begin{flalign*}
			&\int du \:\Big| \int d\xi \:e^{-i\xi u} \:e^{M\xi} \:\psi(\xi)\: \Psi(\xi)\: \Phi_{u_0}(u)\Big|^2 \\
			&\leq \int_{u_0-\rho-1}^{u_0+1} du\: \Big( \int d\xi\: \chi_{(-\infty,1)}(\xi) \:e^{M\xi}\: |\psi(\xi)| \Big)^2\\
			 &\leq (\rho+2) \: \|\chi_{(-\infty,1)}\: e^{\cdot/M}\|_{L^2(\R,\C)}^2\: \|\psi\|_{L^2(\R,\C)}^2 \:,
		\end{flalign*}
(which holds for any~$\psi \in L^2(\R)$). Now Lemma~\ref{Opa Mult Rule} yields
		\[ \Opa(I_\mathcal{J} \,\mathfrak{a})= \Opa(\mathfrak{a})\: P_{\alpha,\mathcal{J}} \:.\]
		Since~$P_{\alpha,\mathcal{J}}$ is a projection operator, we see that~$\|\Opa(I_\mathcal{J} \,\mathfrak{a})\|_\infty
		\leq \|\Opa(\mathfrak{a})\|_\infty$ for all~$\alpha$.
		In particular, the operator~$\Opa(I_\mathcal{J} \,\mathfrak{a})$ is bounded uniformly in~$\alpha$.
		Hence,
		\[ \| A(\mathfrak{a}) \| = \big\| \chi_{\mathcal{K}} \Opa(\afrak_{0,1})\chi_{\mathcal{K}} \big\| = 
		\big\| \chi_{\mathcal{K}} \:\Opa(I_\mathcal{J} \,\mathfrak{a})\: \chi_{\mathcal{K}} \big\| \leq \|\Opa(\mathfrak{a})\|_\infty =: C_1 \:, \] 
		uniformly in~$\alpha$ (recall that~$A(\mathfrak{a})$ is the symmetric localization from Theorem~\ref{ThmWidom}). Moreover, the sup-norm of the symbol~$\afrak_{0,1}$ itself is bounded by a constant~$C_2$. We conclude that we only need to consider the function~$g$ 
		on the interval
		\[[-\max \{C_1,C_2\}, \max \{C_1,C_2\}]\:.\] 
		Therefore, we may assume that
		\[ \supp g \subseteq [-C,C] \qquad \text{with} \qquad  C:=\max \{C_1,C_2\}+1\:, \]
		possibly replacing~$g$ by the function
		\[ \tilde{g}=\Psi_C\: g \]
		with a smooth cutoff function~$\Psi_C\geq 0$ such that~$\Psi_C|_{[-C+1,C-1]} \equiv 1$  and~$\supp \Psi_C \subseteq [-C,C]$. For ease of notation, we will write~$g \equiv \tilde{g}$ in what follows.
		
		We remark that the function~$\eta_\kappa$ which we plan to consider later already  satisfies this property by definition with~$C=2$.
From Lemma~\ref{TechLemmaNonDiff} and Lemma~\ref{FinalEstq-Norm} we 
see that~$D_\alpha(g,{\mathcal{K}},\afrak_{0,1})$ is indeed trace class.
We now compute this trace, proceeding in two steps. \\[-0.5em]

\noindent \underline{First Step: Proof for~$g\in C^2(\R)$.}\\
To this end, we first apply the Weierstrass approximation theorem as given in \cite[Theorem 1.6.2]{Narasimhan} to obtain a polynomial~$g_\delta$ such that~$f_\delta:=g-g_\delta$ fulfills
\begin{align}
	\label{est:fDelta}
	\max_{0 \leq k \leq 2}\max_{|t|\leq C} \big| f_\delta^{(k)}(t) \big|\leq \delta \:.
\end{align}
Without loss of generality we can assume that~$f_\delta(0)=0$ (otherwise replace~$f_\delta$ by the function~$t\mapsto f_{\delta/2}(t)- f_{\delta/2}(0)$). In order to control the error of the polynomial approximation, we
		apply Lemma~\ref{TechLemmaNonSmt} with~$n=2$, $r=C$, some~$\sigma \in (0,1)$, $q=1$ and
		\[ A = \Opa(\afrak_{0,1})\:, \quad P=  \chi_{\mathcal{K}}\:, \quad g = \tilde{f}_\delta :=f_\delta\:\Psi_C \]
		(note that here~$g$ is the function in Lemma~\ref{TechLemmaNonSmt})
		where~$\Psi_C$ is the cutoff function from before (the cutoffs and approximation are visualized in Figure~\ref{fig:approxandcutoff}).
		\begin{figure}
			\centering
			\includegraphics[width=1.2\linewidth, trim=95pt 0pt 0pt 0pt,clip]{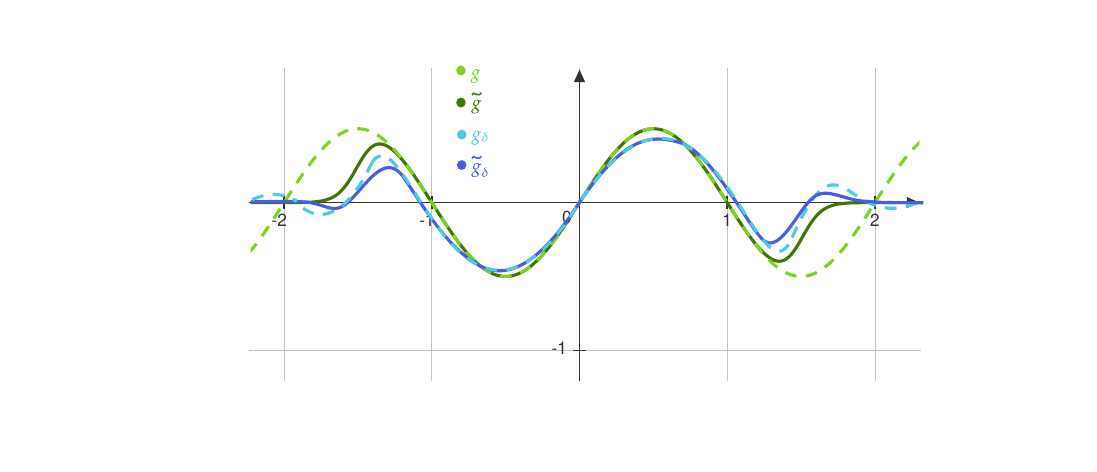}
			\caption{Visualization of the cutoffs and approximations in the first step of the proof of Theorem~\ref{MainThmSec2} for~$C=2$. We start with a function~$g$, which is first multiplied by the cutoff-function~$\Psi_C$, giving~$\tilde{g}$. This function is then approximated by a polynomial~$g_\delta$. Multiplying~$g_\delta$ by the cutoff function~$\Psi_C$ results in a function which is here called~$\tilde{g}_\delta$ (but does not directly appear in the proof). The function~$\tilde{f}_\delta$ is then given by the difference between~$\tilde{g}$ and~$\tilde{g}_\delta$.}
			\label{fig:approxandcutoff}
		\end{figure}
		This gives
		\begin{flalign*}
			&\big\| f_\delta \big( \chi_{\mathcal{K}}\: \Opa(\afrak_{0,1} \big)\: \chi_{\mathcal{K}} \big) - \chi_{\mathcal{K}} \:f_\delta \big( \Opa(\afrak_{0,1}) \big)\: \chi_{\mathcal{K}} \big\|_1 \\
		&= \big\| \tilde{f}_\delta\big( \chi_{\mathcal{K}}\: \Opa(\afrak_{0,1}) \:\chi_{\mathcal{K}} \big) 
		- \chi_{\mathcal{K}} \:\tilde{f}_\delta \big( \Opa(\afrak_{0,1}) \big)\: \chi_{\mathcal{K}} \big\|_1 \\
		&\lesssim \:\delta \:\big\| \chi_{\mathcal{K}} \:\Opa(\afrak_{0,1})\: (1-\chi_{\mathcal{K}}) \big\|_\sigma^\sigma\:.
		\end{flalign*}
		with an implicit constant independent of~$\delta$ and~$\alpha$.
Moreover, applying Lemma~\ref{FinalEstq-Norm} with~$q=\sigma$, we conclude that for~$\alpha$ large enough 
		\begin{flalign*}
&\big\| f_\delta \big(\chi_{\mathcal{K}}\: \Opa(\afrak_{0,1})\: \chi_{\mathcal{K}} \big) 
- \chi_{\mathcal{K}} \:f_\delta(\Opa(\afrak_{0,1}))\: \chi_{\mathcal{K}} \big\|_1 
\lesssim\delta\: \log \alpha\:,
		\end{flalign*}
		(again with an implicit constant independent of~$\delta$ and~$\alpha$).
		Using this inequality, we can estimate the trace by
		\begin{flalign*}
			\tr D_\alpha(g,{\mathcal{K}},\afrak_{0,1}) \leq& \tr D_\alpha(g_\delta,{\mathcal{K}},\afrak_{0,1}) + \| D_\alpha(f_\delta,{\mathcal{K}},\afrak_{0,1})\|_1 \\
			\leq& \tr D_\alpha(g_\delta,{\mathcal{K}},\afrak_{0,1})+ C_3\:\delta\: \log \alpha \:,
		\end{flalign*}
		with a constant~$C_3$ independent of~$\delta$ and~$\alpha$.
In order to compute the remaining trace, we can again apply
Theorem~\ref{ThmWidom} (exactly as in the example~\eqref{D(afrak,f) smooth}). This gives
		\[ 
		\tr D_\alpha(g_\delta,{\mathcal{K}},\afrak_{0,1}) = \frac{1}{2\pi^2}   \log (\alpha)\: U(1;g_\delta) + \mathcal{O}(1)\:,
		\]
		and thus
		\[ 
		\tr D_\alpha(g,{\mathcal{K}},\afrak_{0,1})  \leq \frac{1}{2\pi^2} \: \log (\alpha)\: U(1;g_\delta) + C_3\: \delta\: \log \alpha + \mathcal{O}(1)\:,
		\]
		which yields together with Lemma~\ref{lem:contU(1;g)},
		\begin{flalign}
			\label{LimsupEps}
			\limsup\limits_{\alpha \rightarrow \infty } \frac{1}{\log \alpha} \tr D_\alpha(g,{\mathcal{K}},\afrak_{0,1}) \leq \frac{1}{2\pi^2}U(1;g_\delta) + C_3\:\delta \:.
		\end{flalign}
		Moreover, applying Lemma~\ref{lem:contU(1;g)} to~$f_\delta$ we obtain due the~\eqref{est:fDelta}
		\[
		\lim\limits_{\delta \rightarrow 0} | U(1;g_\delta) - U(1;g)|  = \lim\limits_{\delta \rightarrow 0} | U(1;f_\delta)| = 0\:.
		\]
		Therefore, taking the limit~$\delta \rightarrow 0$ in~\eqref{LimsupEps} gives
		\[ \limsup\limits_{\alpha \rightarrow \infty } \frac{1}{\log \alpha} \tr D_\alpha(g,{\mathcal{K}},\afrak_{0,1}) \leq \frac{1}{2\pi^2}\:U(1;g) \:. \]
		Analogously, using
		\begin{flalign*}
			\frac{1}{2\pi^2}   \log(\alpha)\: U(1;g_\delta) + \mathcal{O}(1) &= \tr D_\alpha(g_\delta,{\mathcal{K}},\afrak_{0,1})\leq \tr D_\alpha(g,{\mathcal{K}},\afrak_{0,1}) + \| D_\alpha(f_\delta,{\mathcal{K}},\afrak_{0,1}) \|_1 \\
			&\leq \tr D_\alpha(g,{\mathcal{K}},\afrak_{0,1}) + C_3\: \delta\: \log \alpha\:,
		\end{flalign*}
		we obtain
		\[ 
		\liminf\limits_{\alpha \rightarrow \infty} \frac{1}{\log \alpha}\tr D_\alpha(g,{\mathcal{K}},\afrak_{0,1}) \geq \frac{1}{2\pi^2} \:U(1;g_\delta) + C_3\: \delta \:.
		\]
Now we can take the limit~$\delta \rightarrow 0$,
		\[ 
		\liminf\limits_{\alpha \rightarrow \infty} \frac{1}{\log \alpha}\tr D_\alpha(g,{\mathcal{K}},\afrak_{0,1}) \geq \frac{1}{2\pi^2} \:U(1;g) \:.
		\]
Combining the inequalities for the~$\limsup$ and~$\liminf$, we conclude that for any~$g\in C^2(\R)$,
\beq \label{firststep}
		\lim\limits_{\alpha \rightarrow \infty} \frac{1}{\log \alpha}\tr D_\alpha(g,{\mathcal{K}},\afrak_{0,1}) = \frac{1}{2\pi^2} \:U(1;g) \:.
\eeq
		
\noindent \underline{Second Step: Proof for~$g$ as in claim.} \\
By choosing a suitable partition of unity and making use of linearity,
it suffices to consider the case~$T=\{z\}$ meaning that~$g$ is non-differentiable only at one point~$z$.
		Next we decompose~$g$ into two parts with a cutoff function~$\xi \in C^\infty_0(\R)$ with the property that
		\begin{flalign*}
			\xi(t)=\begin{cases}
				1, \quad |t|\leq 1/2\\
				0, \quad |t|\geq 1
			\end{cases} \:.
		\end{flalign*}
		and writing
		\[g= g_R^{(1)}+g_R^{(2)}\:,\]
		with
		\begin{flalign*}
		\begin{alignedat}{2}
			g_R^{(1)}(t) &:= g(t) \:\xi\big( (t-z)/R \big)\quad &&\Rightarrow \quad \supp g_R^{(1)}\subseteq [z-R,z+R]\:,\\
			g_R^{(2)}(t) &:= g(t)-g_R^{(1)}(t) \quad &&\Rightarrow \quad \supp g_R^{(2)}\subseteq [-C,C]\:;
			\end{alignedat}
		\end{flalign*}
see also Figure~\ref{fig:g1}.
		\begin{figure}
			\centering
			\includegraphics[width=0.7\textwidth]{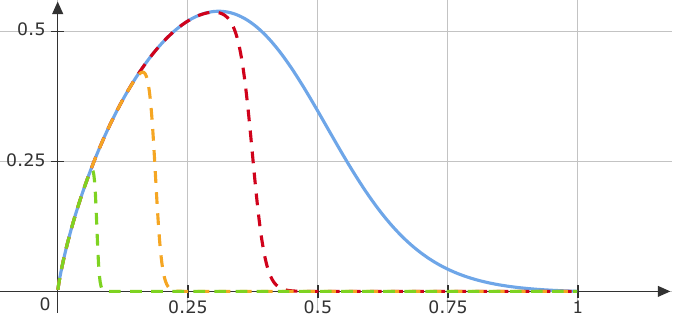}
			\caption{Schematic plot of a function~$g$ (blue) which is non-differentiable at~$z=0$ with the corresponding functions~$g^{(1)}_{1\!/\!2}$ (red), $g^{(1)}_{1/4}$ (orange) and~$g^{(1)}_{1/10}$ (green): They are cutting out the non-differentiable point.}
			\label{fig:g1}
		\end{figure}
		Note that the derivatives of~$g_R^{(1)}$ satisfy the bounds
		\[
		\Big(g_R^{(1)}\Big )^{(k)}(t) = \sum_{n=0}^{k} \:c(n,k) \:g^{(k-n)}(t)\:\xi^{(n)}\big((t-z)/R\big)\:\frac{1}{R^n}\:,
		\]
		(with some numerical constants~$c(n,k)$) and therefore the norm~$\bl . \bl_2$ in Lemma~\ref{TechLemmaNonDiff} can be estimated by
		\begin{flalign*}
			\bl g_R^{(1)}\bl_2 &=  \max_{0 \leq k \leq 2}\; \sup_{t \neq z} \bigg|\sum_{n=0}^{k}c(n,k) \:g^{(k-n)}(t)\:\xi^{(n)}\big((t-z)/R\big)\: \frac{1}{R^n}\bigg|\cdot |t-z|^{-\gamma+k} \\
			&\leq \max_{0 \leq k \leq 2}\; \sup_{t \neq z} \sum_{n=0}^{k} 
\big|c(n,k) \big|\: \big| g^{(k-n)}(t) \big| \:|t-z|^{-\gamma+k-n}\: \big| \xi^{(n)}\big((t-z)/R\big)\big|\:\frac{ |t-z|^n}{R^n}\:.
		\end{flalign*}
Noting that on the support of~$\big(g_R^{(1)}\big)^{(k)}$ we have
		\[ \frac{ |t-z|}{R} \leq 1\:, \]
we conclude that
\begin{align}
	\label{est:blUni}
	 \bl g_R^{(1)}\bl_2  \leq C_4 \bl g \bl_2 
\end{align}
		with~$C_4$ independent of~$R$ (also note that~$\bl g\bl_2$ is bounded by assumption). 
		
		For what follows, it is also useful to keep in mind that
		\[ 0=g_R^{(1)}(0)=g_R^{(2)}(0)\:.  \]	
		Now we apply~\eqref{firststep} to the function~$g_R^{(2)}$ (which clearly is in~$C^{2}(\R)$),
		\[ \lim\limits_{\alpha \rightarrow \infty} \frac{1}{\log \alpha} \tr D_\alpha(g_R^{(2)},{\mathcal{K}},\afrak_{0,1})  =U(1;g_R^{(2)})\:. \]
		Next, we apply Lemma~\ref{TechLemmaNonDiff} to~$g_R^{(1)}$ with~$A$ and~$P$ as before and some~$\sigma \in (0,1)$ with~$\sigma < \gamma$,
\[ \big\| D_\alpha \big( g_R{(1)},{\mathcal{K}},\afrak_{0,1} \big) \big\|_1 \leq  C_\sigma \bl g\bl_2 \:R^{\gamma -\sigma}
\:\big\| \chi_{\mathcal{K}} \:\Opa(\afrak_{0,1})\: (1-\chi_{\mathcal{K}}) \big\|_\sigma^\sigma \:. \]
Applying Lemma~\ref{FinalEstq-Norm} (for~$\alpha$ large enough) with~$q=\sigma$ yields
\[ \big\| D_\alpha \big( g_R{(1)},\mathcal{K},\afrak_{0,1} \big) \big\|_1 \leq C_5 \: \bl g\bl_2 \:R^{\gamma -\sigma} \log \alpha\:. \]
where the constant~$C_5$ is independent of~$R$ and~$\alpha$.
Just as before, it follows that
		\begin{flalign*}
			\limsup\limits_{\alpha \rightarrow \infty} \frac{1}{\log \alpha}\tr D_\alpha(g,{\mathcal{K}},\afrak_{0,1}) &\leq U \big( 1;g_R^{(2)} \big)+ C_5 \: \bl g\bl_2 \:R^{\gamma -\sigma}\\
			\liminf\limits_{\alpha \rightarrow \infty} \frac{1}{\log \alpha}\tr D_\alpha(g,{\mathcal{K}},\afrak_{0,1}) &\geq U\big( 1;g_R^{(2)} \big) + C_5 \: \bl g\bl_2 \:R^{\gamma -\sigma}\:.
		\end{flalign*}
The end result follows just as before by taking the limit~$R\rightarrow 0$, provided that we can show the convergence~$U(1;g_R^{(2)})\rightarrow U(1;g)$ for~$R\rightarrow 0$.
		To this end note that if~$z\notin \{0,1\}$ we have
		\[  |U(1;g_R^{(2)}) - U(1;g)| = |U(1; g_R^{(1)})| \leq C_6 R\, \]
		for some~$C_6 >0$ independent of~$R$ provided that~$R$ is sufficiently small (more precisely,
		so small that~$g_R^{(1)}$ vanishes in neighborhoods around~$0$ and~$1$; note that the integrand is supported in~$[z-R,z+R]$ and bounded uniformly in~$R$). These estimates show that~$\lim_{R \rightarrow 0}
		U(1;g-g_R^{(2)}) = 0$ in the case that~$z$ is neither~$0$ nor~$1$.
		In the remaining cases where~$z$ is either~$0$ or~$1$ we can apply Lemma~\ref{lem:contU(1;g)}, which also yields due to~\eqref{est:blUni},
		\[
			\lim\limits_{R \rightarrow 0} |U(1;g_R^{(2)}) - U(1;g)|  = \lim\limits_{R \rightarrow 0} |U(1; g_R^{(1)})| =  C_4 \bl g\bl_2\lim\limits_{R \rightarrow 0}\frac{R^\gamma}{\gamma (1-R)} = 0 \:.
		\]
This concludes the proof.
	\end{proof}
	
We finally apply Theorem~\ref{MainThmSec2} to the function~$\eta_\kappa$ and
the matrix- valued symbol~$\CA_0$ (see~\eqref{Def eta} and~\eqref{LimOpKernel}).
	
	\begin{Corollary}
	\label{Cor to Main Res Sec 2}
		For any~$\kappa>0$, $\eta_{\kappa}$, $\mathcal{K}$ and~$\CA_0$ as before,
		\[ \lim\limits_{\alpha \rightarrow \infty} \lim\limits_{u_0 \rightarrow \infty} \frac{1}{\log \alpha} \tr D_\alpha(\eta_\kappa,\mathcal{K},\CA_0 ) = \frac{1}{\pi^2}U(1;\eta_\kappa) = \frac{1}{6}\frac{\varkappa +1}{\varkappa}\:.\]
		Moreover, in the case that~$\kappa =1$, we can explicitly compute the coefficient~$U(1;\eta_1)$ to give
		\[\lim\limits_{\alpha \rightarrow \infty} \lim\limits_{u_0 \rightarrow \infty} \frac{1}{\log \alpha} \tr D_\alpha(\eta_1,\mathcal{K},\CA_0 ) = \frac{1}{\pi^2}U(1;\eta_1) = \frac{1}{3}\:. \]
	\end{Corollary}

	\begin{proof}
		As explained in Example~\ref{Ex:CondFRenyi}, the functions~$\eta_\kappa$ satisfy Condition~\ref{cond:f3} with~$n=2$ for any~$\kappa>0$. Moreover, we have~$\eta_\kappa(0)=0$ for any~$\kappa>0$. Therefore we can apply Theorem~\ref{MainThmSec2} and obtain
		\[ \lim\limits_{\alpha \rightarrow \infty} \lim\limits_{u_0 \rightarrow \infty} \frac{1}{\log \alpha} \tr D_\alpha(\eta_\kappa,\mathcal{K},\afrak_{0,1}) = \frac{1}{2 \pi^2}U(1;\eta_\kappa)\:. \]
		Repeating the procedure analogously for~$\afrak_{0,2}$ gives
		\[ \lim\limits_{\alpha \rightarrow \infty} \lim\limits_{u_0 \rightarrow \infty} \frac{1}{\log \alpha} \tr D_\alpha(\eta_\kappa,\mathcal{K},\afrak_{0,2})
		= \frac{1}{2 \pi^2}U(1;\eta_\kappa)\:, \]
		and therefore
		\[ \lim\limits_{\alpha \rightarrow \infty} \lim\limits_{u_0 \rightarrow \infty} \frac{1}{\log \alpha} \tr D_\alpha(\eta_\kappa,\mathcal{K},\CA_0 ) = \frac{1}{\pi^2}U(1;\eta_\kappa)\:. \]
		By \cite[Appendix]{leschke-sobolev-spitzer3},
		evaluating~$U(1;\eta_{\kappa})$ yields
		\[ U(1;\eta_\kappa) =  \int_0^1 \frac{\eta_\kappa(t)}{t(1-t)} \:dt =\frac{\pi^2}{6}\frac{\kappa +1}{\kappa} \:,\]
		and therefore
		\[ \lim\limits_{\varepsilon \searrow 0} \lim\limits_{u_0 \rightarrow \infty} \frac{1}{\ln (1/\varepsilon)} \tr D_{\alpha}(\eta_\kappa,\mathcal{K},\CA_0 ) = \frac{1}{\pi^2}U(1;\eta_\kappa) = \frac{1}{6}\frac{\kappa +1}{\kappa}\:. \]
		This concludes the proof.
	\end{proof}
	
Corollary~\ref{Cor to Main Res Sec 2} already looks quite similar to Theorem~\ref{Main Res.}.
The remaining task is to show equality in~\eqref{Limit Opa(b)}.
To this end, we need to show that all the correction terms drop out in the limits~$u_0 \rightarrow \infty$
	and~$\alpha \rightarrow \infty$.
The next section is devoted to this task.

\section{Estimating  the Error Terms} \label{Sec Estimating  the Error Terms}
In the previous section, we worked with the simplified kernel~\eqref{LimOpKernel}
and computed the corresponding entropy. In this section, we estimate all the errors,
thereby proving the equality in~\eqref{Limit Opa(b)}.
Our procedure is summarized as follows. Using~\eqref{FullOpKernel}, the
regularized projection operator~$(\Pi_-^\varepsilon)_{kn}$ can be written as
	\[
	(\Pi_-^\varepsilon)_{kn}= \Opa \big( 2\:(\CA^{(\varepsilon)}+\mathcal{R}_{0}^{(\varepsilon)}) \big)
	\]
	and~$\CA^{(\varepsilon)}$ as in~\eqref{Def Aeps} and the error term
	\[ \mathcal{R}_{0}^{(\varepsilon)}(u,u',\xi):=\mathcal{R}_{0}(u,u',M\xi/\varepsilon)\:, \]
	We denote the corresponding symbol by
	\[ \big( (\afull)_{ij} \big)_{1\leq i,j \leq 2}:=\Afull^{(\varepsilon)} := 2(\CA^{(\varepsilon)}+\mathcal{R}_{0}^{(\varepsilon)}) \:. \]
In preparation, translate~$\mathcal{K}$ to~$\mathcal{K}_0$ with the help of the unitary operator~$T_{u_0}$ making use of Lemma~\ref{Translation Lemma}.
 Moreover, we use that the operators~$(\Pi_-^\varepsilon)_{kn}$ and~$\Opa(\CA_0)$ are self-adjoint.
 We thus obtain
	\[  D_\alpha(\eta_\kappa, \mathcal{K}, \Afull^{(\varepsilon)}) - D_\alpha(\eta_\kappa, \mathcal{K}, \CA_0) =  D_\alpha(\eta_\kappa, \mathcal{K}_0, T_{u_0}(\Afull^{(\varepsilon)})) - D_\alpha(\eta_\kappa, \mathcal{K}_0, T_{u_0}(\CA_0)) \:, \]
	where~$\CA_0$ is the kernel of the limiting operator from~\eqref{LimOpKernel}.  Note that~$T_{u_0}(\CA_0) = \CA_0$ since the symbol~$\CA_0$ only depends on~$\xi$.
	Now we can estimate
	\begin{flalign}
		&\|D_\alpha(\eta_\kappa, \mathcal{K}_0, T_{u_0}(\Afull^{(\varepsilon)})) - D_\alpha(\eta_\kappa, \mathcal{K}_0, \CA_0) \|_{1} \notag \\
		&\leq \Big\| \eta_\kappa\big(\chi_{\mathcal{K}_0}\: \Opa(T_{u_0}(\Afull^{(\varepsilon)}))\: \chi_{\mathcal{K}_0}\big) -  \eta_\kappa\big(\chi_{\mathcal{K}_0}\: \Opa(\CA_0))\: \chi_{\mathcal{K}_0} \big) \Big \|_1 \label{TermI}
		\tag{I} \\
		&\quad\:+  \Big\| \chi_{\mathcal{K}_0}\: \big( \eta_\kappa\left(\Opa\big(T_{u_0}(\Afull^{(\varepsilon)})\big)\right)- \eta_\kappa\big(\Opa(\CA_0)\big)\big) \:\chi_{\mathcal{K}_0} \Big\|_1 \:. \tag{II} \label{TermII}
	\end{flalign}
	In the following we will estimate the expressions (I) and (II) separately.
	
\subsection{Estimate of the Error Term (I)}
We begin with the following simplified result from \cite{sobolev-functions}, which is related to Lemma~\ref{TechLemmaNonDiff}.
	\begin{Thm}\cite[Condition 2.3 and Theorem 2.4 (simplified to our needs)]{sobolev-functions}
		Let~$f$ satisfy Condition~\ref{cond:f4} with some~$n\geq 2$ and~$\gamma,R>0$.
		 Let~$\SN$ be a $q$-normed ideal of compact operators on~$\mathfrak{H}$ such that there is~$\sigma \in (0,1]$ with~$\sigma < \gamma$ and
		\[ (n-\sigma)^{-1} <q\leq1\:.\]Let~$\CA,\CB$ be two bounded self-adjoint operators on~$\mathfrak{H}$. Suppose that~$|\CA-\CB|^\sigma \in \SN$, then
		\[ \| f(\CA)-f(\CB) \|_\SN \leq C_n \:R^{\gamma-\sigma}\: \bl f\bl_n \:\big\| |\CA-\CB|^\sigma \big\|_\SN  \]
		with a positive constant~$C_n$ independent of~$\CA,\CB,f$ and~$R$.
	\end{Thm}
In order to apply this theorem to the functions~$\eta_\kappa$, as
in the proof of Theorem~\ref{MainThmSec2} we use a partition of unity. Similar as explained in Example~\ref{Ex:CondFRenyi}, we need to choose~$\gamma<1$ for~$\kappa=1$ and~$\gamma \leq \min\{1,\kappa\}$ otherwise. We will later see that with the methods in this work we can only estimate the error terms if~$\kappa >2/3$. Thus we assume that~$2/3 <\gamma <1$ allowing us to treat all these cases simultaneously. This gives rise to the constraint
\[ \sigma \in \Big(\frac{2}{3}\:, \:1\Big)\:. \]
	
Setting~$\CA=\chi_{\mathcal{K}_0} \Opa(T_{u_0}(\Afull^{(\varepsilon)}))\chi_{\mathcal{K}_0}$ and~$\CB=\chi_{\mathcal{K}_0} \Opa(\CA_0)\chi_{\mathcal{K}_0}$ (which are clearly bounded and self-adjoint) we obtain
	\begin{flalign}
		\label{Est I Thm 2.27 applied} 
		\begin{split}
		&\big\|  \eta_\kappa \big( \chi_{\mathcal{K}_0} \:\Opa \big(T_{u_0}(\Afull^{(\varepsilon)}) \big)\:\chi_{\mathcal{K}_0} \big)-\eta_\kappa \big( \chi_{\mathcal{K}_0} \Opa \big(\CA_0 \big)\chi_{\mathcal{K}_0} \big) \big\|_1 \\
		&\leq C\: \big\|\chi_{\mathcal{K}_0} \:\Opa \big(T_{u_0}(\Afull^{(\varepsilon)}) -\CA_0 \big)\: \chi_{\mathcal{K}_0} \big\|_\sigma^\sigma
		\end{split}
	\end{flalign}
	with C independent of~$\CA$ and~$\CB$ (and thus in particular independent of~$u_0$ and~$\alpha$).
	
	For ease of notation, from now on we will denote
	\[  \big(\Delta (\mathfrak{a}_{ij} )^{(\varepsilon)}_{u_0} \big)_{1\leq i,j, \leq 2}  := \Delta \CA_{u_0}^{(\varepsilon)} := T_{u_0}(\Afull^{(\varepsilon)}) - \CA_0 \:. \]
	Note that the symbol of the~$\Opa(.)$ in~\eqref{Est I Thm 2.27 applied} is matrix-valued. But, applying Remark~\ref{Rem:Schatten}~(ii), we obtain
	\[
	\big\| \chi_{\mathcal{K}_0} \Opa(\Delta \CA_{u_0}^{(\varepsilon)}) \chi_{\mathcal{K}_0} \big\|_\sigma^\sigma  \leq 
	\sum_{i,j=1}^2
	\Big\|\chi_{\mathcal{K}_0} \:\Opa \big( \Delta (\mathfrak{a}_{ij})_{u_0}^{(\varepsilon)} \big)\:
	\chi_{\mathcal{K}_0} \Big\|_\sigma^\sigma \:,
	\]
	with scalar-valued symbols~$\Delta (\mathfrak{a}_{ij})_{u_0}^{(\varepsilon)}$. 

	We now proceed by estimating the Schatten norms of the operators
	\[ \chi_{\mathcal{K}_0} \:\Opa\big( \Delta (\mathfrak{a}_{ij})_{u_0}^{(\varepsilon)} \big)\:\chi_{\mathcal{K}_0} \:. \]
	 This will also show that these operators are well-defined and bounded on~$L^2(\mathcal{K}_0,\C)$.
	For the estimates we need the detailed form of the symbols given by
	\begin{flalign}
		(\Da_{11})_{u_0}^{(\varepsilon)}(u,u',\xi) &= e^{M \xi} \chi_{(-m\varepsilon/M,0)}(\xi)\Big(2|f_{0,1}^+|^2\Big(\frac{M \xi}{\varepsilon}\Big)-1\Big)\notag \\
		&\phantom{=} + r_{11}\Big(u+u_0,u'+u_0,\frac{M \xi}{\varepsilon}\Big) \label{err1} \\
		(\Da_{12})_{u_0}^{(\varepsilon)}(u,u',\xi)&= 2e^{-M \xi}e^{2Mi\xi (u+u_0)/(\varepsilon)}\Big[ \overline{\fm}\Big(\frac{-M \xi}{\varepsilon}\Big) \fp \Big(\frac{-M \xi}{\varepsilon}\Big) \chi_{(0,m\varepsilon/M)}(\xi) \notag \\
		&\phantom{=} + t_{12}\Big(\frac{-M \xi}{\varepsilon}\Big)\chi_{(m\varepsilon/M, \infty)}(\xi) \Big]+ r_{12}\Big(u+u_0,u'+u_0,\frac{M \xi}{\varepsilon}\Big)\\
		(\Da_{21})_{u_0}^{(\varepsilon)}(u,u',\xi) &= 2e^{M \xi}e^{2Mi\xi (u+u_0)/\varepsilon}\Big[ \fm\Big(\frac{M \xi}{\varepsilon}\Big)\overline{\fp}\Big(\frac{M \xi}{\varepsilon}\Big) \chi_{(-m\varepsilon/M,0)}(\xi) \notag \\
		& \phantom{=} + t_{21}\Big(\frac{M \xi}{\varepsilon}\Big)\chi_{(-\infty, -m\varepsilon/M)}(\xi) \Big]  + r_{21}\Big(u+u_0,u'+u_0,\frac{M \xi}{\varepsilon}\Big)\\
		(\Da_{22})_{u_0}^{(\varepsilon)}(u,u',\xi) &= -e^{-M \xi}\chi_{(0,m\varepsilon/M)}(\xi)\Big(2|f_{0,1}^-|^2\Big(\frac{-M \xi}{\varepsilon}\Big)-1\Big) \notag \\
		&\phantom{=}+ r_{22}\Big(u+u_0,u'+u_0,\frac{M \xi}{\varepsilon}\Big)\:, \label{err2}
	\end{flalign}
	with~$r_{ij}(u,u',\xi)=(\mathcal{R}_0(u,u',\xi))_{ij}$ for any~$1\leq i,j \leq 2$. 
	Note that these equations only hold as long as~$u$ is smaller than some fixed~$u_2$ which we may always assume as we take the limit~$u_0 \rightarrow - \infty$. 

	One can group the terms in these functions into three classes, each of which will be estimated with different techniques: There are terms which are supported on ``small" intervals~$[-m\varepsilon/M,0]$ or~$[0,m\varepsilon/M]$. There are terms that contain the factor~$e^{2Mi\xi (u+u_0)/\varepsilon}$, which makes them oscillate faster and faster as~$u \leq u_0 \rightarrow -\infty$. And, finally, there are the~$r_{ij}$-terms which decay rapidly in~$u$ and/or~$u'$. Due to the triangle inequality~\eqref{est:triangle}, it will suffice to estimate each of these classes separately.

\subsubsection{Error Terms with Small Support}
We use this method for terms which do not depend on~$u$ and~$u'$ and which in~$\xi$ are supported in a small neighborhood of the origin. More precisely, these terms are of the form
\[ \begin{array}{cl}
\displaystyle e^{M\xi} \chi_{(-m\varepsilon/M,0)}(\xi)\Big(2|f_{0,1}^+|^2(M\xi/\varepsilon)-1\Big)
& \quad\text{in~$\Delta (\mathfrak{a}_{11})_{u_0}^{(\varepsilon)}$} \\[0.8em]
\displaystyle -e^{-M\xi}\chi_{(0,m\varepsilon/M)}(\xi)\Big(2|f_{0,1}^-|^2(-M\xi/\varepsilon)-1\Big) & \quad\text{in~$\Delta (\mathfrak{a}_{22})_{u_0}^{(\varepsilon)}$} \:.
\end{array} \]
Since these operators are translation invariant, we do not need to apply the translation operator~$T_{u_0}$.
This also shows that the error corresponding to these terms can be estimated independent of~$u_0$.
	For the estimate we will apply Proposition~\ref{prp72}.
As an example, consider
	\[ a^{(\varepsilon)}(\xi):= e^{M\xi} \chi_{(-m\varepsilon/M,0)}(\xi)\Big(2|f_{0,1}^+|^2(M\xi/\varepsilon)-1\Big)\:,\]
	and 
	\[ h:=\chi_\mathcal{K}\:, \]
	which are both in~$L^2_{\loc}(\R)$ because~$|f_{0,1}^+|$ is bounded. Moreover, applying Lemma~\ref{lem:Rescaling} we obtain:
	\begin{flalign*}
		\big( \Opa(a^{(\varepsilon)})\big) (u,u')&= \big( \mathrm{Op}_1 (a_{(\varepsilon)}) \big) (u,u') \:,
	\end{flalign*}
	\Magdalena{Wollen wir den~$\omega$ /$\xi$ Unterschied wirklich überall durchziehen? Ich glaube das würde bei Rindler sehr verwirrend werden. F: Ja, arbeite möglichst mit dimensionslosen Größen.}
	\Magdalena{Notation einheitlich, aber in BH mit Dimensionen, in anderen Fällen ohne.}
	for 
	\[ a_{(\varepsilon)}(\xi):= a^{(\varepsilon)}(\xi/\alpha) = e^{ \varepsilon\xi} \:  \chi_{(-m,0)}(\xi)\: \Big(2\big|f_{0,1}^+(\xi)\big|^2-1\Big)\:, \]
	which is again in~$L^2_{\loc}(\R)$ for the same reasons as~$a^{(\varepsilon)}$.
	Now we apply the H\"older-type inequality~\eqref{est:HoelderLike} and Proposition~\ref{prp72} with~$p\in (0,1)$ arbitrary to obtain
	\begin{flalign*}
		\big\| \chi_{\mathcal{K}_0}\: \Opa(a^{(\varepsilon)})\:\chi_{\mathcal{K}_0}  \big\|_p^p &\leq \|\chi_{\mathcal{K}_0}\|_{\infty}^p \:\big\| \chi_{\mathcal{K}_0} \Opa(a^{(\varepsilon)}) \big\|_p^p \\ &\leq
		\big\| \chi_{\mathcal{K}_0} \mathrm{Op}_1 (a_{(\varepsilon)}) \big\|_p^p \leq C \:|\chi_{\mathcal{K}_0}|_p^p\: |a_{(\varepsilon)}|_p	^p\:.
	\end{flalign*}
Next, noting that~$\mathcal{K}_0 = (-\rho,0) \subseteq (-\lceil \rho \rceil, 0)$, it follows that
	\[|\chi_{\mathcal{K}_0}|_p^p \leq\sum_{-\lceil \rho\rceil }^0 1 = \lceil \rho \rceil \:.  \]
Similarly, since~$|a_{(\varepsilon)}(\omega)|$ is bounded by one,
	\begin{flalign*}
		|a_{(\varepsilon)} |_p^p \leq \sum_{-\lceil m\rceil }^0 1 \leq \lceil m\rceil  \:.
	\end{flalign*}
Combining the last two inequalities, we conclude that
\beq \label{bound1}
\| \chi_{\mathcal{K}_0} \: \Opa(a^{(\varepsilon)})\: \chi_{\mathcal{K}_0}  \|_p^p \leq C \lceil \rho \rceil \lceil m\rceil  \:.
\eeq
Completely similar for
	\[ \tilde{a}^{(\varepsilon)}(\xi) := -e^{-M\xi}\chi_{(0,m\varepsilon/M)}(\xi) \: \]
	we obtain
\beq  \label{bound2}
\| \chi_{\mathcal{K}_0}\: \Opa(\tilde{a}^{(\varepsilon)})\: \chi_{\mathcal{K}_0} \|_p^p \leq C \lceil \rho \rceil \lceil m\rceil \:,
\eeq
for any~$p\in(0,1)$ so in particular for~$p=\sigma$.
The estimates~\eqref{bound1} and~\eqref{bound2} show that the error terms with small support
are bounded uniformly in~$u_0$ and~$\alpha$. Therefore, dividing by~$\log \alpha$
and taking the limit~$\alpha \rightarrow \infty$, these error terms drop out.
\subsubsection{Rapidly Oscillating Error Terms}
After translating the symbol by~$u_0$, these error terms are of the form
	\begin{flalign*}
		b^{(\varepsilon)}(u,\xi,\alpha) &= e^{-M\xi} \: e^{2Mi\xi(u+u_0)/\varepsilon}\: g(M\xi/\varepsilon)\:\chi_{(0,\infty)}(\xi) \quad \mathrm{or}\\
		\tilde{b}^{(\varepsilon)}(u,\xi,\alpha) &= e^{M\xi} \:e^{2Mi\xi(u+u_0)/\varepsilon}\:\tilde{g}(M\xi/\varepsilon)\:\chi_{(-\infty,0)}(\xi) \:,
	\end{flalign*}
	for some functions~$g, \tilde{g}$ which are measurable and bounded.
	They appear in~$\Delta (\mathfrak{a}_{12})_{u_0}^{(\varepsilon)}$ and~$\Delta (\mathfrak{a}_{21})_{u_0}^{(\varepsilon)}$. For simplicity, we restrict attention to the symbols of the form~$b^{(\varepsilon)}$, but all estimates work the same for~$\tilde{b}^{(\varepsilon)}$ in the same way.
We will make use of the following results from~\cite[Theorem 4 on page 273 and p.~254, 263, 273]{birman-solomjak}, which adapted and applied to our case of interest can be stated as follows.
	\begin{Thm}
		\label{ThmBuchFelix}
		Let~$l\in \N_0$ and~$K$ be an integral operator on~$L^2(\mathcal{K}_0)$ 
		(again with~$\mathcal{K}_0=(-\rho, 0)$) with kernel~$k$, i.e.\ for any~$\psi \in L^2(\mathcal{K}_0)$:
		\[  (K\psi)(u) = \int_{\mathcal{K}_0} k(u,u')\:\psi(u') \: du' \:.
		\]
		If~$k(.,u') \in W^l_2(\mathcal{K}_0)$  for almost all~$u' \in \mathcal{K}_0$ with
		\[
		\theta_2^2(t):=\int_{\mathcal{K}_0}\big\|k(.,u')\big\|^2_{W^l_2(\mathcal{K}_0)}\: du' < \infty\:,
		\]
		then
		\[ K \in \SN_p \qquad \text{for~$p>p'=(1/2+l)^{-1}$} \]
		and
		\[
		\|K\|_p \leq C_{l,p}\:  \theta_2(t)\:.
		\]
	\end{Thm}

	We want to apply this theorem for~$p\in (0,1)$ arbitrary, thus let~$l\in \N$ arbitrary. Moreover, in view of Lemma~\ref{Simple Integral Rep.}, the integral representation corresponding to~$\chi_{\mathcal{K}_0} \Opa(b^{(\varepsilon)})\chi_{\mathcal{K}_0}$ may be extended to all of~$L^2(\R)$ and we may interchange the~$d\xi$ and~$du'$ integrations. Thus we need to estimate the norm~$\theta_2$ of the kernel of this operator. Thus we consider kernels of the form
	\begin{flalign*}
		k_{u_0}^{(\varepsilon)}(u,u')&:= \frac{M}{2\pi \varepsilon}\int_0^\infty e^{iM\xi (u+u'+u_0)/\varepsilon} \: e^{-M\xi} \: g(M\xi/\varepsilon) \: d\xi \\
		&\phantom{:}= \frac{M}{2\pi}\int_0^\infty e^{i M \tilde{\xi} (u+u'+u_0)}  \:e^{-\varepsilon M \tilde{\xi}} \: g(M \xi) \: d\tilde{\xi}  \:,
	\end{flalign*}
	where we used the change of coordinates~$\tilde{\xi}= \xi/\varepsilon$.
	\Magdalena{Fixiere Notation: $\omega$ für dimensions-Größen, $\xi$ für dimensionslose Größen.}
	\Magdalena{Müssen auch Ortskoordinaten dimensionslos machen!}
Since~$g$ is bounded and the factor~$e^{-\varepsilon M \tilde{\xi}}$ provides exponential decay, these kernels are always differentiable up to arbitrary orders with
	\[  \frac{d^l}{du^l} k_{u_0}^{(\varepsilon)}(u,u')= \frac{1}{2\pi} \int_0^\infty (i M \tilde{\xi})^l \: e^{i\omega (u+u'+u_0)} \: e^{-\varepsilon \omega} \: g(\omega) \: d\omega \:, \qquad \text{for any } l\in \N\:.  \]
	Our goal is to show that the limit~$u_0 \rightarrow - \infty$  of~$\theta_2\big(k_{u_0}^{(\varepsilon)}\big)$ is 
uniformly bounded in~$\varepsilon$. To this end, we first note that
	\begin{flalign*}
		k_{u_0}^{(\varepsilon)}(u,u') = \Four(h_{1,(\varepsilon)}) (-u-u'-u_0)\:, \quad \frac{d^l}{du^l} k_{u_0}^{(\varepsilon)}(u,u') =  \Four(h_{2,(\varepsilon)})(-u-u'-u_0)\:,
	\end{flalign*}
	with
	\begin{flalign*}
		h_{1,(\varepsilon)}(\omega) := \frac{1}{2\pi}\:e^{-\varepsilon \omega}\: g(\omega)\: \chi_{(0,\infty)}(\omega) \:,\quad   				h_{l,(\varepsilon)}(\omega) := \frac{1}{2\pi}\:(i\omega)^l \:e^{-\varepsilon \omega} \:g(\omega)\: \chi_{(0,\infty)}(\omega) \:.
	\end{flalign*}
	Now note that both~$h_{1,(\varepsilon)}$ and~$h_{l,(\varepsilon)}$ are in~$L^1(\R)$, so the Riemann-Lebesgue Lemma (see for example \cite[Theorem 1]{bochner-chandra})
tells us that for any~$\delta >0$ we can find~$R=R(\varepsilon)>0$ such that 
	\[ \big|k_{u_0}^{(\varepsilon)}u,u') \big| , \bigg| \frac{d^l}{du^l} k_{u_0}^{(\varepsilon)}(u,u') \bigg| \leq \delta \qquad \mathrm{for}\; |u+u'+u_0| > R \:.\]
	Keeping in mind that for~$u_0 \leq 0$,
	\[ |u+u'+u_0| \geq |u_0| \qquad \text{for any~$u,u' \in \mathcal{K}_0$}\;,  \]
	this is satisfied for~$u_0 < -R(\varepsilon)$. This yields
	\[ \big\|k_{u_0}^{(\varepsilon)}(.,y) \big\|^2_{W_2^l(\mathcal{K}_0)} \leq l \delta^2\rho\:,  \]
	which in turn leads to
	\[ \theta_2(k_{u_0}^{(\varepsilon)}) \leq \sqrt{l} \delta \rho\:, \]
	and so
	\[ \lim\limits_{u_0\rightarrow \infty}\| \chi_{\mathcal{K}_0}\: \Opa(b^{(\varepsilon)}) \: \chi_{\mathcal{K}_0}  \|_p  \leq C_{1,\sigma}\: \sqrt{l}\rho \delta \:, \]
	for any~$p\in(0,1)$, so in particular this holds for~$p=\sigma$.
	As~$\delta$ is arbitrary, we conclude that the rapidly oscillating error
	terms vanish in the limit~$u_0 \rightarrow - \infty$.
	We note for clarity that, since we take the limit~$u_0 \rightarrow -\infty$ first, we do not need to worry
	about the dependence of the above estimate on~$\varepsilon$.	
	
\subsubsection{Rapidly Decaying Error Terms}
	\label{Fast decaying terms}
	Finally, we consider the rapidly decaying error terms. In order to determine their detailed form,
	we first note that the solutions of the radial ODE have the asymptotics as given in Lemma~\ref{X_at_hor}
	with an error term of the form
	\begin{flalign*}
		R_0(u)=\begin{pmatrix}
			e^{-i\omega u} R_+(u,\omega) \\
			e^{-i\omega u} R_-(u,\omega)
		\end{pmatrix} \:,
	\end{flalign*}
	with 
	\[  R_\pm(u) := f^\pm(u)-f_0^\pm \]
	and~$f^\pm$ as in~\eqref{eqappA1}. Using this asymptotics in the integral representation, the error terms in~\eqref{err1}--\eqref{err2} can (similar as explained in Appendix~\ref{appB}) computed to be
	\begin{align*}
		r_{11}(u,u',\omega) &= \chi_{(-\infty,0)}(\omega) \: e^{\varepsilon \omega} \sum_{a,b=1}^2 t_{ab}(\omega) \\
		&\quad\times\Big[ f_{0,a}^+(\omega)\:\overline{R_{+,b}(u',\omega)} + R_{+,a}(u,\omega)\:\overline{f_{0,b}^+(\omega)} + R_{+,a}(u,\omega)\:\overline{R_{+,b}(u',\omega) } \Big] \\		
		r_{12}(u,u',\omega) &= \chi_{(0,\infty)}(\omega) \: e^{2i\omega u} \: e^{-\varepsilon \omega} \sum_{a,b=1}^2 t_{ab}(-\omega) \\
		&\hspace*{-2cm} \times \Big[ f_{0,a}^+(-\omega)\:\overline{R_{-,b}(u',-\omega)}+ R_{+,a}(u,-\omega)\:\overline{f_{0,b}^-(-\omega)} + R_{+,a}(u,-\omega)\:\overline{R_{-,b}(u',-\omega) } \Big] \\
		r_{21}(u,u',\omega) &= \chi_{(-\infty,0)}(\omega) \:e^{2i\omega u} \:e^{\varepsilon \omega} \sum_{a,b=1}^2 t_{ab}(\omega) \\
		&\quad \times \Big[ f_{0,a}^-(\omega)\:\overline{R_{+,b}(u',\omega)} + R_{-,a}(u,\omega)\:\overline{f_{0,b}^+(\omega)} + R_{-,a}(u,\omega)\:\overline{R_{+,b}(u',\omega) } \Big] \\
		r_{22}(u,u',\omega) &= \chi_{(0,\infty)}(\omega)\:  e^{-\varepsilon \omega} \sum_{a,b=1}^2 t_{ab}(-\omega) \\
		&\hspace*{-2cm} \times \Big[ f_{0,a}^-(-\omega)\:\overline{R_{-,b}(u',-\omega)} 
		+ R_{-,a}(u,-\omega)\:\overline{f_{0,b}^-(-\omega)}  + R_{-,a}(u,-\omega)\:\overline{R_{-,b}(u',-\omega) } \Big]  \:.
	\end{align*}

	In order to estimate these terms, the idea is to apply Theorem~\ref{ThmBuchFelix} (as well as the~$\sigma$-triangle inequality) to each of these terms (with~$u$ and~$u'$ shifted by~$u_0$) and then take the limit~$u_0\rightarrow - \infty$. We will do this for the first few terms explicitly, noting that the other terms
	can be estimated similarly.

Given~$u_2$, we know from Lemma~\ref{X_at_hor} that for all~$u<u_2$,
\beq \label{RRp}
|R_\pm(u,\omega)| \leq c e^{du}\:, \quad |\partial_u R_\pm(u,\omega)|\leq c d e^{du} \:,
\eeq
	with constants~$c,d>0$ that can be chosen independently of~$\omega$ (note that we always normalize the
	solutions by~$|f_0|=1$). Now for any~$a,b \in \{1,2\}$ we consider the symbol
	\[ c^{(\varepsilon)}(u,u',\xi) := \chi_{(-\infty,0)}(\xi) \: e^{M\xi} \: t_{ab}\big(M \xi/\varepsilon \big)\: f_{0,a}^+\big(M\xi/\varepsilon\big)\:\overline{R_{+,b}\big( u',M\xi/\varepsilon \big)}\:\chi_{\mathcal{K}}(u)\:\chi_{\mathcal{K}}(u') \:, \]
	which is contributing to~$r_{11}(u,u',M\xi/\varepsilon)$ (note that we again rescaled here in order to get the correct prefactor~$e^{-i\alpha\xi(u-u')}$).  Translating by~$u_0$ as before gives
	\begin{flalign*}
		&\tilde{c}^{(\varepsilon)}(u,u',\xi)=c^{(\varepsilon)}(u+u_0,u'+u_0,\xi)\\
		&= \chi_{(-\infty,0)}(\xi) \:e^{M \xi} \:t_{ab}\big(M \xi /\varepsilon \big)\: f_{0,a}^+\big(M\xi/\varepsilon \big)
		\:\overline{R_{+,b}\big(u'+u_0,M \xi/\varepsilon \big)}\:\chi_{\mathcal{K}_0}(u)\:\chi_{\mathcal{K}_0}(u') \:.
	\end{flalign*}
By Lemma~\ref{Simple Integral Rep.}, the corresponding integral representation can be extended to all of~$L^2(\R)$, since~$R_+$ is bounded uniformly in~$\omega$ when restricted to the compact interval~$\mathcal{K}_0$ due to Lemma~\ref{X_at_hor}, and the~$e^{M\xi}$-factor provides exponential decay in~$\omega$. Moreover, Lemma~\ref{Simple Integral Rep.} again implies that we may interchange the~$d\xi$ and~$du'$ integrations.

In order to estimate the corresponding error term, we apply Theorem~\ref{ThmBuchFelix} to the kernel
	\[ k_{u_0,\alpha}(u,u'):= \frac{1}{2 \pi} \int_{-\infty}^0 e^{-i\omega(u-u')}\:e^{\varepsilon \omega}\: t_{ab}(\omega) \: f_{0,a}^+(\omega)\: \overline{R_{+,b}(u'+u_0,\omega)}\:d\omega  \]
	(note that we could leave out the~$\chi_{\mathcal{K}_0}$-functions because in Theorem~\ref{ThmBuchFelix} we consider the operator on~$L^2(\mathcal{K}_0)$; moreover, we rescaled back as before).
This kernel is differentiable for similar reasons as before and
	\[ \frac{d}{du} k_{u_0,\alpha}(u,u'):= \frac{1}{2 \pi} \int_{-\infty}^0 (-i\omega)\: e^{-i\omega(u-u')}\:e^{\varepsilon \omega} \:t_{ab}(\omega)\: f_{0,a}^+(\omega)\:\overline{R_{+,b}(u'+u_0,\omega)}\:d\omega  \:.  \]
	Using again the estimates for~$R_\pm$ in~\eqref{RRp} yields for any~$u,u ' \in \mathcal{K}_0$:
	\begin{flalign*}
		|k_{u_0,\alpha}(u,u')|& \leq \frac{c e^{d(u'+u_0)} }{2 \pi} \int_{-\infty}^0 e^{\varepsilon \omega} \: d\omega = \frac{c e^{d(u'+u_0)}}{2\pi \varepsilon}  \:,\\
		\Big| \frac{d}{du} k_{u_0,\alpha}(u,u')\Big| &\leq \frac{c e^{d(u'+u_0)} }{2 \pi}   \int_{-\infty}^0 |\omega|\: e^{\varepsilon \omega} \: d\omega = \frac{c e^{d(u'+u_0)}}{2\pi \varepsilon^2} \:,
	\end{flalign*}
	where we used that~$|t_{ab}|, |f_{0,1\!/\!2}^\pm| \leq 1$.
	Therefore,
	\[ \big\| k_{u_0,\alpha}(.,u') \big\|_{W_2^1(\mathcal{K}_0)}^2 = \rho \Bigg(  \frac{c^2e^{2d(u'+u_0)}}{4\pi^2\varepsilon^2} +  \frac{c^2e^{2d(u'+u_0)}}{4\pi^2\varepsilon^4} \Bigg) \]
	and thus 
	\[
	\theta_2^2(k) \leq  \frac{\rho^2}{\varepsilon^2} \: \big(C_1+ C_2/\varepsilon^2\big) \: e^{2du_0} \:,
	\]
	which makes clear that the corresponding error term vanishes in the limit~$u_0 \rightarrow - \infty$.

We next consider a $u$-dependent contribution to~$r_{11}$ for some~$a,b\in \{1,2\}$
whose kernel of the form
	\[ k_2(u,u'):= \frac{1}{2 \pi}  \int_{-\infty}^0 e^{-i\omega(u-u')}\: e^{\varepsilon \omega}\:  t_{ab}(\omega)\:R_{+,a}(u+u_0,\omega) \: \overline{f_{0,b}^+(\omega)}\: d\omega \:. \]
Differentiating with respect to~$u$ gives
	\begin{flalign*}
		\frac{d}{du}  k_2(u,u'):= \frac{1}{2 \pi}  \int_{-\infty}^0 &e^{-i\omega(u-u')}\:e^{\varepsilon \omega}\: t_{ab}(\omega)\:\overline{f_{0,b}^+(\omega)} \\
		&\Big( \partial_u R_{+,a}(u+u_0,\omega) -i \omega  \: R_{+,a}(u+u_0,\omega)  \Big) \: d\omega \:,
	\end{flalign*}
	so that, similarly as before,
	\begin{flalign*}
		\big|k_2(u,u')\big| &\leq \frac{c e^{d(u+u_0)}}{2\pi \varepsilon}  \:, \\
		\bigg| \frac{d}{du} k_2(u,u') \bigg| &\leq \Big( \frac{c}{2\pi \varepsilon^2} + \frac{c d}{2\pi \varepsilon} \Big) \: e^{d(u+u_0)} \:.
	\end{flalign*}
This gives
	\[  \big\| k_2 (.,u')\big\|_{W_2^1(\mathcal{K}_0)}^2 \leq  \rho \: C(\varepsilon) \: e^{2d u_0} \]
	and thus
	\[ \theta_2(k_2)^2 \leq \rho^2 \: \tilde{C}(\varepsilon) \: e^{2du_0} \]
	with a constant~$C(\varepsilon)$ independent of~$u_0$. This shows that
	the corresponding error term again vanishes as~$u_0 \rightarrow \infty$.

All the other error terms contributing to~$r_{ij}$ can be treated in the same way: The absolute value of  the corresponding kernels (and their first derivatives) can always be estimated by a factor continuous in~$u$  and~$u'$ times a factor exponentially decaying in~$u_0$ like~$e^{du_0}$. This makes it possible to estimate~$\theta_2$ by
a function which decays exponentially as~$u_0 \rightarrow -\infty$.	

\subsection{Estimate of the Error Term (II)}
It remains to estimate the error terms~\eqref{TermII} on page~\pageref{TermII}. First of all note that due to Lemma~\ref{Translation Lemma},
\begin{align*}
	&\Big\| \chi_{\mathcal{K}_0}\: \big( \eta_\kappa\left(\Opa\big(T_{u_0}(\Afull^{(\varepsilon)})\big)\right)- \eta_\kappa\big(\Opa(\CA_0)\big)\big) \:\chi_{\mathcal{K}_0} \Big\|_1 \\
	&= \Big\| \chi_{\mathcal{K}}\: \big( \eta_\kappa\left((\Pi_-^\varepsilon)_{kn}\right)- \eta_\kappa\big(\Opa(\CA_0)\big)\big) \:\chi_{\mathcal{K}} \Big\|_1
\end{align*}
	Luckily, in this case we can directly compute~$\eta_\kappa\big((\Pi_-^\varepsilon)_{kn}\big)$ and~$\eta_\kappa\big(\Opa(\CA_0)\big)$, which simplifies the estimate.
	As explained before in Lemma~\ref{fOpa} we have
	\[  \eta_\kappa\big(\Opa(\CA_0)\big) = \Opa\big(\eta_\kappa(\CA_0)\big) \:. \]	
	Moreover, from Proposition~\ref{Prop g(H)1} and Corollary~\ref{Cor g(H)} we conclude that for any~$\psi \in C^\infty_0(\R, \C^2)$,
\[ \Big( \eta_\kappa\big((\Pi_-^\varepsilon)_{kn} \big) \psi\Big)(u)= \frac{1}{\pi} \int_{-\infty}^0 \: \eta_\kappa(e^{\varepsilon \omega}) \sum_{a,b =1}^{2} t_{ab}(\omega) \: X_a(u,\omega)\: \big\la X_b(.,\omega) \:\big|\: \psi \big\ra \:d\omega \:. \]
So we conclude that we can rewrite
\[
\eta_\kappa\big(\Opa(a_0)\big)-\eta_\kappa\big( (\Pi_-^\varepsilon)_{kn} \big) = \Opa(\Delta \tilde{\CA}^{(\varepsilon)})\:,
\]
where the entries of the symbol~$\tilde{\CA}^{(\varepsilon)} = (\Delta\tilde{\mathfrak{a}}_{i,j}^{(\varepsilon)})_{1\leq i,j\leq 2}$ are given by
	\begin{flalign*}
		\Delta\tilde{\mathfrak{a}}_{1,1}^{(\varepsilon)}(u,u',\xi) &= \eta_\kappa\big(e^{M\xi}\big) \: e^{-M\xi}  \: (\Da_{11})^{(\varepsilon)}_0(\xi) \:,\\
		\Delta\tilde{\mathfrak{a}}_{1,2}^{(\varepsilon)}(u,u',\xi)  &= \eta_\kappa\big(e^{-M\xi}\big)  \: e^{M\xi}  \: \Da_{12})^{(\varepsilon)}_{0}(u,\xi) \:, \\
	(\Delta\tilde{\mathfrak{a}}_{2,1}^{(\varepsilon)}(u,u',\xi)  &= \eta_\kappa\big(e^{M\xi}\big)  \: e^{-M\xi}  \: (\Da_{21})^{(\varepsilon)}_{0}(\xi) \:,\\
		\Delta\tilde{\mathfrak{a}}_{2,2}^{(\varepsilon)}(u,u',\xi)  &= \eta_\kappa\big(e^{-M\xi}\big)  \: e^{M\xi}  \: (\Da_{22})^{(\varepsilon)}_{0}(u,\xi) \:.
	\end{flalign*}
Thus these error terms are almost the same as before, except that the factor~$e^{M\xi}$ has been replaced by~$\eta_\kappa(e^{M\xi})$ etc. and without translation by~$u_0$. Therefore, after applying Lemma~\ref{Translation Lemma} in order to again translate by~$u_0$,
	\begin{align*}
		&\Big\| \chi_{\mathcal{K}}\: \big( \eta_\kappa\left((\Pi_-^\varepsilon)_{kn}\right)- \eta_\kappa\big(\Opa(\CA_0)\big)\big) \:\chi_{\mathcal{K}} \Big\|_1=\Big\| \chi_{\mathcal{K}_0}\: \Opa\big(T_{u_0}(\tilde{\CA}^{(\varepsilon)})\big)\:\chi_{\mathcal{K}_0} \Big\|_1\:,
	\end{align*}
	we can use the same techniques as before if we can show that the functions
	\[ g_1(\omega):=\eta_\kappa(e^{-\varepsilon \omega }) \:  \chi_{(0,\infty)}(\omega)\:, \qquad g_2(\omega):= \omega \:  \eta_\kappa(e^{-\varepsilon \omega }) \:  \chi_{(0,\infty)}(\omega) \:, \]
	are bounded and in~$L^1(\R)$ (and the same for the functions with~$\omega \rightarrow -\omega$, but for simplicity we again omit this case). We start with the case that~$\kappa =1$ and first rewrite these functions in more detail as
	\begin{flalign*}
		g_1(\omega) &= \chi_{(0,\infty)}(\omega)  \: \Big(\varepsilon  \omega  \: e^{-\varepsilon \omega} +\big(e^{-\varepsilon \omega}-1\big)\log\big(1- e^{-\varepsilon \omega}\big)   \Big) \\
		g_2(\omega) &= \chi_{(0,\infty)}(\omega)  \:  \omega  \: \Big(\varepsilon \omega  \:  e^{-\varepsilon \omega} +\big(e^{-\varepsilon \omega}-1\big)\log\big(1- e^{-\varepsilon \omega}\big)   \Big) \:.
	\end{flalign*}
The terms~$e^{-\varepsilon \omega} \omega^k$ with~$k=1,2$ are clearly bounded and in~$L^1(\R)$. Moreover,	\begin{flalign*}
		&\lim\limits_{\omega \searrow 0} \Big( \big(e^{-\varepsilon \omega} -1\big) \: \log\big(1- e^{-\varepsilon \omega}\big) \Big) \\
		&= \lim\limits_{\omega \searrow 0}  \frac{  \log\big(1- e^{-\varepsilon \omega}\big) }{\big(e^{-\varepsilon \omega} -1\big)^{-1}} \overset{l'H}{=}  \lim\limits_{\omega \searrow 0}  \frac{  \varepsilon \: \big(1- e^{-\varepsilon \omega}\big)^{-1}  \: e^{-\varepsilon \omega}  }{ \varepsilon \: \big(e^{-\varepsilon \omega} -1\big)^{-2} \: e^{-\varepsilon \omega} }
		= - \lim\limits_{\omega \searrow 0} \big( 1- e^{-\varepsilon \omega} \big) = 0 \:,
	\end{flalign*}
(where ``$l'H$" denotes the use of L'Hospital's rule) showing that these terms are bounded near~$\omega = 0$. Next,
	\begin{flalign*}
		&\lim\limits_{\omega \rightarrow \infty} \frac{\log\big(1-e^{-\varepsilon \omega}\big)}{e^{-\varepsilon \omega}} \overset{l'H}{=} \lim\limits_{\omega \rightarrow \infty} \frac{  \varepsilon \: \big(1- e^{-\varepsilon \omega}\big)^{-1}  \: e^{-\varepsilon \omega}  }{-\varepsilon \: e^{-\varepsilon \omega}} = \lim\limits_{\omega \rightarrow \infty} \frac{-1}{1-e^{-\varepsilon \omega}} = -1 \:,
	\end{flalign*}
showing that as~$\omega \rightarrow \infty$ those terms decay like~$e^{-\varepsilon \omega}$.
Therefore, these terms are also bounded and in~$L^1(\R)$. Now let~$\kappa \neq 1$, then
\[
	\eta_\kappa(e^{-\varepsilon \omega }) = \log\big( e^{-\kappa \varepsilon \omega } + (1-e^{-\varepsilon \omega })^\kappa\big)\:,
\]
so the functions~$g_1$ and~$g_2$ are bounded. Moreover,~$\eta_\kappa(e^{-\varepsilon \omega })$ decay like~$e^{-\tilde{\kappa}\varepsilon \omega}$ with~$\tilde{\kappa}:=\min\{1,\kappa\}$ for~$\omega \rightarrow \infty$ since
\begin{align*}
	\lim\limits_{\omega \rightarrow \infty}\left| \frac{\eta_\kappa(e^{-\varepsilon \omega }) }{e^{-\tilde{\kappa}\varepsilon \omega}} \right|&\overset{l'H}{=} 
	\frac{\kappa}{1-\kappa} \lim\limits_{\omega \rightarrow \infty} \left| \frac{\frac{e^{-(\kappa -1 )\varepsilon \omega }- (1-e^{-\varepsilon \omega })^{\kappa -1} }{e^{-\kappa \varepsilon \omega } +(1-e^{-\varepsilon \omega })^\kappa} (-\varepsilon)e^{-\varepsilon \omega }}{-\varepsilon e^{-\tilde{\kappa}\varepsilon \omega }}\right| \\
	&\: = \frac{\kappa}{1-\kappa} \lim\limits_{\omega \rightarrow \infty}\left| \frac{\frac{e^{-\kappa\varepsilon \omega }- e^{-\varepsilon \omega } }{e^{-\kappa \varepsilon \omega } +1} }{e^{-\tilde{\kappa}\varepsilon \omega }}\right| = \frac{\kappa}{1-\kappa}
\end{align*}
This shows that~$\kappa \neq 1$, the functions~$g_1$ and~$g_2$ are in~$L^1(\R)$ as well.
	
\section{Proof of the Main Result}
\label{SecProofofMain}
We can now prove our main result.	
\begin{proof}[Proof of Theorem~\ref{Main Res.}]
\label{Proof of Main Res.} Having estimated all the error terms in trace norm and 
knowing that the limiting operator is trace class (see the proof of Theorem~\ref{MainThmSec2}), we
conclude that the operator
\[ \eta_\kappa\big( \chi_{\mathcal{K}} \: (\Pi_-^\varepsilon)_{kn} \: \chi_{\mathcal{K}} \big) - \chi_{\mathcal{K}}\:  \eta_\kappa\big((\Pi_-^\varepsilon)_{kn}\big)\: \chi_{\mathcal{K}} \]
is trace class.
Moreover, we saw that all the error terms vanish after dividing by~$\log \alpha$ and taking
the limits~$u_0\rightarrow - \infty$ and~$\alpha \rightarrow \infty$ (in this order).
We thus obtain by Corollary~\ref{Cor to Main Res Sec 2},
\begin{flalign*}
	&\lim\limits_{\alpha \rightarrow \infty} \lim\limits_{u_0\rightarrow - \infty} \frac{1}{\log \alpha} \Big( \eta_\kappa\big( \chi_{\mathcal{K}} \: (\Pi_-^\varepsilon)_{kn} \: \chi_{\mathcal{K}} \big) - \chi_{\mathcal{K}}\:  \eta_\kappa\big((\Pi_-^\varepsilon)_{kn}\big)\: \chi_{\mathcal{K}} \Big) \\
	&= \lim\limits_{\alpha \rightarrow \infty} \lim\limits_{u_0\rightarrow - \infty}\frac{1}{\log \alpha} \tr D_\alpha(\eta_\kappa, {\mathcal{K}},\CA_0) = \frac{1}{\pi^2}U(1;\eta_\kappa) = \frac{1}{\pi^2} \int_0^1 \frac{\eta_\kappa(t)}{t(1-t)} \:dt \:.
\end{flalign*}
Moreover since by \cite[Appendix]{leschke-sobolev-spitzer3},
\[\int_0^1 \frac{\eta_\varkappa(t)}{t(1-t)} \:dt 
=
\frac{\pi^2}{6}\frac{\varkappa+1}{\varkappa}\:,\]
and recalling that~$\alpha = M / \varepsilon$, we obtain the result.
\end{proof}
	
\section{Conclusions and Outlook} 
\label{Sec Conclusion} 
To summarize this article, we introduced the fermionic entanglement / R{\'e}nyi entropy of a Schwarzschild black hole horizon based on the Dirac propagator as
\beq  S_{\kappa}^{\mathrm{BH}} = \frac{1}{2} \: \sum_{ \mathclap{\substack{ (k,n) \\ \text{occupied}}}} \:\lim\limits_{\varepsilon \searrow 0} \;\lim\limits_{u_0 \rightarrow -\infty}  \frac{1}{\log (M / \varepsilon)} \tr  \Big(\: \eta_\kappa\big( \chi_{\mathcal{K}} (\Pi_-^\varepsilon)_{kn} \chi_{\mathcal{K}} \big) -  \chi_{\mathcal{K}} \eta_\kappa \big((\Pi_-^\varepsilon)_{kn}\big) \chi_{\mathcal{K}} \:\Big)\:. \label{eq:concl1} \eeq
We have shown that we may treat each angular mode separately. This transition enables us to disregard the angular coordinates, which makes the problem essentially one-dimensional in space. Furthermore, in the limiting case we were able to replace the symbol of the corresponding pseudo-differential operator by~$\CA_0$ in~\eqref{LimOpKernel} provided that~$\kappa >\frac{2}{3}$. Since this symbol is diagonal matrix-valued, this reduces the problem to one spin dimension.
Moreover, because~$\CA_0$ is also independent of~$\varepsilon$, the trace with the replaced symbol can be computed explicitly. It turns out to be a numerical constant independent of the considered angular mode.

This leads us to the conclusion that the fermionic entanglement entropy of the horizon is proportional to the number of angular modes occupied at the horizon,
\[ S_{1}^{\mathrm{BH}} = \: \sum_{ \mathclap{\substack{ (k,n) \\ \text{occupied}}}} \: S_{1,kn}^{\mathrm{BH}} = \frac{1}{6} \: \# \big\{(k,n) \: \big|\:\text{angular mode }(k,n) \text{ occupied} \big\} \:, \]
and a similar result holds for the R{\'e}nyi entropies with~$\kappa >\frac{2}{3}$.
This is comparable to the counting of states in string theory~\cite{strominger-vafa}
and loop quantum gravity~\cite{ashtekar-baez}.
Furthermore, assuming that there is a minimal area of order~$\varepsilon^2$, the number of occupied modes at the horizon was given by~$M^2/\varepsilon^2$, which would lead to
\[S_{1}^{\mathrm{BH}} = \frac{1}{6}\:\frac{M^2}{\varepsilon^2} \:. \]
Bringing the factor~$\log(M / \varepsilon)$ in~\eqref{eq:concl1} to the other side, this would mean that, up to lower orders in~$\varepsilon^{-1}$, we would obtain the enhanced area law
\begin{flalign}
\sum_{ \mathclap{\substack{ (k,n) \\ \text{occupied}}}} \:& \;\lim\limits_{u_0 \rightarrow -\infty} \tr  \Big(\: \eta_1\big( \chi_\mathcal{K} (\Pi_-^\varepsilon)_{kn} \chi_\mathcal{K} \big) -  \chi_\mathcal{K} \eta_1 \big((\Pi_-^\varepsilon)_{kn}\big) \chi_\mathcal{K} \:\Big) \notag \\
&= \frac{1}{6}\:\frac{M^2}{\varepsilon^2} \log(M / \varepsilon) + o\big( M^2/\varepsilon^2 \log(M / \varepsilon) \big) \:, \qquad \text{as }\varepsilon
\searrow 0 \:. \label{endsummary}
\end{flalign}

One obvious question is whether this result holds not only in the Schwarzschild geometry,
but for more general and physically realistic black holes. The core of our analysis
makes use only of the asymptotic form of the Dirac wave functions near the event horizon.
This leads to the conjecture that our results should hold for any black hole whose near-horizon
geometry coincides with that of Schwarzschild.
However, it is not clear to us how to prove this result, because our analysis is based on the
integral representation of the Dirac propagator in~\cite{tkerr}, which holds only in the presence of
global spacetime symmetries which allowing for the complete separation of the Dirac equation into
a system of ordinary differential equations. We expect that the integral representation could be derived
in general static and spherically symmetric black hole spacetimes, but this has not been worked out
in the literature. Another generalization which seems possible from the perspective of separation of variables
is the extension from Schwarzschild to the Kerr-Newman and 5D Myers-Perry family of rotating black holes,
where corresponding integral representations have been derived (see~\cite{tkerr,
qiushiwang}).

Another interesting topic for future research would be to analyze
the number of occupied angular momentum modes at the horizon in more detail, in particular for a realistic model where the Schwarzschild black hole
arises in the large time asymptotics of a spherically symmetric star undergoing gravitational collapse.
Keeping into account the dynamical evolution of the regularization
(as described for example by the regularized Hadamard expansion~\cite{reghadamard}),
we do not expect that the resulting effective regularization length~$\varepsilon$ will be the same
for each angular momentum mode. Therefore, the constant in~\eqref{limcount}
might become $n$-dependent, 
\[ \lim\limits_{\varepsilon \searrow 0} \;\lim\limits_{u_0 \rightarrow -\infty}  \frac{1}{\log (M / \varepsilon)} \tr  \Delta_1 \big((\Pi_-^\varepsilon)_{kn}, \mathcal{K}\big) = c(n) \:. \]
In this case, the counting of occupied states in~\eqref{entropy BH} would have to be replaced by a weighted sum,
\[ \sum_{\text{$(k,n)$ occupied}} c(n) \:. \]
This weighted sum can be understood as the effective number of occupied states.
It might be that the weighted sum remains finite even if an infinite number of angular momentum modes~$(k,n)$
is taken into account.

Another interesting extension of our results is to consider more general quasi-free states
in particular {\em{thermal states}}. Although our methods still apply in this situation, it is not clear to us which
results to expect.

Finally, it seems possible to generalize our results technically as follows.
Having mainly the von Neumann entropy function (i.e.\ $\kappa=1$) in mind, we only established estimates for~$R_\pm$ and its first derivative in~$u$. This had the consequence that were not able to estimate the corresponding error terms for R{\'e}nyi entropies with~$\kappa\leq \frac{2}{3}$ with the same methods. However, those methods would also apply for~$\kappa \leq \frac{2}{3}$ if one were able to estimate higher derivatives of~$R_\pm$ suitably well. This is another topic for future research.

\appendix
\section{A Few Technical Results on Pseudo-Differential Operators}
\label{SecPropertiesOpa}
This appendix provides a few technical results on pseudo-differential operators.
By~$\Opa(\CA)$ we always denote an operator as defined in Section~\ref{Sec:TechToolsDefs}.

\begin{Lemma}
	\label{fOpa}
	Let~$\Opa(\CA)$ as defined in Section~\ref{sec:RepPseudo} with~$\mathcal{U}=\R^d$.
	If the symbol~$\CA$ is hermitian matrix-valued, measurable,  independent of~$\bx$ and~$\by$, i.e.\ $\CA(\bx, \by,\bxi)\equiv \CA(\bxi)$ and is uniformly bounded in~$\bxi$ (with respect to the ordinary matrix-sup-norm), then~$\Opa(\CA)$ is well-defined and self-adjoint.
	Moreover, for any Borel function~$f$ defined on the spectrum of~$\Opa(\CA)$, we have
	\[ f(\Opa(\CA)) = \Opa(f(\CA))\:. \]
\end{Lemma}
\begin{proof}
	Note that
	\begin{flalign}
		\label{SpecRepOpa}
		\Opa( \CA ) = \Four\: \CA(./\alpha)\: \Four^{-1} \:,
	\end{flalign}
	where~$\CA(./\alpha)$ denotes the matrix-multiplication operator by~$\CA(./\alpha)$ and~$\Four$ the unitary extension of the Fourier transform on~$L^2(\R^d,C^n)$ (since~\eqref{SpecRepOpa} holds for all~$C^\infty_0(\R^d,\C^n)$-functions and the right hand side defines a bounded operator on~$L^2(\R^d,\C^n)$).
	This also shows, that~$\Opa(\CA)$ is bounded (and therefore well-defined) and self-adjoint.
	Since for any~$\bxi \in \R^d$, the matrix~$\CA(\bxi/\alpha)$ is hermitian matrix valued the spectral theorem for matrices yields a unitary matrix~$U(\bxi)$ such that
	\[
	U(\bxi)\CA(\bxi/\alpha) U(\bxi)^{-1} = \mathrm{diag}(\mathfrak{d}_1(\bxi), \cdot , \mathfrak{d}_n(\bxi))=:\mathcal{D}(\bxi)\:.
	\]
	Using the identification~$L^2(\R^d,\C^n) \cong L^2\big(\{1,\dots, n\}\times \R^d, \C\big)$, the operator~$ U(\cdot)^{-1}\Four^{-1}$ can be interpreted as the unitary transformation form the multiplicative version of the spectral theorem for the operator~$\Opa(\CA)$ and~$\mathcal{D}$ as the corresponding function.
	Thus we have
	\[ f(\Opa(\CA)) = \Four \:U(\cdot) \: f(\mathcal{D}(\cdot/\alpha))\: U(\cdot)^{-1}\: \Four^{-1}  = \Four \:f(\CA(\cdot/\alpha)) \:\Four^{-1} = \Opa(f(\CA))\:.\]
\end{proof}

\begin{Remark}
	\label{ExtScwartzFcts}
	{\em{A similar argument as in the above proof can be used to prove a criterion on when the integral representation of~$\Opa(\CA)$ extends to all Schwartz functions.  
			Let~$\CA$ be a symbol which only depends on~$\bx$. Then, just as in the proof of Lemma~\ref{fOpa} we conclude that 
			\[
			\Opa(\CA)=\Four^{-1} \: \CA(\cdot/\alpha) \: \Four\:.
			\]
			Now take an arbitrary Schwartz function~$\phi \in \mathcal{S}(\R^d,\C^n)$, then~$\Four \phi$ is defined with the usual integral representation. Moreover since~$\Four$ is an automorphism on the Schwartz-space~$\Four \phi \in \mathcal{S}(\R^d,\C^n)$. Furthermore, if the map~$\CA(\cdot/\alpha) \:(\Four \phi)(\cdot)$ is in~$\L^1(\R^d,\C^n)$, the inverse Fourier transform is again given by the usual integral representation meaning that the integral representation of~$\Opa(\CA)$ extends to~$\phi$.
			
			We conclude that if for any Schwartz function~$\psi \in \mathcal{S}(\R^d,\C^n)$, the vector-valued function~$\CA(\cdot/\alpha)\: \psi(\cdot) \in L^1(\R^d,\C^n)$, then the integral representation of~$\Opa(\CA)$ extends to all Schwartz functions.
			
			This is for example the case for any measurable and (in the matrix sup-norm) bounded symbol~$\CA$. \QEDrem}}
\end{Remark}

The next lemma will be needed for consistency reasons when taking the limit~$u_0 \rightarrow -\infty$:
\begin{Lemma}
	\label{Translation Lemma}Let~$\Opa(\CA)$ as in Section~\ref{sec:RepPseudo}, let~$U,V \subset \R^d$ be arbitrary Borel sets and~$\bc \in \R^d$ an arbitrary vector.
	For any~$\bx,\by, \bxi \in \R^d$ we transform a given symbol~$\CA$ by
	\[ T_{\bc}(\CA) (\bx,\by,\bxi):= \CA(\bx+\bc,\by+\bc,\bxi)\:. \]
	Then there is a unitary transformation~$t_c$ on~$L^2(\R^d, \C^n)$ such that
	\begin{align}
		\label{Opa shift}
		t_{\bc} \:\chi_{U+{\bc}} \:\Opa \big(T_{-{\bc}}(\CA) \big)\: \chi_{V+{\bc}}\: t_{\bc}^{-1} = \chi_U \:\Opa(\CA)\: \chi_V \:.
	\end{align}
	Moreover, assuming in addition that~$\Opa(\CA)$ is self-adjoint, we conclude that for any Borel function~$f$,
	\begin{flalign}
		\label{f(Opa) shift}
		f \big( \chi_U\: \Opa(\CA)\: \chi_U \big) = t_{\bc} \,f \big( \chi_{U+\bc} \:\opa(T_{-\bc}(\CA)) \:\chi_{U+\bc}  \big) \,t_{\bc}^{-1} \:.
	\end{flalign}
\end{Lemma}
\begin{proof}
	We will show that the desired unitary operator is given by the translation operator
	\[t_{\bc}: L^2(\R^d) \rightarrow L^2(\R^d),\quad f \mapsto f(.+{\bc})\:,\]
	(which is obviously unitary).
	Note that for any Borel set~$W\subseteq \R^d$
	\[ \chi_{W+{\bc}}  = t_{-{\bc}} \:\chi_W\: t_{{\bc}} \:,\]
	and therefore
	\begin{flalign*}
		\chi_U \:\Opa(\CA)\: \chi_V &=  t_{{\bc}} \:\chi_{U+{\bc}}\: t_{-{\bc}}\: \Opa(\CA) \:t_{{\bc}} \:\chi_{V+{\bc}}t_{-{\bc}}\:.
	\end{flalign*}
	By a change of coordinates we obtain for arbitrary~$\psi \in C^\infty_0(\R^d,\C^n)$
	\begin{flalign*}
		&\big(t_{-{\bc}} \:\Opa(\CA) \:t_{{\bc}}\psi\big)(\bx) \\
		&= \big(\Opa(\CA)\:t_{{\bc}}\psi\big)(\bx-{\bc}) = \frac{\alpha^d}{(2\pi)^d}\int d \bxi \int d\by\:e^{-i\alpha \bxi(\bx- \bc-\by)}\:\CA(\bx-\bc,\by,\bxi)\:\psi(\by+\bc)\\
		&= \frac{\alpha^d}{(2\pi)^d}\int d \bxi \int d\by\:e^{-i\alpha \bxi(\bx-\by)}\: \CA(\bx-\bc,\by-\bc,\bxi,)\:\psi(\by) = \big(\Opa \big(T_{-\bc}(\CA)\big)\psi\big)(\bx)\:.
	\end{flalign*}
	This shows~\eqref{Opa shift}.
	
	For~\eqref{f(Opa) shift} we make use the multiplication operator version of the spectral theorem. This provides a unitary transformation~$\phi$ and a suitable function~$g$ such that
	\[ \chi_U \:\Opa(\CA)\: \chi_U  = \phi^{-1} \:g\: \phi \:.\]
	Combined with the previous discussion this implies
	\[ \chi_{U+\bc} \:\Opa\big(T_{-\bc}(\CA)\big)\: \chi_{U+\bc} = (\phi \:t_{\bc})^{-1} \:g\: (\phi \:t_{\bc})  \:, \]
	which is the multiplication operator representation of~$\chi_{U+{\bc}} \:\opa\big(T_{-{\bc}}(\CA)\big)\: \chi_{U+{\bc}}$,
	because~$\phi\: t_c$ is also a unitary operator. Therefore
	\[ f \big( \chi_{U+{\bc}} \:\Opa(T_{-{\bc}}(\CA))\: \chi_{U+{\bc}} \big)
	= \big( \phi\: t_{\bc} \big)^{-1} \: (f\circ g)\; \big(\phi \:t_{\bc} \big) = t_{\bc}^{-1} \:f \big(\chi_U \:\Opa(\CA)\: \chi_U \big)\: t_{\bc}\:, \]
	concluding the proof.
\end{proof}

\begin{Remark}
	\label{Rem:XiTrans}
	{\em{
			A similar result as Lemma~\ref{Translation Lemma} holds for translations in the $\bxi$-variable: Let~$\bc \in \R^d$ be an arbitrary vector. For any~$\bx, \by,\bxi \in \R^d$ transform a given symbol~$\CA$ by
			\[
			R_{\bc} (\CA) (\bx,\by,\bxi) := \CA(\bx,\by,\bxi + \bc)\:.
			\]	
			Then, for any~$\psi \in C^\infty_0(\R^d,\C^n)$ and~$\bx \in \R^d$ we have
			\begin{align*}
				&\big( \Opa( R_{\bc}(\CA)) \psi \big)(\bx) = \frac{\alpha^d}{(2\pi)^d} \int_{\R^d} d\bxi \int_{\mathcal{U}} d\by\: e^{-i\alpha \bxi (\bx-\by)} \CA(\bx,\by,\bxi + \bc) \psi(\by) \\
				&= \frac{\alpha^d}{(2\pi)^d} \int_{\R^d} d\bxi \int_{\mathcal{U}} d\by\: e^{-i\alpha (\bxi-\bc) (\bx-\by)} \CA(\bx,\by,\bxi + \bc) \psi(\by) \\
				&= \frac{\alpha^d}{(2\pi)^d} e^{-i\alpha \bc \bx } \int_{\R^d} d\bxi \int_{\mathcal{U}} d\by \:e^{-i\alpha \bxi (\bx-\by)} \CA(\bx,\by,\bxi + \bc) e^{i\alpha \bc \by }\psi(\by) \:.
			\end{align*}
			This shows that
			\[
			\Opa( R_{\bc}(\CA)) =e^{i\alpha \bc \cdot }\:\Opa( \CA ) \: e^{-i\alpha \bc \cdot }\:,
			\]
			which implies similar consequences for trace and Schatten-norms since the operator~$\mathcal{M}_{e^{i\alpha \bc \cdot }}$ is unitary. \QEDrem
	}}
\end{Remark}

\begin{Lemma}
	\label{lem:Rescaling}
	Let~$\CA$ be a symbol such that~$\Op_\alpha (\CA)$ is well defined on~$L^2(\R^d,\C^n)$, let~$\mathcal{K} \subset \R$ be some measurable subset and~$\alpha,\beta >0$ some arbitrary constants. Then,
	\begin{itemize}
		\item[(i)] There is a unitary operator~$V_\beta$ on~$L^2(\R^d,\C^n)$ such that 
		\begin{align}
			\label{eq:RescPos}
			V_\beta^{-1} \:\chi_\mathcal{K} \:\Opa(\CA)\: \chi_\mathcal{K} \: V_\beta  = \chi_{\beta \mathcal{K}}\:  \Op_{ \alpha/\beta}(\CA) \:\chi_{\beta \mathcal{K}} \:.
		\end{align}
		We refer to this as {\bf{rescaling in position space}}.
		\item[(ii)] By {\bf{rescaling in momentum space}} we mean the equality
		\begin{align}
			\label{eq:RescMom}
			\Op_\alpha(\CA) = \Op_{\beta \alpha}\big(\CA(\beta\: \cdot)\big)\:.
		\end{align}
	\end{itemize}
\end{Lemma}
\begin{proof}
	First, (ii) simply follows by changing coordinates in the $\bxi$-integral. 
	
	\noindent
	For (i) consider the unitary operator~$V_\beta$, which is for any~$\psi \in L^2(\R^d,\C^n)$ defined by 
	\[\big( V_{\beta} \varphi \big) (\bx) := \beta^{d/2}\: \varphi(\beta \bx)\:, \qquad \text{for any } \bx \in \R\:\:. \]
	Then, for any~$\psi \in L^2(\R^d,\C^n)$ and~$\bx \in \R^d$,
	\[
	\big(V_{\beta}^{-1} \chi_\mathcal{K} V_{\beta}\psi\big)(\bx) = \beta^{d/2}\big(V_{\beta}^{-1} \chi_\mathcal{K} \psi(\beta .) \big)(\bx) =\chi_\mathcal{K}(x/\beta) \psi(\bx) = \big(\chi_{\beta \mathcal{K}}\psi \big)(\bx)
	\]
	and for any~$\psi \in S(\R^d,\C^n)$
	\begin{align*}
		&\big(V_{\beta}^{-1} \Opa(\CA) V_{\beta}\psi \big)(\bx)= 
		\frac{\alpha^d }{2\pi} \int_{-\infty}^\infty d\bxi \int_{-\infty}^\infty dy \: e^{i\alpha \bxi(\bx/\beta-\by)} \CA(\bxi)  \psi(\beta y)\\
		&=
		\frac{(\alpha/\beta)^d}{2\pi} \int_{-\infty}^\infty d\bxi \int_{-\infty}^\infty dy \: e^{i\alpha/\beta \bxi(\bx-\by)} \CA(\bxi)  \psi(\by)
		= (\Op_{\alpha/\beta}(\CA)\psi )(\bx)\:,
	\end{align*}
	where in the second step we applied a change of coordinates in the $\by$-integral.
\end{proof}

\begin{Lemma}
	\label{Simple Integral Rep.}
	Let~$\Opa(\CA)$ as in Section~\ref{sec:RepPseudo}, such that~$\CA$ satisfies 
	\[ \int_{\R^d} d\bxi \:\sqrt{\int_{\mathcal{U}} d\by\: \big\| \CA(\bx,\by,\bxi) \big\|_{n \times n}^2} < \infty \:, \qquad \text{for any } \bx \in \R^d \]
	(where~$\|.\|_{n\times n}$ is the ordinary sup-norm on the $n\times n$-matrices). Then the integral representation of~$\Opa(\CA)$ may be extended to all~$L^2(\R^d,\C^n)$-functions and the~$\by$ and the~$\bxi$ integrations may be interchanged. Thus for any~$\psi \in L^2(\R^d,\C^n)$ and almost any~$\bx \in \R^d$,
	the following equations hold,
	\begin{flalign*}
		\Big(\Opa (\CA)\,\psi \Big)(\bx) =& \Big(\frac{\alpha}{2\pi}\Big)^d \int_{\R^d} d\bxi \int_{\mathcal{U}} d\by \:e^{-i\alpha \bxi \cdot (\bx-\by)} \:\CA(\bx, \by, \bxi) \:\psi(\by) \\
		=& \Big(\frac{\alpha}{2\pi}\Big)^d \int_{\mathcal{U}} d\by \int_{\R^d} d\bxi\: e^{-i\alpha \bxi \cdot (\bx-\by)} \:\CA(\bx, \by, \bxi) \:\psi(\by)\:.
	\end{flalign*}
\end{Lemma}
\begin{proof}
	We first show that, applying the Fubini-Tonelli theorem and H\"older's inequality, the integrations may be interchanged, by estimating
	\begin{flalign*}
		&\int_{\R^d} d\bxi \int_{\mathcal{U}} d\by\: \big|e^{-i\alpha \bxi \cdot (\bx-\by)} \CA(\bx, \by, \bxi) \psi(\by) \big|  \\
		&\leq \int_{\R^d} d\bxi \;\sqrt{\int_{\mathcal{U}}d\by\: \big\|\CA(\bx, \by, \bxi) \psi(\by) \big\|_{n \times n}^2 } \;\big\| \psi\big\|_{L^2(\R^d,\C^n)} < \infty \:.
	\end{flalign*}
	Next, we want to show that we can extend the integral representation to all~$L^2(\R^d,\C^n)$-functions, i.e that the above integral indeed corresponds to~$\big(\Opa(\CA)\psi\big)(\bx)$. To this end let~$(\psi_n)_{n\in \N}$ be a sequence of~$C^\infty_0(\R^d,\C^n)$-functions converging to~$\psi$ with respect to the~$L^2(\R^d,\C^n)$-norm. Then~$\Opa(\CA)\psi$ is by definition given by
	\begin{align}
		\label{eq:convOpPsi}
		\Opa(\CA)\psi = \lim\limits_{n \rightarrow \infty} \Opa(\CA)\psi_n \:,
	\end{align}
	where the convergence is with respect to the~$L^2(\R^d,\C^n)$-norm. However, going over to a subsequence we can assume that this convergence also holds pointwise outside of a null set~$N\subseteq \R^d$.
	Moreover for any~$\bx\in {\mathcal{U}} \setminus N$ and we can compute
	\begin{flalign*}
		&\lim\limits_{n \rightarrow \infty} \bigg| \Big(\frac{\alpha}{2\pi}\Big)^d \int_{\R^d} d\bxi \int_{\mathcal{U}} d\by \: e^{-i\alpha \bxi \cdot (\bx-\by)} \:\CA(\bx, \by, \bxi) \:\psi(\by)- \big(\Opa(\CA)\psi_n\big)(\bx) \bigg| \\
		&= \lim\limits_{n \rightarrow \infty}\bigg| \Big(\frac{\alpha}{2\pi}\Big)^d \int_{\R^d} d\bxi \int_{\mathcal{U}} d\by\:
		e^{-i\alpha \bxi \cdot (\bx-\by)}\: \CA(\bx, \by, \bxi) \;\Delta\psi_n(\by)\bigg| \\
		&\leq \Big(\frac{\alpha}{2\pi}\Big)^d \int_{\R^d}d\bxi \;\sqrt{ \int_{\R^d} d\by \: \big\| \CA(\bx,\by,\bxi) \big\|_{n \times n}^2   }\: \lim\limits_{n \rightarrow \infty}\|\Delta \psi_n\|^2_{L^2(\R^d,\C^n)} =0\:,
	\end{flalign*}
	with~$\Delta \psi_n := \psi - \psi_n$. Combining this estimate with the pointwise convergence~\eqref{eq:convOpPsi} in~${\mathcal{U}}\setminus N$ yields the claim.
\end{proof}
\begin{Remark} {\em{
			We want to apply Lemma~\ref{Simple Integral Rep.} to the operator~$\chi_\mathcal{K} (\Pi_-^\varepsilon)_{kn} \chi_\mathcal{K}$ with~$\mathcal{K} = (u_0-\rho,u_0)$. By rescaling as before, we see that this operator is of the form~$\Opa(\CA^{(\alpha)})$ with
			\begin{flalign*}
				\CA^{(\alpha)}(u,u',\xi)
				&= 2 \:\chi_\mathcal{K}(u)\: \chi_\mathcal{K}(u') \:\chi_{(-\infty,0)}(\xi) \:e^{ M\xi}\\
				& \quad\: \times \sum_{a,b=1}^{2} t_{ab}^{kn(\alpha\xi)}  \begin{pmatrix}
					X_{a,+}^{kn(\alpha\xi)}(u)\overline{X_{b,+}^{kn(\alpha\xi)}(u')} &  X_{a,+}^{kn(\alpha\xi)}(u)\overline{X_{b,-}^{kn(\alpha\xi)}(u')} \\
					X_{a,-}^{kn(\alpha\xi)}(u)\overline{X_{b,+}^{kn(\alpha\xi)}(u')} &  X_{a,-}^{kn(\alpha\xi)}(u)\overline{X_{b,-}^{kn(\alpha\xi)}(u')}
				\end{pmatrix}\:.
			\end{flalign*}
			Note that, a-priori, the integral representation of this operator is well-defined only for~$C^\infty_0(\R)^2$-functions with support in~$\mathcal{K}$. In order
			to extend the integral representation to all~$L^2(\mathcal{K}, \C^2)$-functions it suffices to verify the condition in Lemma~\ref{Simple Integral Rep.}. To this end we note that, due to Lemma~\ref{X_at_hor}, for given~$u_2>u_0$, we have
			\[ | X_{a,i}^{kn(\alpha\xi)}(u)| <  1+ ce^{du} \:, \]
			for any~$a\in \{1,2\}$, $i\in \{+,-\}$, $u<u_2$ and~$\xi \in \R$, where the constants~$c,d>0$ are independent of~$\xi$.
			Also using that the transmission coefficients~$t_{ab}$ are always bounded by~$1/2$, we obtain the estimate
			\[ \| \CA^{(\alpha)}(u,u`,\xi) \|_{n\times n} \leq C \:e^{ M\xi} \: \chi_{(-\infty,0)}(\xi)\: \chi_\mathcal{K}(u)\: \chi_\mathcal{K}(u') \]
			for any~$\alpha >0$ with~$C$ independent of~$u$, $u'$ and~$\xi$.
			Thus for any~$u \in \mathcal{K}$ and~$\alpha >0$ we have
			\begin{flalign*}
				\int d\xi \;\sqrt{\int du'\: \|\CA^{(\alpha)}(u,u',\xi)\|_{n\times n}^2} = C \int_{-\infty}^0 d\xi \;\sqrt{ \int_{u_0-\rho}^{u_0} du' \:e^{ 2M\xi}  } = C \frac{\sqrt{\rho}}{M} < \infty \:.
			\end{flalign*}
			Clearly, the entire expression vanishes for~$u \notin \mathcal{K}$. \\
			This shows that we can indeed apply Lemma~\ref{Simple Integral Rep.} to~$\chi_\mathcal{K} (\Pi_-^\varepsilon)_{kn} \chi_\mathcal{K}$, meaning that the corresponding integral representation can be applied to any $L^2(\mathcal{K},\C^2)$-function, and the $u'$- and $\xi$-integrations may be interchanged. Moreover, due to the characteristic functions in the symbol, we can even extend the integral representation to all 
			functions in~$L^2(\R,\C^2)$. \QEDrem }}
\end{Remark}

\begin{Lemma}
	\label{Opa Mult Rule}
	Let~$\CA(\bx,\by,\bxi)\equiv \CA(\bx,\bxi)$ (i.e.\ $\CA$ is independent of~$\by$) and~$\CB(\bx,\by,\bxi)\equiv \CB(\by,\bxi)$ be symbols such that~$\Opa(\CA)$ and~$\Opa(\CB)$ are well-defined and the following two conditions hold:
	\bitem
	\item[{\rm{(i)}}] The operator~$A$ defined for any~$\psi \in L^2(\R^d, \C^n)$ by
	\begin{flalign*}
		(A \psi) (\bx):= \int_{\R^n} e^{-i\bxi \bx} \:\CA(\bx,\bxi/\alpha) \:\psi(\bxi)\: d \bxi \:,
	\end{flalign*}
	is bounded on~$L^2(\R^d, \C^n)$.
	\item[{\rm{(ii)}}] The operator~$B$ defined for any~$\psi\in C^\infty_0(\R^d, \C^n)$ by
	\begin{flalign*}
		(B \psi) (\bxi):=  \frac{1}{(2\pi)^d}\int_{\R^n}  e^{i \bxi \by} \: \CB(\by,\bxi/\alpha)\: \psi(\by)\: d \by\:,
	\end{flalign*}
	may be continuously extended to~$L^2(\R^d, \C^n)$.
	\eitem
	Then
	\[ \Opa(\CA)\:\Opa(\CB) = \Opa(\CA \CB)\:. \]
\end{Lemma}
\begin{proof}
	We first note that, due to condition~{\rm{(i)}},
	\[ \Opa(\CA) = A\: \Four^{-1}\:, \]
	as both sides define continuous operators on~$L^2(\R^d, \C^n)$ and agree on the~$C^\infty_0(\R^d, \C^n)$-functions (where~$\Four$ is again the continuous extension of the Fourier transform to the Hilbert space~$L^2(\R^d, \C^n)$). Similarly, we conclude that
	\[ \Opa(\CB)= \Four \: B\:.\]
	This yields
	\[ \Opa(\CA)\:\Opa(\CB) = A\:B\:, \]
	and for any~$\psi \in C^\infty_0(\R^d, \C^n)$ we have
	\[ \Opa(\CA)\Opa(\CB)\psi = AB\psi = \int_{\R^d} d\bxi \int_{\R^n} d\bx \:e^{-i\alpha\bxi(\bx-\by)} \CA(\bx,\bxi) \CB(\by,\bxi) \psi(\by)\:. \]
	Note that as~$\Opa(\CA)$  and~$\Opa(\CB)$ are bounded operators, so is~$\Opa(\CA)\: \Opa(\CB)$.
	This concludes the proof by continuous extension and by the definition of~$\Opa(\CA \CB)$.
\end{proof}

\begin{Remark} \label{remark56}  $ $ {\em{ 
			\bitem
			\item [{\rm{(i)}}] In what follows we often apply the previous lemma in the case that
			\[ \CB(\by,\bxi)= \chi_U(\bxi) \]
			for some measurable set~$U\subseteq \R^d$. Then condition~(ii) of Lemma~\ref{Opa Mult Rule} is obviously fulfilled, because for any~$\psi\in C^\infty_0(\R^d, \C^n)$ we have
			\[
			B \psi  = \chi_{\alpha U}\: \Four^{-1} \psi\:,
			\]
			and thus
			\[ \| B \psi \| = \| \chi_{\alpha U}\: \Four^{-1} \psi \| \leq \| \Four^{-1} \psi \| = \|\psi\|\:. \]
			\item[{\rm{(ii)}}] Moreover, in the following, the symbol~$\CA$ is sometimes independent of both~$\bx$ and~$\by$ and bounded by a constant~$C>0$ (with respect to the matrix sup-norm), then for any~$\psi \in L^2(\R^d, \C^n)$ it follows that
			\[
			A \psi  = \CA(./\alpha)\: \Four \psi
			\]
			and moreover
			\[ \| A \psi \| = \|\CA(./\alpha)\: \Four \psi \| \leq C \| \psi \|\:. \]
			Therefore, condition~(i) in Lemma~\ref{Opa Mult Rule} is also fulfilled.
			\item[{\rm{(iii)}}] Another case we will consider later is that~$\CA\equiv a$ is scalar-valued, independent of~$y$ and continuous with compact support
			\[ \supp a \subseteq B_l(v) \times B_r(\mu)\:.\]
			Then from the following argument we conclude that~$A$ also fulfills condition~(i) from Lemma~\ref{Opa Mult Rule}. Take~$\psi \in L^2(\R^d)$ arbitrary and consider
			\begin{flalign*}
				&\int | A\psi(x)|^2\:dx \\
				&= \int_{B_l(v)} dx \int_{B_r(\mu)} d\xi \int_{B_r(\mu)} d\xi'\: e^{-iu(\xi'-\xi)}\: \overline{\psi(\xi)}\: \psi(\xi')\:\overline{a(x,\xi/\alpha)}\:a(x,\xi'/\alpha)
			\end{flalign*}
			Here we may interchange the order of integration due to the Fubini-Tonelli Theorem since
			\begin{flalign*}
				&\int_{B_l(v)} dx \int_{B_r(\mu)} d\xi \int_{B_r(\mu)} d\xi'\: \Big|  \overline{\psi(\xi)}\: \psi(\xi')\:\overline{a(x,\xi/\alpha)}\:a(x,\xi'/\alpha)\Big| \\
				& \leq C^2\: \text{vol}(B_l(v)) \:\|\psi\|_{L^1(B_r(\mu),\C^2)}^2 < \infty\:,
			\end{flalign*}
			where~$C$ is a bound for the absolute value of the continuous and compactly supported function~$a$. Note that~$L^2(B_r(\mu),\C^2) \subseteq L^1(B_r(\mu),\C^2)$ since~$B_r(\mu)$ is bounded.
			We then obtain
			\begin{flalign*}
				&\int | A\psi(u)|^2\:du\\
				&= \int d\xi \:\overline{\psi(\xi)} \int d\xi'\: \psi(\xi') \underbrace{\int dx \:e^{-ix(\xi'-\xi)} \: \overline{a(x,\xi/\alpha)}\:a(x,\xi'/\alpha)}_{=:\tilde{a}(\xi,\xi')}\\
				&\leq \|\psi\| \int |\psi(\xi)|\: \|\tilde{a}(\xi,.)\|\: d\xi
				\leq \| \psi\|^2 \sqrt{ \int d\xi \int d\xi' \:|\tilde{a}(\xi,\xi')|^2}\:,
			\end{flalign*}		
			where the function~$\tilde{a}$ is again continuous and compactly supported, which makes the last integral finite. We remark that in the last line we again applied H\"older's inequality twice.
			
			This estimate shows that condition~(i) from Lemma~\ref{Opa Mult Rule} is again satisfied. \QEDrem
			\eitem
	}}
\end{Remark}

\section{Proof of Lemma~\ref{X_at_hor}} \label{appA}
	\begin{proof}
		We follow the proof of \cite[Lemma 3.1]{tkerr}. As explained there, employing the ansatz
\beq \label{eqappA1}
			X(u) = \begin{pmatrix}
				e^{-i\omega u} f^+(u) \\
				e^{i\omega u} f^-(u)
			\end{pmatrix}\:,
\eeq
		the vector-valued function~$f$ must satisfy the ODE
		\begin{flalign}
			\label{f-Dgl}
			\frac{d}{du}f= \frac{\sqrt{\Delta(r)}}{r^2} \begin{pmatrix}
				0 & e^{2i\omega u} \:(imr-\lambda) \\
				e^{-2i\omega u} \: (-imr-\lambda) & 0
			\end{pmatrix}
			f\:,
		\end{flalign}
		where~$\lambda$ is an eigenvalue of the operator
		\begin{flalign*}
			\mathcal{A}:= \begin{pmatrix}
				0 & \frac{d}{d \vartheta} + \frac{\cot \vartheta}{2}+ \frac{k+1/2}{\sin \vartheta} \\
				-\frac{d}{d \vartheta} - \frac{\cot \vartheta}{2}+ \frac{k+1/2}{\sin \vartheta} & 0
			\end{pmatrix}
		\end{flalign*}
		(see \cite[Appendix~1]{tkerr}) and thus does not depend on~$\omega$ (in contrast with the Kerr-Newman case as explained in \cite{tkerr}). Estimating~\eqref{f-Dgl} gives
		\[ \left|\frac{d}{du} f\right| \leq \left|\frac{\sqrt{r-2M}}{r^{3/2}}\right| (mr+ |\lambda|) |f|=\sqrt{\frac{r-2M}{r}} \Big(m+ \frac{|\lambda|}{r}\Big) |f| \:. \]
		Next, we transform~$r-2M$ to the Regge-Wheeler-coordinate,
		\[ r-2M = 2M\: W\big(e^{u/(2M)-1}/2M\big)\:, \]
		where~$W$ is the inverse log function, i.e.\ the inverse function of~$x \mapsto xe^x$.
		An elementary estimate\footnote{Since the function~$f(x):=xe^x$ is strictly increasing (and differentiable) on~$(0,\infty)$, so is~$W=f^{-1}$ on~$\big(f(0),f(\infty)\big)=(0,\infty)$. So from~$xe^x \geq x$ for any~$x\geq 0$ follows~$x=W(xe^x)\geq W(x)$.} shows that~$W(x)\leq x$ for any~$x\geq 0$ and therefore we can  estimate
		\begin{flalign}
			\label{f' est.}
			\left|\frac{d}{du} f\right| \leq\frac{e^{u/M-1/2}}{\sqrt{2M}} \Big( m+ \frac{|\lambda|}{2M} \Big) |f| \:,
		\end{flalign}
		where we also used that~$r\geq2M$. Setting
		\[ c_1:= (2Me)^{-1/2} \Big(m+  \frac{|\lambda|}{2M}\Big)\:, \quad d:=\frac{1}{M}\:, \]
		we can proceed just as in \cite[Proof of 3.1]{tkerr}:
		
		Without loss of generality we can assume that~$|f|$ is nowhere vanishing\footnote{If~$|f(u)|=0$ for one~$u \in \R$, then due to~\eqref{f-Dgl} any order of derivative of~$f$ vanishes at~$u$ a and therefore, by the Picard-Lindel\"of theorem, the function~$f$ vanishes identically on~$\R$.} and divide~\eqref{f' est.} by~$|f|$ giving
		\begin{flalign*}
			\frac{|d/du f|}{|f|} \leq c_1 e^{du}\:.
		\end{flalign*}
		This yields
		\begin{flalign*}
			\big| \log |f| (u_2) - \log |f| (u) \big| &=  \left| \int^{u_2}_u  \frac{\frac{d}{du}(|f^+|^2+|f^-|^2)}{|f|^2}du' \right| \leq 4 \int^{u_2}_u c_1 \: e^{du'}\: du' \\
			&= \frac{4 c_1}{d} \big(e^{du_2}-e^{du}\big)\:.
		\end{flalign*}
		From this we conclude that
		\begin{flalign*}
			\log |f| (u) \geq \log |f| (u_2) - \frac{4 c_1}{d} \big(e^{du_2}-e^{du}\big) \geq \log |f| (u_2) - \frac{4 c_1}{d} e^{du_2} \\
			\log |f| (u) \leq \log |f| (u_2) + \frac{4 c_1}{d} \big(e^{du_2}-e^{du}\big) \leq \log |f| (u_2) + \frac{4 c_1}{d} e^{du_2} \:,
		\end{flalign*}
		which yields
		\begin{flalign}
			\label{f_bounds}
			|f(u_2)|\: \exp\Big(- \frac{4 c_1}{d} e^{du_2}\Big) \leq |f(u)| \leq |f(u_2)| \: \exp\Big(\frac{4 c_1}{d} e^{du_2}\Big)\:. 
		\end{flalign}
Using this inequality in~\eqref{f' est.}, we obtain
		\begin{flalign}
			\label{f' est. 2}
			\left|\frac{d}{du} f\right| \leq c_1\: |f(u_2)|\: \exp\Big(\frac{4 c_1}{d} e^{du_2}\Big) \: e^{du}\:,
		\end{flalign}
		which shows that~$\frac{df}{du}$ is integrable. Moreover, the function~$f(u)$ converges for~$u\rightarrow -\infty$ to
		\begin{flalign*}
			f_0:= \lim\limits_{u \rightarrow -\infty}f(u) \overset{\eqref{f_bounds}}{\neq} 0\:.
		\end{flalign*}
		Now integrating~\eqref{f' est. 2} from~$-\infty$ to~$u<u_2$, we get
		\begin{flalign}
			\label{f aufint.}
			|f(u)-f_0| \leq \frac{c_1}{d}\:  \big|f(u_2) \big| \: \exp\Big(\frac{4 c_1}{d} e^{du_2}\Big) \: e^{du}\:.
		\end{flalign}
		Finally, in order to get rid of the factor~$|f(u_2)|$, we make use of~\eqref{f_bounds} in the limit~$u\rightarrow -\infty$,
		\begin{flalign}
			\label{f0est}
			|f(u_2)| \leq |f_0|\: \exp\Big(\frac{4 c_1}{d} e^{du_2}\Big) \:.
		\end{flalign}
		Substituting this in~\eqref{f aufint.}, we end up with the desired result
\[ 
			|f(u)-f_0| \leq c  e^{du} \:, \]
		with
		\[ c:= \frac{c_1}{d}\:  |f_0| \: \exp\Big(\frac{8 c_1}{d} e^{du_2}\Big) \]
		and
		\[ g(u):= f(u)-f_0 \:. \]
		Similarly, removing~$|f(u_2)|$ from~\eqref{f' est. 2} using~\eqref{f0est} we obtain
		\begin{flalign*}
			\left| \frac{d}{du}  g \right| \leq dc e^{du} \:,
		\end{flalign*}
		which completes the proof.
	\end{proof}

\section{Computing the Symbol of~$(\Pi_-^\varepsilon)_{kn}$} \label{appB}
In this section, we give a more detailed computation of the symbol of the operator~$(\Pi_-^\varepsilon)_{kn}$ for given~$k$ and~$n$. Recall that~$(\Pi_-^\varepsilon)_{kn}$ is for any function~$\psi \in C_0^\infty(\R,\C^2)$ given by,
\[  \big((\Pi_-^\varepsilon)_{kn}\psi \big)(u)  = \frac{1}{\pi} \int d \omega \int du' e^{-\varepsilon \omega} \sum_{a,b=1}^2 t_{a,b}^\omega X_a(u,\omega) \big\la X_b(.,\omega) \:\big|\: \psi \big\ra\:.  \]
The main task is therefore to determine
\beq \label{XProd Sum} \sum_{a,b=1}^2 t_{a,b}^\omega X_a(u,\omega)  X_b(u',\omega)^\dagger =:(*)\:.  \eeq
To this end first note that the details of the coefficients~$t_{ab}$ in~\eqref{tab} give
\begin{flalign}
	\label{tabEv 1}
	(*)&= \chi_{(-m,0)}(\omega) \:X_1(u,\omega)\:X_1(u',\omega)^\dagger \\
	&+ \chi_{(-\infty,-m)}(\omega)\:\Big[  \: \frac{1}{2} \:X_1(u,\omega)\:X_1(u',\omega)^\dagger + \frac{1}{2}\:X_2(u,\omega)\:X_2(u',\omega)^\dagger  \\
	&+ t_{12}^\omega\: X_1(u,\omega)\:X_2(u',\omega)^\dagger + t_{21}^\omega\: X_2(u,\omega)\:X_1(u',\omega)^\dagger \: \Big]\:. \label{tabEv 3}
\end{flalign}
Moreover, using the asymptotics of the radial solutions given in Lemma~\ref{X_at_hor} the matrix~$X_i(u,\omega)\:X_j(u',\omega)^\dagger$ can for any~$i,j \in \{1,2\}$ be written as
\begin{flalign}
&X_i(u,\omega)\:X_j(u',\omega)^\dagger \notag \\
&= \begin{pmatrix} \label{XProd 1}
	f_{0,i}^+(\omega)\:\overline{f_{0,j}^+(\omega)}\: e^{-i\omega(u-u')} & f_{0,i}^+(\omega)\:\overline{f_{0,j}^-(\omega)} \:e^{-i\omega(u+u')} \\
	f_{0,i}^-(\omega)\overline{f_{0,j}^+(\omega)} \:e^{i\omega(u+u')} & f_{0,i}^-(\omega)\:\overline{f_{0,j}^-(\omega)} \:e^{i\omega(u-u')}
\end{pmatrix} \\
&+R_{0,i}(u,\omega)\begin{pmatrix}\: \label{XProd 2}
f_{0,j}^+(\omega)\: e^{-i\omega u'} \\ f_{0,j}^-(\omega)\: e^{i\omega u'}
\end{pmatrix}^\dagger + \begin{pmatrix}
f_{0,i}^+(\omega)\: e^{-i\omega u} \\ f_{0,i}^-(\omega)\: e^{i\omega u}
\end{pmatrix}  R_{0,j}(u',\omega)^\dagger \\
&+ R_{0,i}(u,\omega)\:R_{0,j}(u',\omega)^\dagger \:. \label{XProd 3}
\end{flalign}
The terms in~\eqref{XProd 2}-\eqref{XProd 3} will result in the error matrix~$\mathcal{R}_0$ and are computed in Section~\ref{Fast decaying terms}. Here we are mainly interested in the terms in~\eqref{XProd 1}.

Combining our choices of~$f_0$ from Section~\ref{rnProp at the Horizon} with~\eqref{XProd 1} and~\eqref{tabEv 1}-\eqref{tabEv 3}, we obtain
\begin{flalign*}
	(*)&= \chi_{(-m,0)}(\omega) \begin{pmatrix} 
			|f_{0,1}^+(\omega)|^2\: e^{-i\omega(u-u')} & f_{0,1}^+(\omega)\:\overline{f_{0,1}^-(\omega)} \:e^{-i\omega(u+u')} \\
			f_{0,1}^-(\omega)\overline{f_{0,1}^+(\omega)} \:e^{i\omega(u+u')} & |f_{0,1}^-(\omega)|^2 \:e^{i\omega(u-u')}
		\end{pmatrix} \\
	&+\chi_{(-\infty,-m)}(\omega) \begin{pmatrix}
		\frac{1}{2} \: e^{-i\omega(u-u')} & t_{12}^\omega \: e^{-i\omega(u+u')} \\
		t_{21}^\omega \: e^{i\omega(u+u')}  & \frac{1}{2} \: e^{i\omega(u-u')} 
	\end{pmatrix}
+ \tilde{\mathcal{R}}_0(u,u',\omega)\:,
\end{flalign*}
where~$\tilde{\mathcal{R}}_0(u,u',\omega)$ consists of the terms~\eqref{XProd 2}-\eqref{XProd 3} inserted in the sum~\eqref{XProd Sum}.

In order to rewrite~$(\Pi_-^\varepsilon)_{kn}$ as a pseudo-differential operator, we need a prefactor of the form~$e^{-i\omega(u-u')}$ before the symbol. The matrix components in~$(*)$ indeed involve such
plane waves. However, the $(2,2)$-components oscillate with the wrong sign. In order to circumvent this issue, we can use the freedom of coordinate change~$\omega \rightarrow - \omega$ in the~$d\omega$ integration of the $(2,2)$- and $(1,2)$-components. This yields~\eqref{Pi kernel}.

\section{Regularity of~$\eta_\kappa$}
We now verify in detail that the functions~$\eta_\kappa$ satisfy Condition~\ref{cond:f3}.

\begin{Lemma}
	\label{Properties eta}
	Consider the functions~$\eta_\kappa$ in~\eqref{Def eta_gamma}. Then for any~$\kappa \neq 1$, $\eta_\kappa$ satisfies Condition~\ref{cond:f3} with~$T=\{0,1\}$ for any~$\gamma \leq \min\{1,\kappa\}$.
	Moreover, the function~$\eta_\kappa$ in~\eqref{Def eta} satisfies Condition~\ref{cond:f3} with~$T=\{0,1\}$ for any~$\gamma=\kappa$.

\end{Lemma}
\begin{proof}
	We start with the case that~$\kappa = 1$. Then,
	\[ \eta_1 \in C^2(\R \setminus \{0,1\})\cap C^0(\R) \:, \]
	as
	\[ \lim\limits_{t \rightarrow 0 }\big( -t\log(t)-(1-t)\log(1-t) \big)  = - \lim\limits_{t \rightarrow 0 }\frac{\log(t)}{t^{-1}} \overset{l'H}{=} \lim\limits_{t \rightarrow 0 }\frac{t^{-1}}{t^{-2}} = 0\:, \]
	(where  ``$l'H$" denotes the use of l'Hospital's rule) and
	\[ \lim\limits_{t \rightarrow 1 }\big( -t\log(t)-(1-t)\log(1-t) \big)  = - \lim\limits_{t \rightarrow 1 }\frac{\log(1-t)}{(1-t)^{-1}} \overset{l'H}{=} - \lim\limits_{t \rightarrow 1 }\frac{(1-t)^{-1}}{(1-t)^{-2}} = 0\:. \]
	Moreover, for any~$t\in (0,1)$ we have
	\begin{flalign*}
		\eta_1'(t)&= -\log(t)+\log(1-t)\:, \\
		\eta_1''(t)&=-\frac{1}{t}-\frac{1}{1-t}\:.
	\end{flalign*}
	Thus, for any~$\gamma<1$
	\begin{flalign*}
		\lim\limits_{t \searrow 0} \eta_1(t)t^{-\gamma} =& -\lim\limits_{t \searrow 0} \frac{\log t}{t^{\gamma-1}} - \lim\limits_{t \searrow 0}\frac{\log (1-t)}{t^\gamma} \overset{l'H}{=} -\lim\limits_{t \searrow 0} \frac{t^{-1}}{(\gamma-1) t^{\gamma-2}}+ \lim\limits_{t \searrow 0} \frac{(1-t)^{-1}}{\gamma t^{\gamma-1}} \\
		=& \lim\limits_{t \searrow 0} \frac{t^{1-\gamma}}{1-\gamma}  + \lim\limits_{t \searrow 0} \frac{t^{1-\gamma}}{\gamma(1-t)}=0 \:,
	\end{flalign*}
	and obviously
	\[\lim\limits_{t \nearrow 0} \eta_1(t)t^{-\gamma} = 0 \:. \]
	Therefore, there exists a neighborhood~$U_{0,0}$ of~$z=0$ and a constant~$C_{0,0}$ such that for any~$t \in U_{0,0}$,
	\[ |\eta_1(t)|\leq C_{0,0} |t|^\gamma\:. \]
	Similarly,
	\begin{flalign*}
		\lim\limits_{t \nearrow 1} \eta_1(t) (t-1)^{-\gamma} &=  \lim\limits_{t \nearrow 1} \frac{\log t}{(1-t)^{\gamma}} + \lim\limits_{t \nearrow 1}\frac{\log(1-t)}{(1-t)^{\gamma-1}}\\
		&\!\overset{l'H}{=}\, \lim\limits_{t \nearrow 1} -\frac{t^{-1}}{\gamma (1-t)^{\gamma-1}} + \lim\limits_{t \nearrow 1} \frac{(1-t)^{-1}}{(\gamma-1)(1-t)^{\gamma-2}} = 0\:, \\
		\lim\limits_{t \searrow 1} \eta_1(t) (t-1)^{-\gamma} &=0\:,
	\end{flalign*}
	yielding a neighborhood~$U_{1,0}$ of~$z=1$ and a constant~$C_{1,0}$ such that for any~$t \in U_{1,0}$:
	\[ |\eta_1(t)| \leq C_{1,0}\: |t-1|^\gamma \:. \]
	
	The other estimates follow  analogously by computing the limits
	\begin{flalign*}
		\lim\limits_{t \searrow 0} \eta_1'(t) t^{1-\gamma} =& - \lim\limits_{t \searrow 0} \frac{\log(t)}{t^{\gamma-1}} \overset{l'H}{=} \lim\limits_{t \searrow 0}\frac{t^{-1}}{(1-\gamma) t^{\gamma-2}} = \lim\limits_{t \searrow 0 } \frac{t^{1-\gamma}}{1-\gamma} =0 \:,\\
		\lim\limits_{t \nearrow 1} \eta_1'(t) (t-1)^{1-\gamma} =& -\lim\limits_{t \nearrow 1} \frac{\log(1-t)}{(1-t)^{\gamma-1}} \overset{l'H}{=} \lim\limits_{t \nearrow 1} \frac{(1-t)^{-1}}{(1-\gamma)(1-t)^{\gamma-2}} \\
		=& \lim\limits_{t \nearrow 1} \frac{(1-t)^{1-\gamma}}{1-\gamma} =0\:, \\
		\lim\limits_{t \searrow 0} \eta_1''(t)t^{2-\gamma} =&- \lim\limits_{t \searrow 0} t^{1-\gamma} = 0 \:,\\
		\lim\limits_{t \nearrow 1} \eta_1''(t)(t-1)^{2-\gamma} =& -\lim\limits_{t \nearrow 1} (1-t)^{1-\gamma} = 0\:, \\
		\lim\limits_{t \nearrow 0} \eta_1'(t)t^{1-\gamma} =& \lim\limits_{t \searrow 1} \eta_1'(t)(t-1)^{1-\gamma} = \lim\limits_{t \nearrow 0} \eta_1''(t)t^{2-\gamma} \\
		=& \lim\limits_{t \searrow 1} \eta_1''(t)(t-1)^{2-\gamma} =0 \:.
	\end{flalign*}
	This concludes the proof for the case that~$\kappa=1$.
	
	Next, consider~$\kappa \neq 1$. It is evident that
	\[ \eta_\kappa \in C^2(\R \setminus \{0,1\})\cap C^0(\R) \:. \]
	Moreover, the derivatives of~$\eta_\kappa$ for~$t\in(0,1)$ are given by
	\begin{align*}
		&\eta_\kappa'(t) = \frac{\kappa}{1-\kappa} \frac{t^{\kappa-1} - (1-t)^{\kappa-1}}{t^\kappa +(1-t)^\kappa} \:,\\
		&\eta_\kappa''(t) = \kappa \frac{t^{\kappa-2}+(1-t)^{\kappa-2}}{t^{\kappa}+(1-t)^{\kappa}} - \frac{\kappa^2}{1-\kappa} \frac{\big(t^{\kappa-1}+(1-t)^{\kappa-1}\big)^2}{t^{\kappa}+(1-t)^{\kappa}}\:.
	\end{align*}
	Thus we conclude that for~$\kappa <1$:
	\begin{align*}
		&\eta_\kappa'(t) \simeq t^{\kappa-1}\:, \qquad \eta_\kappa''(t) \simeq t^{\kappa-2} \:, \qquad \text{for } t\searrow 0\:, \\
		&\eta_\kappa'(t) \simeq (1-t)^{\kappa-1}\:, \qquad \eta_\kappa''(t) \simeq (1-t)^{\kappa-2} \:, \qquad \text{for } t\nearrow 1\:,
	\end{align*}	
	so that we may choose~$\gamma \leq \kappa$. For~$1<\kappa\leq 2$ first note that none of the first derivatives vanishes at~$t=0$ or~$t=1$, so we have
	\begin{align*}
		&\eta_\kappa'(t) \simeq 1 \:, \qquad \eta_\kappa''(t) \simeq t^{\kappa-1} \:, \qquad \text{for } t\searrow 0\:, \\
		&\eta_\kappa'(t) \simeq 1\:, \qquad \eta_\kappa''(t) \simeq (1-t)^{\kappa-1} \:, \qquad \text{for } t\nearrow 1\:,
	\end{align*}	
	and therefore we can only take~$\gamma \leq 1$. Similarly for~$\kappa >2$ we have
	\begin{align*}
		&\eta_\kappa'(t) \simeq 1 \:, \qquad \eta_\kappa''(t) \simeq 1 \:, \qquad \text{for } t\searrow 0\:, \\
		&\eta_\kappa'(t) \simeq 1\:, \qquad \eta_\kappa''(t) \simeq 1\:, \qquad \text{for } t\nearrow 1\:,
	\end{align*}
	so we can only take~$\gamma \leq 1$ as well.
\end{proof}

\Thanks{\em{Acknowledgments:}}
We are grateful to Erik Curiel, Jos{\'e} M.\ Isidro, Claudio F.\ Paganini, Alexander V.\ Sobolev and
Wolfgang Spitzer for helpful discussions.
We would like to thank the referees for valuable feedback.
M.L.\ gratefully acknowledges support by the Studienstiftung des deutschen Volkes and the Marianne-Plehn-Programm.

\providecommand{\bysame}{\leavevmode\hbox to3em{\hrulefill}\thinspace}
\providecommand{\MR}{\relax\ifhmode\unskip\space\fi MR }
\providecommand{\MRhref}[2]{%
  \href{http://www.ams.org/mathscinet-getitem?mr=#1}{#2}
}
\providecommand{\href}[2]{#2}

\end{document}